\documentclass[runningheads]{llncs}

\usepackage{graphicx}
\usepackage[utf8]{inputenc}
\usepackage[T1]{fontenc}
\usepackage{subfigure}
\usepackage{amsmath,amssymb, mathtools, hyperref}
\usepackage{graphicx}
\usepackage{cleveref}
\usepackage{tikz}
\usetikzlibrary{patterns,shapes}
\usepackage{gastex}
\usepackage{todonotes}
\usepackage{comment}

\def\Q{\mathbb Q}
\def\N{\mathbb N}
\def\A{\mathcal A}
\def\B{\mathcal B}
\def\C{\mathcal C}

\def\E{\mathcal E}
\def\LL{\mathcal L}

\def\R{\mathcal R}
\def\S{\mathcal S}
\def\W{\mathcal W}
\def\uu{\mathbf u}
\def\vv{\mathbf v}
\def\xx{\mathbf x}
\def\yy{\mathbf y}

\def\dd{\mathbf d}
\def\ff{\mathbf f}
\def\gg{\mathbf g}
\def\bs{\beta}
\def\Parickh{\vec{\Psi}}
\newcommand{\gap}[2]{\textrm{gap}(#1,#2)}
\def\ind{\text{ind}}
\def\CR{E}
\def\lcm{\text{lcm}}
\def\barva{{\rm colour}}
\newtheorem{thm}{Theorem}

\newtheorem{Observation}[thm]{Observation}

\usepackage[pagewise]{lineno}

\begin{document}
	
	\title{On balanced sequences and their critical exponent\thanks{The research received funding from  the Ministry of Education, Youth and Sports of the Czech Republic through the project  CZ.02.1.01/0.0/0.0/16\_019/0000765 and CZ.02.1.01/0.0/0.0/16\_019/0000778.}}
	
	\author{Francesco Dolce\inst{1} \and L\!'ubom\'{i}ra Dvo\v{r}\'{a}kov\'{a}\inst{2} \and Edita Pelantov\'{a}\inst{2}}
	
	\authorrunning{F. Dolce et al.}
	
	\institute{FIT, Czech Technical University in Prague, Czech Republic \and FNSPE, Czech Technical University in Prague, Czech Republic \\ \email{dolce.fra@fit.cvut.cz, \\ \{lubomira.dvorakova, edita.pelantova\}@fjfi.cvut.cz}}
	
	\maketitle

	\begin{abstract}
		We study aperiodic balanced sequences over finite alphabets.
		A sequence $\vv$ of this type is fully characterised by a Sturmian sequence $\uu$ and two constant gap sequences $\yy$ and $\yy'$.
		We show that the language of $\vv$ is eventually dendric and we focus on return words to its factors.
		We develop a method for computing the critical exponent and asymptotic critical exponent of balanced sequences, provided the associated Sturmian sequence $\uu$ has a quadratic slope.
		The method is based on looking for the shortest return words to bispecial factors in $\vv$.
		We illustrate our method on several examples; in particular we confirm a conjecture of Rampersad, Shallit and Vandomme that two specific sequences have the least critical exponent among all balanced sequences over $9$-letter (resp., $10$-letter) alphabets.
	\end{abstract}
	

	\section{Introduction}
	An infinite sequence over a finite alphabet is balanced if, for any two of its factors $u$ and $v$ of the same length, the number of occurrences of each letter in $u$ and $v$ differs by at most 1.
	Over a binary alphabet aperiodic balanced sequences coincide with Sturmian sequences, as shown by Hedlund and Morse~\cite{HeMo}.
	Hubert~\cite{Hu} provided a construction of balanced sequences on a $d$-letter alphabet (see also~\cite{Gr}).
	It consists in colouring the letters of a Sturmian sequence $\uu$ by two constant gap sequences $\yy$ and $\yy'$.
	In this paper we study combinatorial properties of balanced sequences.
	We first show that such sequences belong to the class of eventually dendric sequences introduced by Berthé et al.~\cite{BertheDeFeliceDolceLeroyPerrinReutenauerRindone2015}.
	We give formul\ae \ for the factor complexity and the number of return words to each factor.
	The main goal of this paper is to develop a method for computing the critical exponent and asymptotic critical exponent of a given balanced sequence.
	To help achieve this goal we deduce, in Section~\ref{Section_CriticalExponent}, new formul\ae \ expressing the critical exponent and asymptotic critical exponent of a general uniformly recurrent sequence.
	They exploit notions of bispecial factor and return word.   
	
	Our work can be understood as a continuation of research on balanced sequences with the least critical exponent initiated by Rampersad, Shallit and Vandomme~\cite{RaShVa}.
	
	Finding the best lower bound on the critical exponent of sequences
	over an alphabet of size $d$	is a classical problem.
	The answer is the well-known Dejean’s conjecture~\cite{Dejean}, that despite the name is not a conjecture anymore, since it was proved step by step by several people.
	The least critical exponent was determined also for some particular classes of sequences: by Carpi and de Luca~\cite{Cade} for Sturmian sequences, and by Currie, Mol and Rampersad~\cite{CuMoRa} for binary rich sequences.
	
	Recently, Rampersad, Shallit and Vandomme~\cite{RaShVa} found balanced sequences with the least critical exponent over alphabets of size 3 and 4 and also conjectured that the least critical exponent of balanced sequences over a $d$-letter alphabet with $d \geq 5$ is $\frac{d-2}{d-3}$.
	Their conjecture was confirmed for $d \leq 8$ by Baranwal and Shallit~\cite{BaranThesis,BaShBalanced}.
	
	Here we first focus on the asymptotic critical exponent, which reflects repetitions of factors of length growing to infinity.
	We show that the asymptotic critical exponent depends on the slope of the associated Sturmian sequence and, unlike the critical exponent, on the length of the minimal periods of $\yy$ and $\yy'$, but not on $\yy$ and $\yy'$ themselves.
	We also give a general lower bound on the asymptotic critical exponent.
	We provide an algorithm computing the exact value of the asymptotic critical exponent for balanced sequences originating from Sturmian sequences with a quadratic slope.
	The algorithm ignores the behaviour of short bispecial factors.
	
	Secondly, we refine our approach to all bispecial factors, not only the sufficiently long ones.
	It enables us to extend our algorithm and compute the critical exponent of the balanced sequences as well.
	For $d=9$ and $d=10$ we confirm, using our algorithm, the conjecture that the least critical exponent of balanced sequences over a $d$-letter alphabet is $\frac{d-2}{d-3}$.
	However, in the course of the referee process of this paper, the conjecture was disproved by introducing $d$-ary balanced sequences with the critical exponent equal to $\frac{d-1}{d-2}$ for $d=11$ and also for all even $d$'s larger than 10 (see~\cite{DvOpPeSh2022}). 
	
	Some of the results in this paper were first presented, in a less general form and without detailed proofs, in two conference papers~\cite{LATA,WORDS}.
	For the sake of self-consistency we decided to add all proofs here, even the ones that already appeared in the two previous contributions.
	
	The algorithms computing asymptotic critical exponent and critical exponent of balanced sequences were implemented by our student Daniela Opo\v censk\'a.
	We are very grateful for her careful, readily usable and user-friendly implementation, which was extremely helpful for us.
	We also want to thank the anonymous referee for their comments and suggestions that helped us improving the presentation of this paper.

	\section{Preliminaries}
	\label{Section_Preliminaries}
	
	An \textit{alphabet} $\A$ is a finite set of symbols called \textit{letters}.
	A (finite) \textit{word} over $\A$ of \textit{length} $n$ is a string $u = u_0 u_1 \cdots u_{n-1}$, where $u_i \in \A$ for all $i \in \{0,1, \ldots, n-1\}$.
	The length of $u$ is denoted by $|u|$.
	If $u_0 u_i \cdots u_{n-1} = u_{n-1} u_{n-2} \cdots u_0$, we call the word $u$ a \emph{palindrome}.
	The set of all finite words over $\A$ together with the operation of concatenation forms a monoid, denoted $\A^*$.
	Its neutral element is the \textit{empty word} $\varepsilon$ and we write $\A^+ = \A^* \setminus \{\varepsilon\}$.
	
	If $u = xyz$ for some $x,y,z \in \A^*$, then $x$ is a \textit{prefix} of $u$, $z$ is a \textit{suffix} of $u$ and $y$ is a \textit{factor} of $u$.
	We sometimes use the notation $yz = x^{-1}u$.
	To every word $u$ over $\A$ with cardinality $\#\A = d$,  we assign its \textit{Parikh vector} $\Parickh(u) \in \N^{d}$ defined as $(\Parickh(u))_a = |u|_a$ for all $a \in \A$, where $|u|_a$ is the number of letters $a$ occurring in $u$.
	
	A \textit{sequence} over $\A$ is an infinite string $\uu = u_0 u_1 u_2 \cdots$, where $u_i \in \A$ for all $i \in \N$.
	In this paper we always denote sequences by bold letters.
	A sequence $\uu$ is \textit{eventually periodic} if $\uu = vwww \cdots = v w^\omega$ for some $v \in \A^*$ and $w \in \A^+$.
	It is periodic if $\uu = w^{\omega}$.
	If $\uu$ is not eventually periodic, then it is \textit{aperiodic}.
	A \textit{factor} of $\uu = u_0 u_1 u_2 \cdots$ is a word $y$ such that $y = u_i u_{i+1} u_{i+2} \cdots u_{j-1}$ for some $i, j \in \N$, $i \leq j$.
	The number $i$ is called an \textit{occurrence} of the factor $y$ in $\uu$.
	In particular, if $i = j$, the factor $y$ is the empty word $\varepsilon$ and every index $i$ is its occurrence.
	If $i=0$, the factor $y$ is a \textit{prefix} of $\uu$.
	If each factor of $\uu$ has infinitely many occurrences in $\uu$, the sequence $\uu$ is \textit{recurrent}.
	Moreover, if for each factor the distances between its consecutive occurrences are bounded, $\uu$ is \textit{uniformly recurrent}.
	
	The \textit{language} $\LL(\uu)$ of a sequence $\uu$ is the set of all its factors.
	A factor $w$ of $\uu$ is \textit{right special} if $wa, wb$ are in $\LL(\uu)$ for at least two distinct letters $a,b \in \A$.
	Analogously, we define a \textit{left special} factor.
	A factor is \textit{bispecial} if it is both left and right special.
	Note that the empty word $\varepsilon$ is bispecial if at least two distinct letters occur in $\uu$.
	The \textit{factor complexity} of a sequence $\uu$ is the mapping $\C_\uu: \N \to \N$ defined by
	$\C_\uu(n) = \# \{w \in \LL(\uu) : |w| =  n \}$.
	The first difference of the factor complexity is defined as $s_{\uu}(n) = \C_{\uu}(n+1) - \C_{\uu}(n)$.
	
	Given a word $w \in \LL(\uu)$, we define the sets of left extensions, right extensions and bi-extensions of $w$ in $\LL(\uu)$ respectively as
	$$
	L_{\uu}(w) = \{ a \in \A : aw \in \LL(\uu) \},
	\qquad
	R_{\uu}(w) = \{ b \in \A : wb \in \LL(\uu) \}
	$$
	and
	$$
	B_{\uu}(w) = \{ (a,b) \in \A \times \A : awb \in \LL(\uu) \}. $$
	The \emph{extension graph} of $w$ in $\LL(\uu)$, denoted $\E_{\uu}(w)$, is the undirected bipartite graph whose set of vertices is the disjoint union of $L_{\uu}(w)$ and $R_{\uu}(w)$ and whose edges are the elements of $B_{\uu}(w)$.
	A sequence $\uu$ (resp., a language $\LL(\uu)$) is said to be \emph{eventually dendric} with \emph{threshold} $m \ge 0$ if $\E_\uu(w)$ is a tree for every word $w \in \LL(\uu)$ of length at least $m$.
	It is said to be \emph{dendric} if we can choose $m = 0$.
	Dendric languages were introduced by Berthé et al.~\cite{BertheDeFeliceDolceLeroyPerrinReutenauerRindone2015} under the name of tree sets.
	It is known that Sturmian sequences are dendric.
	
	\begin{example}
		\label{ex:Fibo_extGraph}
		Let $\uu$ be a sequence such that its only factors of length at most three are $\varepsilon$, ${\tt a}$, ${\tt b}$, ${\tt aa}$, ${\tt ab}$, ${\tt ba}$, ${\tt aab}$, ${\tt aba}$, ${\tt baa}$, ${\tt bab}$ (an explicit instance of such a sequence will be given in Example~\ref{ex:FiboDef}).
		The extension graphs $\E_{\uu}(\varepsilon)$, $\E_{\uu}(\tt a)$ and $\E_{\uu}(\tt b)$ are shown in Figure~\ref{fig:FiboEG}.
	\end{example}
	
	\begin{figure}[hbt]
		\centering
		\tikzset{node/.style={rectangle,draw,rounded corners=1.2ex}}
		\begin{tikzpicture}
			\node[node](eal) {$\tt a$};
			\node[node](ebl) [below= 0.1cm of eal] {$\tt b$};
			\node[node](ear) [right= 1cm of eal] {$\tt a$};
			\node[node](ebr) [below= 0.1cm of ear] {$\tt b$};
			\path[draw,thick, shorten <=0 -1pt, shorten >=-1pt]
			(eal) edge node {} (ear)
			(eal) edge node {} (ebr)
			(ebl) edge node {} (ear);
			\node[node](aal) [right= 2cm of ear] {$\tt a$};
			\node[node](abl) [below= 0.1cm of aal] {$\tt b$};
			\node[node](aar) [right= 1cm of aal] {$\tt a$};
			\node[node](abr) [below= 0.1cm of aar] {$\tt b$};
			\path[draw,thick]
			(aal) edge node {} (abr)
			(abl) edge node {} (aar)
			(abl) edge node {} (abr);
			\node[node](bal) [below right = 0cm and 2cm of aar] {$\tt a$};
			\node[node](bar) [right= 1cm of bal] {$\tt a$};
			\path[draw,thick]
			(bal) edge node {} (bar);
		\end{tikzpicture}
		\caption{The graphs $\E_{\uu}(\varepsilon)$ (on the left), $\E_{\uu}(\tt a)$ (in the centre) and $\E_{\uu}(\tt b)$ (on the right).}
		\label{fig:FiboEG}
	\end{figure}
	
	Aperiodic sequences with the lowest possible factor complexity, i.e., such that $\C_{\uu}(n) = n+1$ for all $n \in \N$, are called \textit{Sturmian sequences} (for other equivalent definitions see~\cite{BaPeSt}).
	Clearly, all Sturmian sequences are defined over a binary alphabet, e.g., $\{ {\tt a,b} \}$.
	Moreover, they are such that $s_{\uu}(n) = 1$ for every $n \in \N$.
	If both sequences ${\tt a}\uu$ and ${\tt b}\uu$ are Sturmian, then $\uu$ is called a \textit{standard Sturmian sequence}.
	In other words, for a standard Sturmian sequence $\uu$, the left special factors are exactly the prefixes of $\uu$.
	Moreover, bispecial factors correspond to palindromic prefixes of $\uu$.
	It is well-known that for every Sturmian sequence there exists a unique standard Sturmian sequence with the same language.
	
	A~sequence $\uu$ over the alphabet $\A$ is \textit{balanced} if for every letter $a \in \A$ and every pair of factors $u,v \in \LL(\uu)$ with $|u| = |v|$, we have $|u|_a - |v|_a \leq 1$.
	The class of Sturmian sequences and the class of aperiodic balanced sequences over a~binary alphabet coincide (see~\cite{HeMo}).
	Vuillon~\cite{Vui} provides a survey on some previous work on balanced sequences.
	
	A \textit{morphism} over $\A$ is a mapping $\psi: \A^* \to \A^*$ such that $\psi(uv) = \psi(u) \psi(v)$ for all $u, v \in \A^*$.
	A morphism $\psi$ can be naturally extended to sequences by setting
	$\psi(u_0 u_1 u_2 \cdots) = \psi(u_0) \psi(u_1) \psi(u_2) \cdots\,$.
	A \textit{fixed point} of a morphism $\psi$ is a sequence $\uu$ such that $\psi(\uu) = \uu$.
	
	Consider a factor $w$ of a recurrent sequence $\uu = u_0 u_1 u_2 \cdots$.
	Let $i < j$ be two consecutive occurrences of $w$ in $\uu$.
	Then the word $u_i u_{i+1} \cdots u_{j-1}$ is a \textit{return word} to $w$ in $\uu$.
	The set of all return words to $w$ in $\uu$ is denoted by $\R_\uu(w)$.
	If $\uu$ is uniformly recurrent, the set $\R_\uu(w)$ is finite for each factor $w$.
	The opposite is true if $\uu$ is recurrent.
	In this case, if $p$ is the shortest prefix of $\uu$ such that $pw$ is a prefix of $\uu$, then $p^{-1}\uu$ can be written as a concatenation $p^{-1} \uu = r_{d_0} r_{d_1} r_{d_2} \cdots$ of return words to $w$.
	The \textit{derived sequence} of $\uu$ to $w$ is the sequence
	$\dd_\uu(w) = d_0 d_1 d_2 \cdots$
	over the alphabet of cardinality $\# \R_\uu(w)$.
	The concept of derived sequences was introduced by Durand~\cite{Dur98}.
	
	\begin{remark}
		\label{rem:retwords_extension_to_BS}
		If $\uu$ is an aperiodic recurrent sequence, then each factor $u$ can be uniquely extended to the shortest bispecial factor $b = xuy$ for some possibly empty factors $x, y$.
		It is readily seen that $\R_\uu(b) = x \R_\uu(u) x^{-1}$.
		In particular, the Parikh vectors (and obviously the lengths) of return words to $u$ and $b$ coincide.
		
	\end{remark}
	
	\begin{example}
		\label{ex:FiboDef}
		The well-known \emph{Fibonacci sequence} is the sequence
		$$
		\ff = {\tt abaababaabaababaababaabaababaab}\cdots\,
		$$
		obtained as fixed point of the morphism $\varphi: {\tt a} \mapsto {\tt ab}$, ${\tt b} \mapsto {\tt a}$. Such a sequence is Sturmian (see~\cite{PytheasFogg2002}).
		The reader is invited to check that the return words to the bispecial prefix $b = {\tt abaaba}$ are $r = {\tt  abaab}$ and $s = {\tt aba}$.
		The return words to the factor $u = {\tt aaba}$ are $\hat{r} = {\tt aabab}$ and $\hat{s} = {\tt aab}$.
		This corresponds to Remark~\ref{rem:retwords_extension_to_BS} because $b = {\tt abaaba} = xuy = {\tt ab} u \varepsilon$ and the return words to $b$ and $u$ satisfy
		$r = {\tt  abaab} = x \hat{r} x^{-1}$ and $s = {\tt aba} = x \hat{s} x^{-1}$.
		Note that the extension graphs of the factors of length at most one in $\LL(\ff)$ are the same as in Example~\ref{ex:Fibo_extGraph}.
	\end{example}
	
	Vuillon~\cite{Vui01} showed that an infinite recurrent sequence $\uu$ is Sturmian if and only if each of its factors has exactly two return words.
	Moreover, the derived sequence to a factor of a Sturmian sequence is Sturmian too.

	\section{Languages of balanced sequences}
	\label{Section_BalancedSturmian}
	
	In 2000 Hubert~\cite{Hu} characterised balanced sequences over alphabets of higher cardinality (see also~\cite{Gr}).
	A suitable tool for their description is the notion of constant gap.
	
	\begin{definition}
		\label{Def_ConstantGap}
		A~sequence $\yy$ over an alphabet $\A$ is a {\em constant gap sequence} if for each letter $a \in \A$ appearing in $\yy$ there is a positive integer denoted by ${\rm gap}_{\yy}(a)$ such that the distance between successive occurrences of $a$ in $\yy$ is always ${\rm gap}_{\yy}(a)$.
	\end{definition}
	
	Obviously, every constant gap sequence is periodic.
	
	We let ${\rm Per}(\yy)$ denote the minimal period length of $\yy$.
	Note that ${\rm gap}_{\yy}(a)$ divides ${\rm Per}(\yy)$ for each letter $a$ appearing in $\yy$.
	Given a constant gap sequence $\yy$ and a word $u \in \LL(\yy)$ we let ${\rm gap}_{\yy}(u)$ denote the length of the gap between two successive occurrences of $u$ in $\yy$.
	Note that ${\rm gap}_{\yy}(\varepsilon) = 1$.
	If $u = u_0 u_1 \cdots u_{k-1}$, with $u_i \in \A$, then ${\rm gap}_{\yy}(u) = \lcm \{ {\rm gap}_{\yy}(u_i) \; :  \; 0 \le i \le k-1 \}$.
	
	\begin{example}
		\label{ex:sequences}
		The sequences $\yy = ({\tt 0102})^{\omega}$ and $\yy' = ({\tt 34})^\omega$ are constant gap sequences.
		Indeed, the distance between consecutive occurrences of a letter $a \in \{ {\tt 3,4} \}$ is always $2$ in $\yy'$;
		while in $\yy$ one has ${\rm gap}_{\yy}({\tt 0}) = 2,\ {\rm gap}_{\yy}({\tt 1}) = {\rm gap}_{\yy}({\tt 2})  = 4$ and ${\rm gap}_{\yy}(u) = 4$ for each factor $u$ of $\yy$ with $|u| \ge 2$.
		Clearly ${\rm Per}(\yy) = 4$ and ${\rm Per}(\yy') = 2$.
		
		The sequence $({\tt 011})^{\omega}$ is periodic but it is not a constant gap sequence since the distance between consecutive ${\tt 1}$'s is sometimes $1$ and sometimes $2$.
	\end{example}
	
	\begin{Observation}
		\label{obs:per}
		Given a constant gap sequence $\yy$ we have
		$$
		{\rm Per}(\yy) = \max \left\{ {\rm gap}_{\yy}(u) : u \in \LL(\yy) \right\}.
		$$
		Moreover, ${\rm gap}_{\yy}(u)$ divides ${\rm Per}(\yy)$ for every factor $u \in \LL(\yy)$.
	\end{Observation}
	
	Given a constant gap sequence $\yy$, for every positive integer $n$ we define the set
	$$
	\gap{\yy}{n} = \{ i \; : \; \exists \, y \in \LL(\yy), |y| = n, \; {\rm gap}_{\yy}(y) = i \}\,.
	$$
	It is clear that $\gap{\yy}{0} = \{1\}$ for every constant gap sequence $\yy$.
	
	\begin{example}
		Let $\yy = ({\tt 0102})^{\omega}$ and $\yy' = ({\tt 34})^\omega$ be the sequences as in Example~\ref{ex:sequences}.
		One has $\gap{\yy'}{n} = \{ 2 \}$ for every $n \ge 1$; $\gap{\yy}{1}= \{ 2,4 \}$ and $\gap{\yy}{n}= \{ 4\}$ for every $n \ge 2$.
	\end{example}
	
	\begin{theorem}[\cite{Hu}]
		\label{Hubert}
		A recurrent aperiodic sequence $\vv$ is balanced if and only if $\vv$ is obtained from a Sturmian sequence $\uu$ over $\{ {\tt a,b} \}$ by replacing the ${\tt a}$'s in $\uu$ by a constant gap sequence $\yy$ over some alphabet $\A$, and replacing the ${\tt b}$'s in $\uu$ by a constant gap sequence $\yy'$ over some alphabet $\B$ disjoint from $\A$.
	\end{theorem}
	
	\begin{definition}
		Let $\uu$ be a Sturmian sequence over the alphabet $\{ {\tt a}, {\tt b}\}$, and ${\yy}, {\yy'}$ be two constant gap sequences over two disjoint alphabets $\A$ and $\B$.
		The \emph{colouring} of $\uu$ by $\yy$ and $\yy'$, denoted $\vv = \barva( \uu, \yy, \yy')$, is the sequence over $\A \cup \B$ obtained by the procedure described in Theorem~\ref{Hubert}.
	\end{definition}
	
	For $\vv = \barva( \uu, \yy, \yy')$ we use the notation $\pi(\vv) = \uu$ and $\pi(v) = u$ for any $v \in \LL(\vv)$ and the corresponding $u \in \LL(\uu)$.
	Symmetrically, given a word $u \in \LL(\uu)$, we write $\pi^{-1}(u) = \{ v \in \LL(\vv) \; : \; \pi(v) = u\}$.
	We say that $\uu$ (resp., $u$) is a \emph{projection} of $\vv$ (resp., $v$).
	The map $\pi : \LL(\vv) \to \LL(\uu)$ is clearly a morphism.
	Indeed $\pi(\varepsilon) = \varepsilon$ and for every $v, v' \in \LL(\vv)$ one has $\pi(v v') = \pi(v) \pi(v')$.
	
	\begin{example}
		\label{ex:Fibo}
		Let $\ff$ be as in Example~\ref{ex:FiboDef}.
		Let us consider the constant gap sequences $\yy = ({\tt 0102})^{\omega}$ and $\yy'=({\tt 34})^{\omega}$ over the alphabets $\A = \{ {\tt 0,1,2} \}$ and $\B = \{ {\tt 3,4} \}$ respectively.
		The sequence
		$$
		\gg = \barva ( \ff, \yy, \yy' ) = {\tt 0310423014023041032401302403104} \cdots
		$$
		is balanced according to Theorem~\ref{Hubert}.
		One has $\pi(\gg) = \ff$.
		Moreover,
		$\pi({\tt 031}) = \pi({\tt 041}) = {\tt aba}$,
		and
		$\pi^{-1}({\tt aba}) = \{ {\tt 031}, {\tt 032}, {\tt 041}, {\tt 042}, {\tt 130}, {\tt 140}, {\tt 230}, {\tt 240} \}$.
	\end{example}
	
	In the sequel we focus on symmetries of the languages of balanced sequences.
	To do that, the WDO property on binary sequences will be useful.
	
	\begin{definition}
		\label{WDO}
		An aperiodic sequence $\uu$ over $\{ {\tt a,b} \}$ has {\em well distributed occurrences}, or has the {\em WDO property}, if for every $m \in \N$ and for every $w \in \LL(\uu)$ one has
		$
		\left\{ \Parickh(p) \bmod m
		:
		pw \mbox{ is a prefix of } \uu
		\right\}
		=
		\mathbb{Z}_m^2.
		$
	\end{definition}
	
	It is known that Sturmian sequences have the WDO property (see~\cite{WDO}).
	
	\begin{example}
		Let $\ff$ be as in Example~\ref{ex:FiboDef} and let us consider $m = 2$ and $w = {\tt ab} \in \LL(\ff)$.
		Then it is easy to check that
		$$
		\begin{array}{cc}
			\Parickh(\varepsilon) \equiv \begin{pmatrix} 0 \\ 0 \end{pmatrix} \pmod 2, &
			\Parickh({\tt aba}) \equiv \begin{pmatrix} 0 \\ 1 \end{pmatrix} \pmod 2, \\
			\Parickh({\tt abaab}) \equiv \begin{pmatrix} 1 \\ 0 \end{pmatrix} \pmod 2, &
			\quad
			\Parickh({\tt abaababa}) \equiv \begin{pmatrix} 1 \\ 1 \end{pmatrix} \pmod 2,
		\end{array}
		$$
		where $w, {\tt aba}w, {\tt abaab}w$ and ${\tt abaababa}w$ are prefixes of $\ff$.
	\end{example}
	
	Using the WDO property we can prove that in order to study the language of aperiodic recurrent balanced sequences, it is enough to consider standard Sturmian sequences.
	
	The \emph{shift} of a constant gap sequence $\yy = (y_0 y_1 \cdots y_{k-1})^\omega$ is the sequence $\sigma(\yy) = (y_{1} y_2 \cdots y_{k-1} y_0)^\omega$.
	
	\begin{example}
		Let ${\yy = ({\tt 0102})^\omega}$.
		Then we have $\sigma^0({\yy}) = {\yy}$, \ $\sigma({\yy}) = ({\tt 1020})^\omega$, \ $\sigma^2({\yy}) = ({\tt 0201})^\omega$, \ $\sigma^3({\yy}) = ({\tt 2010})^\omega$ and $\sigma^4({\yy}) = {\yy}$.
	\end{example}
	
	The next proposition says that for a balanced sequence $\vv$, the language $\LL(\vv)$ does not depend on the projection $\uu$ itself but only on $\LL(\uu)$.
	In the following chapters, we will always consider balanced sequences obtained as colouring of standard Sturmian sequences.
	
	\begin{proposition}
		\label{pro:standardSturm}
		Let $\uu,\uu'$ be two Sturmian sequences such that $\LL(\uu) = \LL(\uu')$, $\yy$ and $\yy'$ two constant gap sequences over disjoint alphabets and $i,j \in \N$.
		Let $\vv = \barva( \uu, \yy, \yy')$, $\vv' = \barva( \uu', \yy, \yy')$ and $\vv'' = \barva( \uu, \sigma^i(\yy), \sigma^j(\yy'))$.
		Then $\LL(\vv) = \LL(\vv') = \LL(\vv'')$.
	\end{proposition}
	\begin{proof}
		Let $v$ and $w$ be words such that $v \in \LL(\vv)$ and $wv$ is a prefix of $\vv$.
		Then $\pi(w) \pi(v)$ is a prefix of $\uu$ and $|\pi(w)|$ is an occurrence of $\pi(v)$ in $\uu$.
		Since $\pi(v) \in \LL(\uu')$, using the WDO property, we can find $p \in \LL(\uu')$ such that $p \pi(v)$ is a prefix of $\uu'$ and $\Parickh(\pi(w)) \equiv \Parickh(p) \pmod {{\rm Per}(\yy) {\rm Per}(\yy')}$.
		Thus $v$ appears both in $\vv$ at occurrence $|\pi(w)|$ and in $\vv'$ at occurrence $|p|$.
		Hence $\LL(\vv) \subset \LL(\vv')$.
		Using the same argument we can prove the opposite inclusion.
		
		Let $p$ be a prefix of $\uu$ such that
		$\Parickh(p) \equiv \left( \begin{smallmatrix} i \\ j \end{smallmatrix} \right) \pmod {{\rm Per}(\yy) {\rm Per}(\yy')}$.
		Let us write $\uu''= p^{-1}\uu$.
		Then $\barva(\uu'', \sigma^i(\yy), \sigma^j(\yy'))$ gives the same sequence as the one obtained by erasing the prefix of length $|p|$ from $\vv$.
		Since $\LL(\uu) = \LL(\uu'')$, using the same argument as before we have $\LL(\vv'') = \LL(\vv)$.
	\end{proof}
	
	The following invariance of the language of a balance sequence is a consequence of the previous proposition.
	
	\begin{corollary}
		\label{cor:cyclicshift}
		Let $\vv = \barva( \uu, \yy, \yy')$ and $u \in \LL(\uu)$.
		For any non-negative integers $i,j$, the word $v$ obtained from $u$ by replacing the ${\tt a}$'s by $\sigma^i(\yy)$ and the ${\tt b}$'s by $\sigma^j(\yy')$ belongs to  $\LL(\vv)$.
	\end{corollary}
	
	\begin{example}
		\label{ex:fibo-pi}
		Let $\ff, \gg, \yy$ and $\yy'$ be as in Example~\ref{ex:Fibo}.
		Let $v = {\tt 03104} \in \LL(\gg)$ and let us write $u = \pi(v) = {\tt abaab}$.
		One can easily check that the word $v' = {\tt 24013}$ obtained from $u$ by replacing the ${\tt a}$'s by $\sigma^3(\yy)$ and the ${\tt b}$'s by $\sigma(\yy')$ is also in $\LL(\gg)$.
	\end{example}
	
	Since a constant gap sequence is periodic, it is clear that every sufficiently long factor in the sequence is neither right special nor left special.
	Let us define, for a given constant gap sequence $\yy$, the number
	$$
	\bs(\yy) = \max \{ |u| : u \mbox{ is a bispecial factor of} \ \yy \}.
	$$
	It is an obvious consequence that no factor of $\yy$ of length more than $\bs(\yy)$ is left special or right special.
	It immediately follows that for $n > \bs(\yy)$, we have ${\rm{gap}}(\yy, n) = \{ \rm{Per}(\yy) \}$.
	
	Note that if $\vv = \barva( \uu, \yy, \yy')$, there exists an $m \in \N$ such that every factor $v \in \LL(\vv)$ longer than $m$ contains more than $\bs(\yy)$ letters in $\A$ and more than $\bs(\yy')$ letters in $\B$.
	Indeed, it is enough to find $m$ such that all factors of length $m$ in $\LL(\uu)$ contain more than $\bs(\yy)$ ${\tt a}$'s and more than $\bs(\yy')$ ${\tt b}$'s.
	
	We will see in the sequel that sufficiently long factors $w$ in a colouring $\vv = \barva( \uu, \yy, \yy')$, i.e., such that $|\pi(w)|_{\tt a} > \bs(\yy)$ and $|\pi(w)|_{\tt b} > \bs(\yy')$, are easy to treat.
	
	\begin{example}
		\label{ex:Fibo7}
		Let $\ff, \gg, \yy$ and $\yy'$ be as in Example~\ref{ex:Fibo}.
		One has $\bs(\yy) = 1$ since the letter ${\tt 0}$ is bispecial in $\yy$ and $\bs(\yy') = 0$ since the only bispecial factor in $\LL(\yy')$ is the empty word.
		It is easy to check that all factors of length $4$ in $\LL(\ff)$ contain at least 2 ${\tt a}$'s and one ${\tt b}$.
		Thus, all factors of length $4$ in $\LL(\gg)$ contain at least two letters in $\A$ and at least one letter in $\B$.
		On the other hand, ${\tt bab} \in \LL(\ff)$ has length $3$ and contains only one ${\tt a}$.
	\end{example}
	
	As we saw in Example~\ref{ex:Fibo}, the set $\pi^{-1}(u)$ for a word $u \in \LL(\uu)$ is not, in general, a singleton.
	However, it is not difficult to prove that every sufficiently long factor in $\LL(\vv)$ is uniquely determined, between the words having the same projection in $\uu$, by its first letters in $\A$ and its first letters in $\B$.
	
	\begin{lemma}
		\label{lem:uniquepre}
		Let $\vv = \barva( \uu, \yy, \yy')$ and $u \in \LL(\uu)$ such that $|u|_{\tt a} > \bs(\yy)$ and $|u|_{\tt b} > \bs(\yy')$.
		Let $a_0 a_1 \cdots a_{\bs(\yy)} \in \LL(\yy)$ and $b_0 b_1 \cdots b_{\bs(\yy')} \in \LL(\yy')$.
		There exists exactly one word in $\pi^{-1}(u)$ having $a_0, a_1, \ldots, a_{\bs(\yy)}$ (in this order) as first letters in $\A$ and $b_0, b_1, \ldots, b_{\bs(\yy')}$ (in this order) as first letters in $\B$.
	\end{lemma}
	\begin{proof}
		Since $a_0 a_1 \cdots a_{\bs(\yy)} \in \LL(\yy)$, by Corollary~\ref{cor:cyclicshift} there exists an integer $i$ such that $a_0 a_1 \cdots a_{\bs(\yy)}$ is a prefix of $\sigma^i(\yy)$.
		Since no factor of length $\bs(\yy) + 1$ is right special in $\yy$, the unique prolongation of $a_0 a_1 \cdots a_{\bs(\yy)}$ to any length is a prefix of $\sigma^i(\yy)$.
		A similar argument can be used for $b_0 b_1 \cdots b_{\bs(\yy')}$.
	\end{proof}
	
	\begin{example}
		\label{ex:Fibou}
		Let $\ff, \gg, \yy$ and $\yy'$ be as in Example~\ref{ex:Fibo} and $u = {\tt abaab} \in \LL(\ff)$.
		One has $|u|_{\tt a} = 3 > 1 = \bs(\yy)$ and $|u|_{\tt b} = 2 > 0 = \bs(\yy')$.
		One can check that, according to Lemma~\ref{lem:uniquepre}, the only word in $\pi^{-1}(u)$ having ${\tt 0}, {\tt 2}$ as first letters in $\A$ and ${\tt 4}$ as first letter in $\B$ is ${\tt 04203}$, which is the word obtained from $u$ by $\sigma^2(\yy)$ and $\sigma(\yy')$ respectively.
		On the other hand, no word in $\LL(\gg)$ can have, for instance, ${\tt 1}$ as first letter in $\A$ and  ${\tt 2}$ as second letter in $\A$ since ${\tt 12} \notin \LL(\yy)$.
	\end{example}
	
	Putting together Corollary~\ref{cor:cyclicshift} and Lemma~\ref{lem:uniquepre}, we obtain the following result.
	
	\begin{lemma}
		\label{lem:cardpi-1}
		Let $\vv = \barva( \uu, \yy, \yy')$ and $u \in \LL(\uu)$ be such that $|u|_{\tt a} > \bs(\yy)$ and $|u|_{\tt b} > \bs(\yy')$.
		Then $\# \pi^{-1}(u) = {\rm Per}(\yy) {\rm Per}(\yy')$.
	\end{lemma}
	\begin{proof}
		Let us set $\yy = (y_0 y_1 \cdots y_{k-1})^{\omega}$ and $\yy' = (y'_0 y'_1 \cdots y'_{\ell-1})^{\omega}$, with $k = {\rm Per}(\yy)$ and $\ell = {\rm Per}(\yy')$.
		Using Corollary~\ref{cor:cyclicshift} we know that for every $i,j$ such that $0 \le i < {\rm Per}(\yy)$ and $0 \le j < {\rm Per}(\yy')$, the word $v_{i,j}$ obtained from $u$ by replacing the ${\tt a}$'s by $\sigma^i(\yy)$ and the ${\tt b}$'s by $\sigma^j(\yy')$ is in $\pi^{-1}(u)$.
		Moreover, the factors $v_{i,j}$ are distinct for distinct pairs $(i,j)$.
		This is a consequence of the assumption that $|u|_{\tt a} > \bs(\yy)$ and $|u|_{\tt b} > \bs(\yy')$ and the fact that for $n > \bs(\yy)$, we have ${\rm{gap}}(\yy, n) = \{ \rm{Per}(\yy) \}$ and similarly, for $n' > \bs(\yy')$, we have ${\rm{gap}}(\yy', n') = \{ \rm{Per}(\yy') \}$.
		Thus, we obtain exactly ${\rm Per}(\yy) {\rm Per}(\yy')$ distinct words.
	\end{proof}
	
	\begin{example}
		Let $\ff, \gg, \yy$ and $\yy'$ be as in Example~\ref{ex:Fibo} and $u = {\tt abaa} \in \LL(\ff)$.
		The set
		$\pi^{-1}(u) = \{ {\tt 0310},$
		${\tt 0410},$
		${\tt 0320},$
		${\tt 0420},$
		${\tt 1302},$
		${\tt 1402},$
		${\tt 2301},$
		${\tt 2401}
		\}$
		has exactly $8 = 4 \cdot 2$ elements, which is consistent with Lemma~\ref{lem:cardpi-1}.
	\end{example}
	
	\begin{lemma}
		\label{lem:special}
		Let $\vv = \barva(\uu, \yy, \yy')$ and $w \in \LL(\vv)$.
		\begin{enumerate}
			\item If $\pi(w)$ is bispecial in $\uu$, then $w$ is bispecial in $\vv$.
			\item If $w$ is bispecial in $\vv$, $|\pi(w)|_{\tt a} > \bs(\yy)$ and $|\pi(w)|_{\tt b} > \bs(\yy')$, then $\pi(w)$ is bispecial in $\uu$.
		\end{enumerate}
	\end{lemma}
	\begin{proof}
		\begin{enumerate}
			\item
			The first statement follows directly from the definition of colouring.
			
			\item
			If $w$ is bispecial in $\vv$, $|\pi(w)|_{\tt a} > \bs(\yy)$ and $|\pi(w)|_{\tt b} > \bs(\yy')$, then there exist a unique right extension of $w$ in $\A$ and a unique right extension of $w$ in $\B$; similarly for left extensions.
			Thus $\pi(w)$ is bispecial in $\uu$.
		\end{enumerate}
	\end{proof}
	
	\begin{example}
		\label{ex:shortBS}
		Let $\ff, \gg, \yy$ and $\yy'$ be as in Example~\ref{ex:Fibo}.
		\begin{itemize}
			\item
			Consider $w = {\tt 1}$ with $\pi(w) = {\tt a}$.
			Since $\pi(w)$ is bispecial in $\ff$, by Item 1 of Lemma~\ref{lem:special}, the factor $w$ has to be bispecial too.
			Indeed, ${\tt 310, 014, 410, 013} \in \LL({\gg})$.
			
			\item
			Consider $w = {\tt 240}$, then ${\tt 32401, 02403, 02401} \in \LL(\gg)$, hence $w$ is bispecial.
			Moreover, $\pi(w) = {\tt aba}$, $|\pi(w)|_{\tt a} = 2 > 1 = \bs(\yy )$ and $|\pi(w)|_{\tt b} = 1 > 0 = \bs(\yy')$.
			By Item 2 of Lemma~\ref{lem:special}, $\pi(w)$ is bispecial in $\ff$.
			Indeed, ${\tt aabaa, aabab, babaa} \in \LL(\ff)$.
			
			\item
			However, for a factor $w  \in \LL(\gg)$, it may happen that $w$ is bispecial and $\pi(w)$ is not.
			Of course, by Item 2 of Lemma~\ref{lem:special}, in such a case $|\pi(w)|_{\tt a} \leq \bs(\yy)$ or $|\pi(w)|_{\tt b} \leq \bs(\yy')$.
			Consider $w = {\tt 3}$, then ${\tt 031, 130, 032, 230} \in \LL({\gg})$; thus $w$ is bispecial; however $\pi(w) = {\tt b}$ is not a bispecial factor in $\ff$.
		\end{itemize}
	\end{example}
	
	Using the previous lemmata we can prove the following result.
	
	\begin{proposition}
		\label{pro:dendric}
		Let $\vv = \barva(\uu, \yy, \yy')$.
		The language $\LL(\vv)$ is eventually dendric with threshold $m = \min\{|u| : u \in \LL(\uu), |u|_{\tt a} > \bs(\yy) \text{ and } |u|_{\tt b} > \bs(\yy') \}$.
	\end{proposition}
	\begin{proof}
		Let $w \in \LL(\vv)$ with length at least $m$ and $u = \pi(w)$.
		It easily follows from the proofs of Lemmata~\ref{lem:uniquepre} and~\ref{lem:special} that $\E_{\vv}(w)$ is isomorphic to $\E_{\uu}(u)$ via the projection $\pi$.
		Since $\uu$ is Sturmian, $\LL(\uu)$ is dendric.
		Thus $\E_{\vv}(v)$ is a tree.
		Hence $\LL(\vv)$ is eventually dendric of threshold $m$.
	\end{proof}
	
	Note that from the proof of Proposition~\ref{pro:dendric} it follows that for a sufficiently long word $u \in \LL(\uu)$ all words in $\pi^{-1}(u)$ have isomorphic extension graphs.
	
	\begin{example}
		\label{ex:dendric}
		Let $\ff, \gg, \yy$ and $\yy'$ be as in Example~\ref{ex:Fibo}.
		As we have seen in Example~\ref{ex:Fibo7}, all words of length at least $4$ in $\LL(\ff)$ contain at least two ${\tt a}$'s and at least one ${\tt b}$.
		Let $v = {\tt 230140} \in \LL(\gg)$.
		The extension graphs of $v$ in $\LL(\gg)$ and of $u = \pi(v)$ in $\LL(\ff)$ are represented in Figure~\ref{fig:uvextgraphs}.
	\end{example}
	
	\begin{figure}
		\centering
		\tikzset{node/.style={rectangle,draw,rounded corners=1.2ex}}
		\begin{tikzpicture}
			\node[node](0l) {${\tt 0}$};
			\node[node](4l) [below= 0.2cm of 0l] {${\tt 4}$};
			\node[node](2r) [right= 1cm of 0l] {${\tt 2}$};
			\node[node](3r) [below= 0.2cm of 2r] {${\tt 3}$};
			\path[draw,thick, shorten <=0 -1pt, shorten >=-1pt]
			(0l) edge node {} (3r)
			(4l) edge node {} (2r)
			(4l) edge node {} (3r);
			\node[node](al) [right= 2cm of 2r] {${\tt a}$};
			\node[node](bl) [below= 0.2cm of al] {${\tt b}$};
			\node[node](ar) [right= 1cm of al] {${\tt a}$};
			\node[node](br) [below= 0.2cm of ar] {${\tt b}$};
			\path[draw,thick, shorten <=0 -1pt, shorten >=-1pt]
			(al) edge node {} (br)
			(bl) edge node {} (ar)
			(bl) edge node {} (br);
		\end{tikzpicture}
		\caption{The extension graphs of $v = {\tt 230140} \in \LL(\gg)$ (on the left) and $\pi(v) = {\tt abaaba} \in \LL(\ff)$ (on the right).}
		\label{fig:uvextgraphs}
	\end{figure}
	
	Dolce and Perrin~\cite{DolcePerrin20} studied eventually dendric sequences.
	In particular, they showed that the sequence $s_\uu(n)$ is eventually constant.
	It immediately gives the following.
	
	\begin{proposition}
		\label{pro:complexityeventually}
		Let $\uu$ be an eventually dendric sequence with threshold $m$.
		For every $n \ge m$ one has
		$\C_\uu(n) = s_\uu(m) n + K$,
		with $K$ a constant.
	\end{proposition}
	
	The following result easily follows from Lemma~\ref{lem:cardpi-1} and it can be seen as a particular case of the previous proposition.
	
	\begin{proposition}
		\label{pro:complexityv}
		Let $\vv = \barva( \uu, \yy, \yy')$ and $m$ be the threshold given in  Proposition \ref{pro:dendric}.
		Then for every $n \ge m$ one has
		$\C_{\vv}(n) = {{\rm Per}(\yy)}{{\rm Per}(\yy') }(n+1)$.
	\end{proposition}
	
	\begin{example}
		Let $\ff, \gg, \yy$ and $\yy'$ be as in Example~\ref{ex:Fibo}.
		According to Proposition~\ref{pro:dendric} and Example~\ref{ex:Fibo7} the language $\LL(\gg)$ is eventually dendric with threshold $4$.
		The factor complexity of $\gg$ is defined by $\C_{\gg}(n) = 8(n + 1)$ for every $n \ge 4$, according to Proposition~\ref{pro:complexityv}.
		
		Let us now consider $\gg' = \barva( \ff, \yy', \yy)$.
		It is easy to check that every factor of length at least $6$ in $\LL(\ff)$ contains at least one ${\tt a}$ and at least two ${\tt b}$'s.
		Thus, one has $\C_{\gg'}(n) = 8(n+1)$ for every $n \ge 6$.
		The initial values of $\C_{\gg}(n)$ and $\C_{\gg'}(n)$ are given in Table~\ref{tab:complexityvf}.
	\end{example}
	
	\begin{table}
		\centering
		\begin{tabular}{c|c|c|c|c|c|c}
			n & \; 0 \; & \; 1 \; & \; 2 \; & \; 3 \; & \; 4 \; & \; 5 \; \\
			\hline
			$\C_{\gg}(n)$ & 1 & 5 & 16 & 30 & & \\
			$\C_{\gg'}(n)$ & 1 & 5 & 14 & 26 & 36 & 46
		\end{tabular}
		\caption{Initial values of $\C_{\gg}(n)$ and $\C_{\gg'}(n)$.}
		\label{tab:complexityvf}
	\end{table}
	
	The following result, concerning return words, is~\cite[Theorem 7.3]{DolcePerrin20} (see also~\cite{DolcePerrin19}).
	
	\begin{theorem}[\cite{DolcePerrin20}]
		\label{thm:returneventually}
		Let $\uu$ be a recurrent eventually dendric sequence with threshold $m$.
		For every $w \in \LL(\uu)$, the set $\R_{\uu}(w)$ is finite.
		Moreover, for every $w \in \LL(\uu)$ of length at least $m$, one has
		$\#\R_\uu(w) = s_{\uu}(m) + 1\,$.
	\end{theorem}
	
	Combining Theorem~\ref{thm:returneventually} with  Propositions~\ref{pro:dendric} and~\ref{pro:complexityv} gives us the number of return words to sufficiently long factors in balanced sequences.
	
	\begin{proposition}
		\label{pro:returnbalanced}
		Let $\vv = \barva( \uu, \yy, \yy')$ and $v \in \LL(\vv)$ such that $|\pi(v)|_{\tt a} > \bs(\yy)$ and $|\pi(v)|_{\tt b} > \bs(\yy')$.
		Then $\#\R_{\vv}(v) = 1 + {\rm Per}(\yy) {\rm Per}(\yy')$.
	\end{proposition}
	
	\begin{example}
		Let $\ff, \gg, \yy$ and $\yy'$ be as in Example~\ref{ex:Fibo}.
		The factor $v = {\tt 230}$ is such that its projection $\pi(v) = {\tt aba}$ has more than $\bs(\yy) = 1$ ${\tt a}$'s and more than $\bs(\yy) = 0$ ${\tt b}$'s.
		Thus, according to Proposition~\ref{pro:returnbalanced}, there are exactly $9 = 1 + 4 \cdot 2$ return words to $v$ in $\gg$.
	\end{example}

	Recall that a recurrent sequence is uniformly recurrent if and only if the number of return words to any given factor of the sequence is finite.
	Thus, an interesting consequence of Proposition~\ref{pro:returnbalanced} is the following one.
	
	\begin{corollary}
		\label{cor:recUR}
		A recurrent aperiodic balanced sequence is uniformly recurrent.
	\end{corollary}
	
	Proposition~\ref{pro:returnbalanced} describes the number of return words to sufficiently long factors.
	In the sequel we study the critical exponent of balanced sequences and for this purpose, we need to compute lengths of shortest return words to all factors.
	The following proposition is crucial for our purpose since it detects occurrences of the same factor.
	
	In the sequel we will use the following notation:
	$$
	\begin{pmatrix} a \\ b \end{pmatrix} \bmod \begin{pmatrix} n \\ n' \end{pmatrix} :=
	\begin{pmatrix} a \bmod n \\ b \bmod n' \end{pmatrix}.
	$$
	
	\begin{proposition}
		\label{zeroParikh1}
		Let $\vv = \barva(\uu, \yy, \yy')$.
		Let $u, f \in \LL(\uu)$ such that $fu \in \LL(\uu)$ and $u$ is a prefix of $fu$.
		Then the two statements are equivalent:
		\begin{enumerate}
			\item there exist $w$ and $v$ such that $vw \in \LL(\vv)$, $w$ is a prefix of $vw$, $|w| = |u|$ and $\pi(vw) = fu;$
			
			\item $\Parickh(f) \equiv \left( \begin{smallmatrix} 0 \\ 0 \end{smallmatrix} \right) \pmod {\left( \begin{smallmatrix} n \\ n' \end{smallmatrix} \right)}$ for some $n \in {\rm gap}(\yy, |u|_{\tt a})$ and $n' \in {\rm gap}(\yy', |u|_{\tt b})$.
		\end{enumerate}
	\end{proposition}
	\begin{proof}
		Let $v$ and $w$ be as in Item 1.
		Then $u$ is a prefix and a suffix of $\pi(vw)$ and $f = \pi(v)$.
		By Proposition~\ref{pro:standardSturm}, the factor $w$ occurring as a prefix of $vw$ is obtained from $u$ by colouring the ${\tt a}$'s with $\sigma^s(\yy)$ and the $\tt b$'s with $\sigma^t(\yy')$ for some $s, t \in \N$.
		Hence, the same factor $w$ occurring as a suffix of $vw$ is obtained from $u$ by colouring the ${\tt a}$'s with $\sigma^{S}(\yy)$ and the $\tt b$'s with $\sigma^{T}(\yy')$, where $S = s + |f|_{\tt a}$ and $T = t + |f|_{\tt b}$.
		Hence the prefixes of length $|u|_a$ of $\sigma^s(\yy)$ and $\sigma^S(\yy)$ coincide, and similarly the prefixes of length $|u|_b$ of $\sigma^t(\yy')$ and $\sigma^T(\yy')$ coincide.
		This implies that $|f|_{\tt a}$ is divisible by some $n \in {\rm{gap}}(\yy, |u|_{\tt a})$ and that $|f|_{\tt b}$ is divisible by some $n' \in {\rm{gap}}(\yy', |u|_{\tt b})$.
		In other words, $|f|_{\tt a} \equiv 0 \pmod n$ and $|f|_{\tt b} \equiv 0 \pmod {n'}$.
		
		Let $f, n$ and $n'$ be as in Item 2.
		Let us consider $y \in \LL(\yy)$ and $y' \in \LL(\yy')$ such that ${\rm gap}_{\yy}(y) = n$ with $|y| = |u|_{\tt a}$ and $ {\rm gap}_{\yy}(y') = n'$ with $ |y'| = |u|_{\tt b}$.
		Let $s, t \in \N$ be such that $y$ is a prefix of $\sigma^s(\yy)$ and $y'$ is a prefix of $\sigma^t(\yy')$.
		Colouring the letters ${\tt a}$'s in $fu$ with $\sigma^s(\yy)$ and the letters ${\tt b}$'s with $\sigma^t(\yy')$, we get, by Proposition~\ref{pro:standardSturm}, a factor $x$ of $\vv$.
		Since $|f|_{\tt a }$ is a multiple of ${\rm gap}_{\yy}(y)$ and $|f|_{\tt b }$ is a multiple of ${\rm gap}_{\yy}(y')$, the prefix and the suffix of length $|u|$ of $x$ coincide, i.e., $x = vw$, $w$ is a prefix of $vw$, $|w| = |u|$ and $\pi(vw) = fu$.
	\end{proof}
	
	\begin{example}
		\label{ex:retwords_short}
		Let $\ff, \gg, \yy$ and $\yy'$ be as in Example~\ref{ex:Fibo}, i.e,
		$$
		\ff = {\tt \underline{\textcolor{red}{a}ba}\textcolor{red}{a}baba\underline{\textcolor{red}{a}ba}\textcolor{red}{a}babaababa\underline{\textcolor{red}{a}ba}\textcolor{red}{a}babaab}\cdots\, ,
		$$
		$$
		\gg ={\tt \overline{\textcolor{green}{0}31}\textcolor{green}{0}4230\textcolor{blue}{1}402304103240\textcolor{blue}{1}302403104}\cdots .
		$$
		Let $u = \textcolor{red}{\tt a}$ and $f = \underline{{\tt aba}}$.
		Then $fu$ is a prefix of $\ff$ and ${\rm gap}(\yy, |u|_{\tt a})={\rm gap}(\yy, 1) = \{ 2,4 \}$ and ${\rm gap}(\yy', |u|_{\tt b}) = {\rm gap}(\yy', 0) = \{ 1 \}$.
		Set $n = 2$ and $n' = 1$.
		We have $\Parickh(f) = \left( \begin{smallmatrix} 2 \\ 1 \end{smallmatrix} \right) \equiv \left( \begin{smallmatrix} 0 \\ 0 \end{smallmatrix} \right) \pmod {\left( \begin{smallmatrix} n \\ n' \end{smallmatrix} \right)}$.
		Then, by Proposition~\ref{zeroParikh1}, there exist $w$ and $v$ such that $vw \in \LL(\gg)$, $w$ is a prefix of $vw$, $|w| = |u|$ and $\pi(vw) = fu$.
		Indeed, it suffices to put $w =\textcolor{green}{ {\tt 0}}$ and $v = \overline{{\tt 031}}$.
		Moreover, $v = {\tt 031}$ is a return word to $w = {\tt 0}$.
		On the other hand, the reader may easily check that no other factor $w'$ with the same projection $\pi(w') = {\tt a}$, i.e., $w' \in \{ {\tt \textcolor{blue}{1},2} \}$, has a return word $v'$ with projection $\pi(v') = f = {\tt aba}$; any such factor $w'$ has return words of length greater than or equal to six.
	\end{example}
	
	The next corollary shows that the set of lengths of return words is the same for all sufficiently long factors of balanced sequences having the same projection.
	
	\begin{corollary}
		\label{coro:retvw_long}
		Let $\vv = \barva( \uu, \yy, \yy')$, $u \in \LL(\uu)$ with $|u|_{\tt a} > \bs(\yy)$, $|u|_{\tt b} > \bs(\yy')$ and $w, w' \in \LL(\vv)$ with $\pi(w) = \pi (w') = u$.
		Then $\pi(\R_{\bf v}(w)) = \pi(\R_{\bf v}(w'))$.
		
		In particular, all shortest return words to $w$ and to $w'$ have the same length.
	\end{corollary}
	\begin{proof}
		Let $v \in \R_{\vv}(w)$.
		Then $u$ is both a prefix and a suffix of $fu=\pi(vw)$.
		By Proposition~\ref{zeroParikh1}, we have $\Parickh(f) \equiv \left( \begin{smallmatrix} 0 \\ 0 \end{smallmatrix} \right) \pmod {\left( \begin{smallmatrix} {\rm Per(\yy)} \\ {\rm Per(\yy')}\end{smallmatrix} \right)}$, and no shorter prefix of $f$ satisfies this condition.
		By Corollary~\ref{cor:cyclicshift}, the factor $w'$ is obtained by colouring of $u$ with $\sigma^i(\yy)$ and $\sigma^j(\yy')$ for some $i,j \in \N$.
		When colouring $fu$ with $\sigma^i(\yy)$ and $\sigma^j(\yy')$, we get, by Proposition~\ref{zeroParikh1}, a factor $v'w'$ starting in $w'$ and having no other occurrence of $w'$.
		Thus $v' \in \R_{\vv}(w')$ and $\pi(v') = f = \pi(v)$.
	\end{proof}
	
	\begin{example}
		\label{ex:retwords_long}
		Let $\ff, \gg, \yy$ and $\yy'$ be as in Example~\ref{ex:Fibo}, i.e,
		$$
		\ff = {\tt abaab\underline{\textcolor{red}{aba}aba}\textcolor{red}{aba}baab\underline{\textcolor{red}{aba}aba}\textcolor{red}{aba}baab}\cdots\, ,
		$$
		$$
		\gg ={\tt 03104\overline{\textcolor{green}{230}140}\textcolor{green}{230}4103\overline{\textcolor{blue}{240}130}\textcolor{blue}{240}3104}\cdots .
		$$
		Consider $w = \textcolor{green}{{\tt 230}}$, then $u = \pi(w) =\textcolor{red}{ {\tt aba}}$, $|u|_{\tt a} = 2 > 1 = \bs(\yy)$ and $|u|_{\tt b} = 1 > 0 = \bs(\yy')$.
		One can easily check that $v = \overline{{\tt 230140}}$ is a return word to $w$ and $\pi(v) = \underline{{\tt abaaba}}$.
		By Corollary~\ref{coro:retvw_long}, the factor $w' =\textcolor{blue}{ {\tt 240}}$ satisfying $\pi(w') = {\tt aba} = u$ has to have a return word $v'$ with the same projection $\pi(v') = {\tt abaaba}$.
		Indeed, $v' =\overline{ {\tt 240130}}$ is a return word to $w'$ in $\gg$ and $\pi(v') = {\tt abaaba}$.
	\end{example}

	\section{Critical exponent and its relation to return words}
	\label{Section_CriticalExponent}
	
	If $z \in \A^+$ is a prefix of a periodic sequence $u^\omega$ with $u \in \A^+$, we write $z = u^e$, where $e = |z|/|u|$.
	For a non-empty factor $u$ of an infinite sequence $\uu$ we define the \emph{index} of $u$ in $\uu$ as
	$$
	\ind_\uu(u) = \sup\{ e \in \Q : u^e \in \LL(\uu)\}\,.
	$$
	If $\uu$ is periodic, then there exists a factor of $\uu$ with an infinite index.
	Even an aperiodic sequence may contain factors having an infinite index.
	This phenomenon is excluded in aperiodic uniformly recurrent sequences.
	
	\begin{example}
		The language of the Fibonacci sequence ${\ff}$ defined in Example~\ref{ex:FiboDef} contains ${\tt abaabaaba} = ({\tt aba})^3$.
		It is not difficult to check that $\ind_{\ff}({\tt aba}) = 3$.
	\end{example}
	
	The \emph{critical exponent} of an infinite sequence $\uu$ is defined as
	$$
	\CR(\uu)
	=
	\sup\{e \in \Q : \text{there exist } y, x \in \LL(\uu) \text{ with } |x| > 0 \text{ and } x^e = y  \}.
	$$
	Obviously, if $\uu$ is aperiodic  uniformly recurrent, then $\CR(\uu)  = \sup\{\ind_\uu(u) : u \in \LL(\uu)^+ \}$.
	Let us point out that, although each factor of a uniformly recurrent sequence has a finite index, the critical exponent $\CR(\uu)$ may be infinite.
	An example of such a sequence is given by Sturmian sequences whose slope has a continued fraction expansion with unbounded partial quotients (see~\cite{Cade} and \cite{DaLe}).
	
	If no factor of $\uu$ has an infinite index, we define the \emph{asymptotic critical exponent} of $\uu$ as
	$$
	\CR^*(\uu)
	=
	\limsup\limits_{n \to \infty} \ \left( \max \{ \ind_\uu(u) : u \in \LL(\uu) \text{ with } |u| = n \} \right)\,.
	$$
	Otherwise, we set $\CR^*(\uu) = +\infty$.
	Clearly, $\CR^*(\uu) \leq \CR(\uu)$.
	Nevertheless, if $\CR(\uu) = +\infty$, then $\CR^*(\uu) = +\infty$ as well.
	If $\CR(\uu) < +\infty$, then the asymptotic critical exponent can be expressed as
	$$
	\CR^*(\uu)
	=
	\lim_{n\to \infty} \sup \left\{ e \in \Q : \text{there exist } y, x \in \LL(\uu) \text{ with } |x| > n \text{ and } x^e = y \right\}.
	$$
	In the remaining part of this section we give more handy formulae for critical and asymptotic critical exponents in the case of uniformly recurrent sequences.
	They exploit return words to bispecial factors.
	For this purpose we need to state first two auxiliary lemmata.
	
	\begin{lemma}
		\label{Lem_URbispecial_retwords}
		Let $u, w$ be non-empty factors of a recurrent sequence $\uu$.
		If $u$ is a return word to $w$, then $w = u^{e}$ for some $e \in \Q$.
		Moreover, if $\uu$ is aperiodic and uniformly recurrent, then $u$ is a return word to a finite number of factors in $\uu$.
	\end{lemma}
	\begin{proof}
		Since $u \in \R_{\uu}(w)$, then $w$ is a prefix of $uw$.
		Hence there exists $z \in \LL(\uu)$ such that $uw = wz$.
		Using the first Lyndon-Sch\u{u}tzenberger Theorem (see~\cite{LySc}) we know that there exist $x, y \in \LL(\uu)$ and a non-negative integer $i$ such that $u = xy$, $z = yx$ and $w = (xy)^ix$.
		Thus, $w$ is a prefix of $u^\omega = (xy)^\omega$.
		
		Let us now suppose that $u$ is a return word to infinitely many factors.
		By the previous argument, $u$ is a fractional root of all those factors.
		This implies that $u^n \in \LL(\uu)$ for all $n \in \N$.
		Thus, $\uu$ is either periodic or not uniformly recurrent, a contradiction.
	\end{proof}
	
	\begin{lemma}
		\label{Lem_URbispecial_retwords2}
		Let $\uu$ be a uniformly recurrent aperiodic sequence and $f$ a non-empty factor of $\uu$ such that $\ind_\uu(f) > 1$.
		Then there exist a factor $u \in \LL(\uu)$ and a bispecial factor $w \in \LL(\uu)$ such that $|f| = |u|$, $\ind_\uu(f) \leq \ind_\uu(u) = 1 + \frac{|w|}{|u|}$ and $u \in \R_{\uu}(w)^+$, i.e., $u$ is a concatenation of one or more return words to $w$.
	\end{lemma}
	\begin{proof}
		Since $\uu$ is uniformly recurrent and aperiodic, the index of every factor is finite.
		Let us remark that it is enough to consider a factor $u$ satisfying
		$\ind_\uu(u) = \max\{\ind_\uu(v) : v \in \LL(\uu)\ \text{and}\ |v| = |f|\}$.
		First, we describe a mapping which assigns to such $u$ a bispecial factor $w$.
		Let $k \in \N$ and $\alpha \in [0, 1)$ be such that $\ind_\uu(u) = k + \alpha$.
		Then $u$ can be written in the form $u = u'u''$, with $u' = u^\alpha$.
		Clearly, $k \geq 1$ and $u'' \neq \varepsilon$.
		Let $a$ and $b$ denote the first and the last letter of $u''$ respectively.
		Since $\uu$ is uniformly recurrent, there exist letters $x,y \in \A$ such that $x u^{k + \alpha} y = x (u'u'')^k u' y \in \LL(\uu)$.
		Obviously, $\uu$ contains both the factors $x (u' u'')^{k - 1} u'a$ and $b (u' u'')^{k - 1} u' y$.
		Let us observe that:
		\begin{itemize}
			\item $y \neq a$, otherwise $\uu$ would contain $(u' u'')^k u' a = u^e$, with $e > k + \alpha = \ind_\uu(u)$;
			\item $x \neq b $, otherwise $\uu$ would contain $v^k b u'$, with $v = bub^{-1}$ and $b u'$ a prefix of $v$, and that would imply $|v| = |bub^{-1}| = |u|$ and $\ind_\uu(v) > \ind_\uu(u)$.
		\end{itemize}
		Hence $w= (u' u'')^{k-1} u'$ is bispecial in $\uu$.
		Since $u^{k + \alpha} = u w\in \LL(\uu)$ and $w$ is a prefix of $uw$, the factor $u$ is a~concatenation of one or more return words to $w$ and $\ind_\uu(u) = \frac{|uw|}{|u|} = 1 + \frac{|w|}{|u|}$.
	\end{proof}
	
	\begin{theorem}
		\label{Prop_FormulaForCR}
		Let $\uu$ be a uniformly recurrent aperiodic sequence.
		Let $(w_n)$ be a sequence of all bispecial factors ordered by their length.
		For every $n \in \N$, let $v_n$ be a shortest return word to $w_n$ in $\uu$.
		Then
		$$
		\CR(\uu) = 1 + \sup\limits_{n \in \N} \left\{ \frac{|w_n|}{|v_n|} \right\}
		\qquad \text{and} \qquad
		\CR^*(\uu) = 1 + \limsup\limits_{n \to \infty}  \frac{|w_n|}{|v_n|} .
		$$
	\end{theorem}
	\begin{proof}
		By Lemma~\ref{Lem_URbispecial_retwords}, $v_nw_n = v_n^{e_n}$ for some exponent $e_n \in \Q$ and thus
		$\ind_\uu(v_n) \geq e_n = \frac{|v_nw_n|}{|v_n|} = 1 + \frac{|w_n|}{|v_n|}$.
		Hence
		$\CR(\uu) \geq 1 + \sup \left\{ \frac{|w_n|}{|v_n|} \right\} >1$.
		By the second statement of the same lemma, $\lim\limits_{n \to \infty} |v_n| = \infty$.
		Therefore,
		$\CR^*(\uu) \geq 1 + \limsup \frac{|w_n|}{|v_n|} \geq 1$.
		
		To show the opposite inequality, we distinguish two cases.
		
		We first assume that $\CR(\uu) = +\infty$.
		We find a sequence $(f_n)$ of factors of $\uu$ having the property $\ind_\uu(f_n) \to +\infty$.
		Since $\uu$ is uniformly recurrent, $|f_n| \to +\infty$ too.
		By Lemma~\ref{Lem_URbispecial_retwords2} for each $n \in \N$ there exists a bispecial factor $w_{k_n}$ in $\LL(\uu)$ and $u_{k_n} \in \R_\uu(w_{k_n})^+$ such that $\ind_\uu(f_n) \leq  1 + \frac{|w_{k_n}|}{|u_{k_n}|}$.
		Obviously, $\frac{|w_{k_n}|}{|u_{k_n}|} \leq \frac{|w_{k_n}|}{|v_{k_n}|}$.
		Hence,
		$$
		\CR^*(\uu)
		\leq
		\CR(\uu)
		=
		+\infty
		=
		\lim_{n \to \infty} \ind_\uu (f_n)
		\leq
		1 + \limsup_{n \to \infty} \frac{|w_{k_n}|}{|v_{k_n}|}
		\leq
		1 + \sup\limits_{n \in \N} \left\{ \frac{|w_n|}{|v_n|} \right\}\,.
		$$
		
		Now assume that $\CR(\uu) < +\infty$.
		Let $\delta > 0$ be such that $\CR(\uu) - \delta > 1$.
		Thus there exists $f \in \LL(\uu)$ satisfying $\CR(\uu) - \delta < \ind_\uu(f)$.
		Using Lemma~\ref{Lem_URbispecial_retwords2}, we find $u \in \LL(\uu)$ and a bispecial factor $w$ such that $\ind_\uu(f) \leq \ind_\uu(u) = 1 + \frac{|w|}{|u|}$, where $u \in \R_{\uu}(w)^+$.
		Therefore, for some index $m \in \N$, one has $w = w_m$ and $|u| \geq |v_m|$.
		Altogether, for an arbitrary positive $\delta$, we have
		$$
		\CR(\uu) - \delta
		<
		\ind_\uu(f)
		\leq
		\ind_\uu(u)
		=
		1 + \frac{|w|}{|u|}
		\leq
		1 + \frac{|w_m|}{|v_m|}
		\leq
		1 + \sup_{n \in \N} \left\{ \frac{|w_n|}{|v_n|} \right\}.
		$$
		Consequently, $\CR(\uu) \leq 1 + \sup\left\{ \frac{|w_n|}{|v_n|} \right\}$.
		
		If $\CR^*(\uu) = 1$, then the above proven inequality $\CR^*(\uu) \geq 1 + \limsup \frac{|w_n|}{|v_n|} \geq 1$ implies the second statement of the proposition.
		If $\CR^*(\uu) > 1$, then there exists a sequence of factors $(f^{(n)})$ of $\uu$ with  $\ind_\uu(f^{(n)}) > 1$ for every $n$, such that  $|f^{(n)}| \to \infty$ \ and \ $\ind_\uu(f^{(n)}) \to \CR^*(\uu)$.
		For each $n$, we find a factor $u^{(n)}$ and a bispecial factor $w^{(n)}$ with the properties given in Lemma~\ref{Lem_URbispecial_retwords2} and we proceed analogously as before.
	\end{proof}
	
	\begin{example}
		Let us consider the Fibonacci sequence $\ff$ from Example~\ref{ex:FiboDef}.
		We will recall later how to calculate the lengths of bispecial factors and their return words in Sturmian sequences (see Proposition~\ref{ParikhRSB1}).
		Let us use the same notation as in Theorem~\ref{Prop_FormulaForCR}. Then $|w_n|=F_{n+2}+F_{n+1}-2$ and $|v_n|=F_{n+1}$ with $F_0=0, F_1=1$.
		We therefore get
		$E(\ff) = 2 + \tau = 2 + \frac{1+\sqrt{5}}{2}$, which is in correspondence with the formula from~\cite{Cade} and \cite{DaLe}.
	\end{example}
	
	\begin{remark}
		\label{rem:supset}
		As already mentioned in Remark~\ref{rem:retwords_extension_to_BS}, every factor $f$ of a recurrent aperiodic sequence $\uu$ can be uniquely extended to the shortest bispecial factor $w = xfy \in \LL(\uu)$ and the lengths of return words to $f$ and to $w$ coincide.
		Clearly, if $|v|$ is the length of a shortest return word to $f$, then $\frac{|f|}{|v|} \leq \frac{|w|}{|v|}$.
		Therefore, in Theorem~\ref{Prop_FormulaForCR}, the sequence of bispecial factors $(w_n)$ can be replaced by any sequence of factors $(f_n)$ having $(w_n)$ as its subsequence.
	\end{remark}

	\section{Sturmian sequences}
	\label{sec:sturmian}
	
	Our aim is to use Theorem~\ref{Prop_FormulaForCR} for calculation of the (asymptotic) critical exponent of balanced sequences.
	Thanks to Proposition~\ref{pro:standardSturm} we can, without loss of generality, restrict our study to the colouring of standard Sturmian sequences.
	In order to determine the shortest return words to bispecial factors in balanced sequences, we need to list some important facts on a standard Sturmian sequence over the binary alphabet $\{ {\tt a,b} \}$.
	They are partially taken from~\cite{DvMePe} and~\cite{Lo}.
	
	First, we point out that a factor of a standard Sturmian sequence $\uu$ is bispecial if and only if it is a palindromic prefix of $\uu$.
	In particular, the first letter of $\uu$ is a bispecial factor.
	We adopt the convention of letting ${\tt b}$ denote the first letter of $\uu$.
	Obviously, the letter frequencies satisfy $\rho_{\tt b} > \rho_{\tt a}$.
	We also use the convention that the first component of the Parikh vector of a factor of $\uu$ corresponds to the least frequent letter and the second component to the most frequent letter of the sequence $\uu$ (even when we consider standard Sturmian sequences over binary alphabets other than $\{ {\tt a}, {\tt b} \}$).
	
	Second, we point out that the derived sequence of a standard Sturmian sequence $\uu$ to any factor $u \in \LL(\uu)$ is a standard Sturmian sequence as well.
	Let ${\tt r}$ and ${\tt s}$ denote the letters of $\dd_\uu(u)$ coding the occurrences of the two return words to $u$ in $\uu$.
	We adopt the convention that the first letter in $\dd_\uu(u)$ is ${\tt r}$.
	In particular, it means that the letter frequencies in $\dd_\uu(u)$ satisfy $\rho_{\tt r} > \rho_{\tt s}$ and ${\tt r}$ is a coding of the first return word to $u$ appearing in $\uu$.
	This return word to $u$ in $\uu$ will be denoted by $r$ and called the most frequent return word to $u$ in $\uu$.
	The other return word to $u$ in $\uu$ will be denoted by $s$ and called the least frequent return word to $u$ in $\uu$.
	In particular, $r$ and $s$ are factors of $\uu$ and ${ \tt r} $ and ${\tt s}$ are letters of $\dd_\uu(u)$.
	
	Third, we use the characterisation of standard Sturmian sequences by their directive sequences.
	To introduce them, we define the two morphisms
	$$
	G = \
	\left\{ \,
	\begin{aligned}
		{\tt a} & \to {\tt a} \\
		{\tt b} & \to {\tt ab} \,
	\end{aligned}
	\right.
	\quad \text{and} \quad
	D = \
	\left\{ \,
	\begin{aligned}
		{\tt a} & \to {\tt ba} \\
		{\tt b} & \to {\tt b} \,
	\end{aligned}
	\right. .
	$$
	
	\begin{proposition}[\cite{JuPi}]
		\label{Lem_Standard}
		For every standard Sturmian sequence $\uu$ there is a uniquely given sequence ${\mathbf \Delta} = \Delta_0 \Delta_1 \Delta_2 \cdots \in \{ G, D \}^\N$ of morphisms and a sequence $(\uu^{(n)})$ of standard Sturmian sequences such that
		$$
		\uu = {\Delta_0 \Delta_1 \cdots \Delta_{n-1}} \left( \uu^{(n)} \right) \, \ \text{for every } \ n \in \N\,.
		$$
		The sequence ${\mathbf \Delta}$, called the \emph{directive sequence} of $\uu$, contains infinitely many morphisms $G$ and infinitely many morphisms $D$.
		
		If moreover ${\tt b}$ is the most frequent letter in $\uu$, then $\Delta_0 = D$ and the directive sequence can be written in the form
		$
		{\mathbf \Delta} = D^{a_1} G^{a_2} D^{a_3} G^{a_4} \cdots
		$
		for some sequence $(a_i)_{i \geq 1}$ of positive integers.
	\end{proposition}
	
	We associate an irrational number $\theta$ with the directive sequence ${\mathbf \Delta}$ as follows:
	$$
	\theta =
	\theta(\uu): =
	[0, a_1,a_2, a_3 , \ldots].
	$$
	It was shown in~\cite{Lo} that $\theta = \frac{\rho_{\tt a}}{\rho_{\tt b}}$ and it is usually called the {\em slope} of $\uu$. 
	\begin{remark}
		\label{rem:prefix_u}
		If $\uu$ has directive sequence ${\mathbf \Delta} = D^{a_1} G^{a_2} \cdots$, then $\uu$ is a concatenation of blocks $ {\tt b}^{a_1}{\tt a}$ and ${\tt b}^{a_1+1}{\tt a}$.
		Moreover, ${\tt b}^{a_1}{\tt a}$ is a prefix of $\uu$.
		Indeed, Proposition~\ref{Lem_Standard} gives $\uu = D^{a_1}(\uu')$, where $\uu'$ is the image of a standard Sturmian sequence under $G$.
		The form of $G$ implies that $\uu'$ is a concatenation of the blocks ${\tt a}$ and ${\tt ab}$.
		Their images under $D^{a_1}$ are ${\tt b}^{a_1}{\tt a}$ and ${\tt b}^{a_1+1}{\tt a}$ respectively.
	\end{remark}
	
	\begin{example}
		\label{ex:illustrative_prefix}
		Let $\uu$ be a standard Sturmian sequence with slope $\theta = [0, 3,2, \overline{3, 1}].$
		Then by Remark~\ref{rem:prefix_u} the word $D^3 G^2 D^3 G^1({\tt bbba})$ is a prefix of $\uu$, i.e.,
		$$
		\uu =
		{\tt b^3 a b^3 a b^4 a b^3 a b^4 a b^3 a b^4 a b^3 a b^3 a b^4 a b^3 a b^4 a b^3 a b^4 a b^3 a b^4 a b^3 a b^3 a b} \cdots .$$
	\end{example}
	
	The Parikh vectors of bispecial factors in $\uu$ and the corresponding return words can be easily expressed using the convergents $\frac{p_N}{q_N}$ to $\theta$.
	Let us recall that the nominator $p_N$ and the denominator $q_N$ of the $N^{\text{th}}$ \textit{convergent} to $\theta$ satisfy for all $N \ge 1$ the recurrence relation
	\begin{equation}
		\label{eq:convergents}
		X_{N} = a_{N} X_{N-1} +X_{N-2},
	\end{equation}
	but they differ in their initial values:
	$p_{-1} = 1, p_0 = 0$;
	and
	$q_{-1} = 0, q_0 = 1$.
	In the following we will also consider the quantity $Q_N := p_N + q_N$.
	The notation we use corresponds to the fact that $Q_N$ are denominators of the convergents to the number $\frac{1}{1+\theta}$, which is the frequency of the letter $\tt b$. Obviously, $Q_N$ satisfy the recurrence relation~\eqref{eq:convergents} with the initial values $Q_{-1} = 1$ and $Q_{0} = 1$.
	
	\begin{proposition}[\cite{DvMePe}]
		\label{ParikhRSB1}
		Let $\theta = [0, a_1,a_2, a_3, \ldots]$ be the slope of a standard Sturmian sequence $\uu$ and $b$ a bispecial factor of $\uu$.
		Let $r$ (resp., $s$) denote the return word to $b$ which is (resp., is not) a prefix of $\uu$.
		Then
		\begin{enumerate}
			\item there exists a unique pair $(N,m) \in \N^2$ with $0 \le m < a_{N+1}$ such that the Parikh vectors of $r$ and $s$ and $b$ are respectively
			$$
			\Parickh(r) = \begin{pmatrix} p_N \\ q_N \end{pmatrix},
			\; \;
			\Parickh(s) = \begin{pmatrix} m \, p_{N} + p_{N-1} \\ m \, q_{N} + q_{N-1} \end{pmatrix}
			\ \ \text{and} \ \
			\Parickh(b) = \Parickh(r) + \Parickh(s) - \begin{pmatrix} 1 \\ 1 \end{pmatrix};
			$$
			\item the slope of the derived sequence $\dd_\uu(b)$ to $b$ in $\uu$ is
			${\theta}' = [0, a_{N+1} - m, a_{N+2}, a_{N+3}, \ldots]$.
		\end{enumerate}
	\end{proposition}
	
	\begin{remark}
		\label{rem:N,m}
		A list of the bispecial factors of a Sturmian sequence $\uu$  ordered by their length starts with $\varepsilon$, ${\tt b}$, etc.
		For each $n \in \N$, the pair assigned to the $n^{\text{th}}$ bispecial factor is the unique pair $(N,m)$ satisfying $n = a_0 + a_1 + \cdots + a_N + m$ and $0 \leq m < a_{N+1}$, where we put $a_0 = 0$.
	\end{remark}
	
	\begin{example}
		\label{ex:ParikhRetwords}
		Let $\uu$ be the sequence
		$$
		\uu = {\tt b^3 a b^3 a b^4 a b^3 a b^4 a b^3 a b^4 a b^3 a b^3 a b^4 a b^3 a b^4 a b^3 a b^4 a b^3 a b^4 a b^3 a b^3 a b }\cdots
		$$
		from Example~\ref{ex:illustrative_prefix} having slope $\theta = [0, 3, 2, \overline{3, 1}]$.
		The value of $\theta$ calculated from its continued fraction expansion is $\theta = \frac{39 + \sqrt{21}}{150} \doteq 0,29$.
		
		Since $\uu$ is a standard Sturmian sequence, each bispecial factor is a palindromic prefix of $\uu$.
		Thus the bispecial factors ordered by their length are
		$$
		b_0 = \varepsilon,\ \
		b_1 = {\tt b}, \ \
		b_2 = {\tt b}^2, \ \
		b_3 = {\tt b}^3, \ \
		b_4 = {\tt b}^3{\tt ab}^3, \ \
		b_5 = {\tt b}^3{\tt ab}^3 {\tt ab}^3, \ \
		\text{etc.}
		$$
		\begin{itemize}
			\item Consider the bispecial factor $ b_2 = {\tt b^2}$.
			The pair associated with $b_2$ is $(N,m) = (0,2)$.
			The prefix return word to $\tt b^2$ is $r = {\tt b}$, the non-prefix return word is $s = {\tt b^2 a}$ and their Parikh vectors are in correspondence with Proposition~\ref{ParikhRSB1} since
			$$
			\Parickh(r) = \begin{pmatrix} p_0 \\ q_0 \end{pmatrix} = \begin{pmatrix} 0 \\ 1 \end{pmatrix}
			\quad \text{and} \quad
			\Parickh(s) = \begin{pmatrix} 2p_0+p_{-1} \\ 2q_0+q_{-1} \end{pmatrix} = \begin{pmatrix} 1 \\ 2 \end{pmatrix}.
			$$
			The slope $\theta'$ of the derived sequence $\dd_\uu(b_2)$ is
			${\theta}' = [0, 1, 2,\overline{3, 1}]$.
			Hence $\tfrac{1}{\theta'} = \frac{1}{\theta} - 2$, i.e. $ \theta'  \doteq 0,694$.
			
			\item Consider $b_3 = {\tt b^3}$.
			The associated pair is $(N,m) = (1,0)$.
			The prefix return word to $b_3$ is $r = {\tt b^3a}$ and the non-prefix return word is $s = {\tt b}$.
			It is in correspondence with Proposition~\ref{ParikhRSB1} since
			$$
			\Parickh(r) = \begin{pmatrix} p_1 \\ q_1 \end{pmatrix} = \begin{pmatrix} 1 \\ 3 \end{pmatrix}
			\quad \text{and} \quad
			\Parickh(s) = \begin{pmatrix} p_{0} \\ q_0 \end{pmatrix} = \begin{pmatrix} 0 \\ 1 \end{pmatrix}.
			$$
			The slope $\theta'$ of $\dd_\uu(b_3)$ is ${\theta}' = [0,  2,\overline{3, 1}]$.
			Hence $\theta' = \frac{1}{\theta} - 3 \doteq 0,442$.
			
			\item Consider $ b_4 = {\tt b^3ab^3}$.
			The associated pair is $(N,m) = (1,1)$.
			The Parikh vectors of the return words $r = {\tt b^3a}$ and $s = {\tt b^3ab}$ are
			$$
			\Parickh(r) = \begin{pmatrix} p_1 \\ q_1 \end{pmatrix} = \begin{pmatrix} 1 \\ 3 \end{pmatrix}
			\quad \text{and} \quad
			\Parickh(s) = \begin{pmatrix} p_{1} + p_0 \\ q_1 +q_0 \end{pmatrix} = \begin{pmatrix} 1 \\ 4 \end{pmatrix}.
			$$
			The slope $\theta'$ of $\dd_\uu(b_4)$ is ${\theta}' = [0, 1, \overline{3, 1}]$.
			This gives $\theta' = \frac{-3+\sqrt{21}}{2} \doteq 0,791$.
		\end{itemize}
	\end{example}
	
	Since every factor of a Sturmian sequence has exactly two return words, every piece of $\uu$ between two occurrences of $u$ is a concatenation of these two return words.
	This implies the following observation.
	
	\begin{Observation}
		\label{observation}
		Let $r$ and $s$ be respectively the most and the least frequent return word to $u$ in $\uu$.
		If $fu \in \LL(\uu)$ and $u$ is a prefix of $fu$, then $\Parickh(f) = k \Parickh(r) + \ell \Parickh(s)$, where $\left( \begin{smallmatrix} \ell \\ k \end{smallmatrix} \right)$ is the Parikh vector of a factor of the derived sequence $\dd_\uu(u)$.
	\end{Observation}
	
	The Parikh vectors of factors occurring in a given Sturmian sequence $\uu$ are fully characterised by the slope $\theta$ of $\uu$.
	
	\begin{lemma}
		\label{lem_kl}
		Let ${\tt b}$ be the most frequent letter of a Sturmian sequence $\uu$ and $\theta$ the slope of $\uu$.
		Then $\uu$ contains a factor $w$ such that $|w|_{\tt b} = k$ and $|w|_{\tt a} = \ell$ if and only if
		\begin{equation}
			\label{pocet}
			(k-1) \theta - 1 < \ell < (k+1) \theta + 1 \ \ \text{and} \ k, \ell \in \N.
		\end{equation}
	\end{lemma}
	\begin{proof}
		As $\theta = \frac{\rho_{\tt a}}{\rho_{\tt b}}$ and $1 = \rho_{\tt a} + \rho_{\tt b}$, the density of the letter ${\tt b}$ in $\uu$ is $\rho_{\tt b} = \frac{1}{1 + \theta} \notin \Q$.
		For every length $n \in \N$, with $n > 0$, there exist factors $u^{(1)}$ and $u^{(2)}$ of length $n$ such that $|u^{(1)}|_{\tt b} > n \rho_{\tt b} > |u^{(2)}|_{\tt b}$.
		Since $\uu$ is balanced, necessarily
		$|u^{(1)}|_{\tt b} - |u^{(2)}|_{\tt b} \leq 1$.
		Therefore,
		\begin{equation}
			\label{pocet2}
			|u^{(1)}|_{\tt b} = \lceil n \rho_{\tt b} \rceil
			\quad \text{ and } \quad
			|u^{(2)}|_{\tt b} = \lfloor n \rho_{\tt b} \rfloor\,.
		\end{equation}
		Equation~\eqref{pocet} can be rewritten as
		\begin{equation}
			\label{pocet3}
			\frac{1}{1 + \theta}(k + \ell) - 1
			< k <
			\frac{1}{1 + \theta}(k + \ell) +1\,.
		\end{equation}
		Let us write $n = k + \ell$.
		Then Equation~\eqref{pocet3} is equivalent to
		$$
		\rho_{\tt b} n - 1
		< k <
		\rho_{\tt b} n + 1.
		$$
		Since $\rho_{\tt b}$ is irrational, we can write
		$$
		\lfloor n \rho_{\tt b} \rfloor
		=
		\lceil n \rho_{\tt b} \rceil - 1
		\leq k \leq
		\lfloor n \rho_{\tt b} \rfloor + 1
		=
		\lceil n \rho_{\tt b} \rceil\,.
		$$
		Equation~\eqref{pocet2} says that either $u^{(1)}$ or $u^{(2)}$ is a factor of length $n = k + \ell$ containing $k$ times the letter ${\tt b}$ and, consequently, $\ell$ times the letter ${\tt a}$.
	\end{proof}

	\section{Shortest return words to bispecial factors in  balanced sequences}
	\label{sec:shortest}
	
	As seen in Proposition~\ref{ParikhRSB1}, the length of the return words to factors of a Sturmian sequence $\uu$ is well-known.
	The aim of this section is to find a formula for the length of the shortest return words to factors of a colouring of $\uu$.
	
	As occurrences of a factor $u$ in a Sturmian sequence $\uu$ and occurrences of factors from $\pi^{-1}(u)$ in every colouring of $\uu$ coincide, we are able to give a formula based on the knowledge of the length of return words in $\uu$.
	Proposition~\ref{zeroParikh1} and Observation~\ref{observation} justify the following definition.
	
	\begin{definition}
		\label{def:setS}
		Let $\vv = \barva(\uu, \yy, \yy')$.
		Let  $u \in \LL(\uu)$ and  $r$ be the most and $s$ the least frequent return word to $u$ in $\uu$.
		We write
		$\S(u) = \S_1(u) \cap \S_2(u) \cap \S_3\,,$
		where
		$$
		\begin{array}{l}
			\S_1(u) =
			\left\{ \begin{pmatrix} \ell \\ k \end{pmatrix} : \begin{pmatrix} \ell \\ k \end{pmatrix} \text{ is the Parikh vector of a factor of} \ \dd_{\uu}(u) \right\}; \\
			\S_2(u) = \displaystyle
			\bigcup_{n \in {\rm{gap}}(\yy, |u|_{\tt a}) } \;
			\bigcup_{n' \in {\rm{gap}}(\yy', |u|_{\tt b}) }
			\left\{ \begin{pmatrix} \ell \\ k \end{pmatrix} : \ k \Parickh(r) + \ell \Parickh(s) \equiv \begin{pmatrix} 0 \\ 0 \end{pmatrix} \pmod {\begin{pmatrix} n \\ n' \end{pmatrix}} \right\}; \\
			\S_3 = \, \left\{ \begin{pmatrix} \ell \\ k \end{pmatrix} : 1 \leq k + \ell \leq {\rm Per}(\yy) {\rm Per}(\yy') \right\}.
		\end{array}
		$$
	\end{definition}
	
	\begin{remark}
		\label{longBS}
		The formula defining $\S_2(u)$ can be simplified when   $|u|_{\tt a} > \bs(\yy)$ and $|u|_{\tt b} > \bs(\yy')$.
		For such a factor $u$, one has ${\rm{gap}}(\yy, |u|_{\tt a}) = \{ \rm{Per}(\yy) \}$ and ${\rm{gap}}(\yy', |u|_{\tt b}) = \{ \rm{Per}(\yy') \}$ and thus
		$$
		\S_2(u) =
		\left\{ \begin{pmatrix} \ell \\ k \end{pmatrix} : \ k \Parickh(r) + \ell \Parickh(s) \equiv \begin{pmatrix} 0 \\ 0 \end{pmatrix} \pmod {\begin{pmatrix} {\rm Per}(\yy) \\ {\rm Per}(\yy') \end{pmatrix}} \right\}\,.
		$$
		Consequently, for a sufficiently long word $u$, the set $\S(u)$ depends only on $\rm{Per}(\yy)$ and $\rm{Per}(\yy')$ and does not depend on the structure of $\yy$ and $\yy'$ themselves.
	\end{remark}
	
	Lemma~\ref{lem_kl} helps us to recognise which vector is the Parikh vector of a factor of a given Sturmian sequence.
	This is important to decide whether $\left( \begin{smallmatrix} \ell \\ k \end{smallmatrix} \right)$ belongs to $\S_1(u)$.
	
	\begin{example}
		\label{ex:Lubka_S(b^3)}
		Let us colour the sequence $\uu$ from Example~\ref{ex:illustrative_prefix} by the constant gap sequences $\yy = ({\tt 01})^{\omega}$ and $\yy'=({\tt 234235})^{\omega}$.
		We get a balanced sequence $\vv = \barva{(\uu,\yy,\yy')}$.
		\begin{enumerate}
			
			\item Consider the bispecial factor ${\tt b^3}$ of $\uu$.
			Let us examine the set $\S({\tt b^3})$.
			Using Example~\ref{ex:ParikhRetwords}, we know the slope $\theta'\doteq 0.442$ of $\dd_\uu({\tt b^3})$ and the Parikh vectors $\Parickh(r)$ and $\Parickh(s)$.
			Moreover, $\gap{\yy}{|{\tt b^3}|_{\tt a}} = \gap{\yy}{0} = \{ 1 \}$ and $\gap{\yy'}{|{\tt b^3}|_{\tt b}} = \gap{\yy'}{3} = \{ 6 \}$.
			Thus $\left( \begin{smallmatrix} \ell \\ k \end{smallmatrix} \right) \in \S(\tt b^3)$ satisfies
			$$
			(k-1) \theta' - 1 <
			\ell <
			(k+1) \theta' + 1
			\ \ \text{and} \
			k, \ell \in \N;
			$$
			$$
			k \Parickh(r) + \ell \Parickh(s) =
			k \begin{pmatrix} 1 \\ 3 \end{pmatrix} + \ell \begin{pmatrix} 0 \\ 1 \end{pmatrix}
			\equiv
			\begin{pmatrix} 0 \\ 0 \end{pmatrix}
			\pmod
			{\begin{pmatrix} 1 \\ 6 \end{pmatrix}};
			$$
			$$
			1 \leq k + \ell \leq 12.
			$$
			Examining the three conditions above, we get
			$$
			\S({\tt b^3}) =
			\left\{
			\begin{pmatrix} 0\\ 2 \end{pmatrix},
			\begin{pmatrix} 3 \\ 5 \end{pmatrix},
			\begin{pmatrix} 3 \\ 7 \end{pmatrix},
			\begin{pmatrix} 3 \\ 9 \end{pmatrix}
			\right\}.
			$$
			
			\item Let us consider the factor ${\tt ab}$ of $\uu$, which is not bispecial.
			Let us examine the set $\S({\tt ab})$.
			The shortest bispecial factor containing ${\tt ab}$ is ${\tt b^3 a b^3}$.
			By Remark~\ref{rem:retwords_extension_to_BS} the derived sequences $\dd_\uu({\tt ab}) = \dd_\uu({\tt b^3ab^3})$, and the Parikh vectors of the corresponding return words coincide.
			In Example~\ref{ex:ParikhRetwords} we determined the slope $\theta'\doteq 0.791$ of $\dd_\uu({\tt b^3ab^3})$ and the Parikh vectors
			$\Parickh(r)$ and $\Parickh(s)$.
			Moreover, $\gap{\yy}{|{\tt ab}|_{\tt a}} = \gap{\yy}{1} = \{ 2 \}$ and $\gap{\yy'}{|{\tt ab}|_{\tt b}} = \gap{\yy'}{1} = \{ 3,6 \}$.
			Thus $\left( \begin{smallmatrix} \ell \\ k \end{smallmatrix} \right) \in \S({\tt ab})$ satisfies
			$$
			(k-1) \theta' - 1 < \ell < (k+1) \theta' + 1
			\ \ \text{and} \
			k, \ell \in \N;
			$$
			$$
			k \Parickh(r) + \ell \Parickh(s) =
			k \begin{pmatrix} 1 \\ 3 \end{pmatrix} + \ell \begin{pmatrix} 1 \\ 4 \end{pmatrix}
			\equiv
			\begin{pmatrix} 0 \\ 0 \end{pmatrix}
			\pmod
			{\begin{pmatrix} 2 \\ 3\ \text{or} \ 6 \end{pmatrix}};
			$$
			$$
			1 \leq k + \ell \leq 12.
			$$
			The second condition reduces to $k+\ell \equiv 0 \pmod 2$ and $\ell \equiv 0 \pmod 3$. 
			Examining the three conditions above, we get
			$$
			\S({\tt ab}) =
			\left\{
			\begin{pmatrix} 0 \\ 2 \end{pmatrix},
			\begin{pmatrix} 3 \\ 3 \end{pmatrix},
			\begin{pmatrix} 3 \\ 5 \end{pmatrix},
			\begin{pmatrix} 6 \\ 6 \end{pmatrix}
			\right\}.
			$$
		\end{enumerate}
	\end{example}
	
	Using the formula provided in Theorem~\ref{Prop_FormulaForCR}, we can treat all bispecial factors of the same length simultaneously.
	
	\begin{theorem}
		\label{prop:ShortestReturn}
		Let $\vv = \barva( \uu, \yy, \yy')$ and $u \in \LL(\uu)$.
		The shortest words in the set
		$\{ v : v \in  \R_{\vv}(w) \text{ and } \pi(w) = u \}$
		have length
		$$
		|v|
		=
		\min \left\{ k |r| + \ell |s| : \begin{pmatrix} \ell \\ k \end{pmatrix} \in \S(u) \right\}.
		$$
	\end{theorem}
	\begin{proof}
		First, let us show that the length of every return word in $\vv$ to a factor from $\pi^{-1}(u)$ is contained in the set $\left\{ k |r| + \ell |s| : \left( \begin{smallmatrix} \ell \\ k \end{smallmatrix} \right) \in \S_1 (u) \cap \S_2(u) \right\}$.
		By Proposition~\ref{zeroParikh1} and Observation~\ref{observation}, a vector $\left( \begin{smallmatrix} \ell \\ k \end{smallmatrix} \right)$ belongs to $\S_1(u) \cap \S_2(u)$ if and only if $k \Parickh(r) + \ell \Parickh(s)$ is the Parikh vector of $\pi(v)$, where $v$ is a factor between two (possibly not consecutive) occurrences of a factor $w \in \pi^{-1}(u)$ in $\vv$.
		Obviously, the length of $v$ is $k |r| + \ell |s|$.
		It is evident that if we consider above $|v| = \min\{ k |r| + \ell |s| \}$, where $\left( \begin{smallmatrix} \ell \\ k \end{smallmatrix} \right) \in \S_1(u) \cap \S_2(u)$, then $v$ is a return word to a factor $w \in \pi^{-1}(u)$.
		
		To finish the proof, we have to show that the minimum value of $|v|$ is attained for $k$ and $\ell$ satisfying $1 \leq k + \ell \leq {\rm Per}(\yy) {\rm Per}(\yy')$.
		Let $\left( \begin{smallmatrix} \ell \\ k \end{smallmatrix} \right) \in \S_1(u) \cap \S_2(u)$ and $k + \ell > {\rm Per}(\yy) {\rm Per}(\yy')$.
		Thus $\Parickh(d) = \left( \begin{smallmatrix} \ell \\ k \end{smallmatrix} \right)$ for some $d = d_1 d_2 d_3 \cdots d_{k+\ell} \in \LL(\dd_\uu(u))$.
		For every $i = 1, 2, \ldots, k + \ell$, we write $\left( \begin{smallmatrix} \ell_i \\ k_i \end{smallmatrix} \right) = \Parickh(d_1 d_2 \cdots d_i)$.
		We assign to each $i$ the vector $X_i = k_i \Parickh(r) + \ell_i \Parickh(s)$.
		Since the number of equivalence classes $\bmod \left( \begin{smallmatrix} n \\ n' \end{smallmatrix} \right)$ is $n n' \leq {\rm Per}(\yy) {\rm Per}(\yy')$, there exist $i, j$ with $1 \leq i <  j \leq k + \ell$ such that $X_i \equiv X_j \pmod {\left( \begin{smallmatrix} n \\ n' \end{smallmatrix} \right)}$.
		Let $\left( \begin{smallmatrix} \ell' \\ k' \end{smallmatrix} \right)$ be the Parikh vector of $d_{i+1} d_{i+2} \cdots d_j$.
		Obviously, $\left( \begin{smallmatrix} \ell' \\ k' \end{smallmatrix} \right) \in \S_1(u)$, $1 \leq j-i = k' + \ell' < k + \ell$ and $k' \leq k$ and $\ell' \leq \ell$.
		Hence $k'|r| + \ell'|s| <  k|r| + \ell|s|$.
		Since $k' \Parickh(r) + \ell' \Parickh(s) = X_j - X_i \equiv \left( \begin{smallmatrix} 0 \\ 0 \end{smallmatrix} \right) \pmod {\left( \begin{smallmatrix} n \\ n' \end{smallmatrix} \right)}$, the vector $\left( \begin{smallmatrix} \ell' \\ k' \end{smallmatrix} \right) \in \S_2(u)$.
		Therefore, the minimum length cannot be achieved for $k + \ell > {\rm Per}(\yy) {\rm Per}(\yy')$.
	\end{proof}
	
	\begin{example}
		\label{ex:Lubka_shortest_retword}
		Let us consider the sequence $\vv = \barva(\uu, \yy, \yy')$ as given in Example~\ref{ex:Lubka_S(b^3)}.
		Let us write down a prefix of $\vv$,
		$$
		\vv = {\tt 234 0 235 1 2342 0 352 1 3423 0 523 1 4235 0 234 1 235 0 2342 1 352}\cdots .
		$$
		Using Theorem~\ref{prop:ShortestReturn}, we find the length of the shortest word in the set $\{ v : v \in  \R_{\vv}(w) \text{ and } \pi(w) = {\tt b^3} \}$.
		The set $\S(\tt b^3)$ was examined in Example~\ref{ex:Lubka_S(b^3)}.
		We have
		$$
		|v|
		=
		\min \left\{ 4k  + \ell : \begin{pmatrix} \ell \\ k \end{pmatrix} \in \left\{
		\begin{pmatrix} 0 \\ 2 \end{pmatrix},
		\begin{pmatrix} 3 \\ 5 \end{pmatrix},
		\begin{pmatrix} 3 \\ 7 \end{pmatrix},
		\begin{pmatrix} 3 \\ 9 \end{pmatrix}
		\right\} \right\} = 8.
		$$
		Indeed, for instance the prefix ${\tt 234}$ with projection $\pi({\tt 234}) = {\tt b^3}$ has the return word ${\tt 234 0 235 1}$ of length $8$.
	\end{example}
	
	\begin{remark}
		\label{rem:tildeS}
		In Theorem~\ref{prop:ShortestReturn} instead of $\S(u)$ it is sufficient to consider the set $\hat{\S}(u)$ containing all integer vectors of the form $\left( \begin{smallmatrix} \ell \\ k \end{smallmatrix} \right)$ such that $\left( \begin{smallmatrix} \ell \\ k \end{smallmatrix} \right)\in \S(u)$ and no other vector $\left( \begin{smallmatrix} \ell' \\ k' \end{smallmatrix} \right) \in \S(u)$ satisfies $\ell' \leq \ell$ and $k' \leq k$.
		This follows from the fact that for $\ell' \leq \ell$ and $k' \leq k$, we have $k'|r|+\ell'|s| \leq k|r|+\ell|s|$.
		For instance, in Example~\ref{ex:Lubka_shortest_retword} it was sufficient to consider $\hat{\S}({\tt b^3}) = \{ \left( \begin{smallmatrix} 0 \\ 2 \end{smallmatrix} \right) \}$.
	\end{remark}
	
	If a projection of a bispecial factor $w$ in $\vv$ is bispecial in $\LL(\uu)$, we can deduce an explicit formula for $1 + \tfrac{|w|}{|v|}$, where $|v|$ is the length of a shortest return word to $w$ in $\vv$.
	These values are crucial for the computation of $E(\vv)$ and $E^*(\vv)$.
	
	The following statement is a direct consequence of Proposition~\ref{ParikhRSB1} and Theorem~\ref{prop:ShortestReturn}. Recall that $Q_N=p_N+q_N$, see~\eqref{eq:convergents}.
	
	\begin{corollary}
		\label{indexlong}
		Let $\vv = \barva(\uu, \yy, \yy')$ and $\left( \tfrac{p_N}{q_N} \right)$ be the sequence of convergents to the slope $\theta$ of $\uu$.
		Let $b \in \LL(\uu)$ be a bispecial factor and $(N,m)$ be the pair assigned in Proposition~\ref{ParikhRSB1} to $b$.
		Then a shortest word in the set $\{ v : v \in \R_{\vv}(w) \text{ and } \pi(w) = b \}$ satisfies
		$$
		1 + \frac{|w|}{|v|} =
		1 + \max \left\{ \frac{(1+m) Q_N+Q_{N-1} - 2}{(k+ \ell m) Q_{N} + \ell Q_{N-1} } : \begin{pmatrix} \ell \\ k \end{pmatrix} \in \S(b) \right\}\,.
		$$
	\end{corollary}
	
	Due to Theorem~\ref{Prop_FormulaForCR}, the formula given in the previous corollary plays an important  role in computation of the (asymptotic) critical exponent.
	Therefore we introduce the notation
	\begin{equation}
		\label{modified}
		I(N,m)
		:=
		1 + \max \left\{ \frac{1 + m + \tfrac{Q_{N-1}-2}{Q_N}}{k + \ell m + \ell\tfrac{Q_{N-1}}{Q_N} } : \begin{pmatrix} \ell \\ k \end{pmatrix} \in \S(b) \right\}\,.
	\end{equation}
	
	\begin{example}
		\label{ex:Lubka_I(N,m)}
		Let us consider the sequence $\vv$ as in Example~\ref{ex:Lubka_shortest_retword}.
		Let us determine $I(1,0)$ from Equation~\eqref{modified}.
		We already know that the pair $(1,0)$ is associated with the bispecial factor ${\tt b^3}$ of $\uu$.
		The first values of $Q_N = p_N + q_N$ are $Q_0 = 1$ and $Q_1 = 4$.
		$$
		I(1,0)
		=
		1 + \max \left\{ \frac{1 + \tfrac{Q_{0}-2}{Q_1}}{k + \ell\tfrac{Q_{0}}{Q_1} } : \begin{pmatrix} \ell \\ k \end{pmatrix} \in \hat{\S}({\tt b^3}) = \left\{ \begin{pmatrix} 0 \\ 2 \end{pmatrix} \right\} \right\}
		=
		1 + \frac{3}{8}.
		$$
		This is in correspondence with Example~\ref{ex:Lubka_shortest_retword}, where we have seen that a shortest return word $v$ to a factor $w$ with projection $\pi(w) = {\tt b^3}$ is of length $8$.
		Thus $1 + \frac{|w|}{|v|} = 1 + \frac{3}{8}$.
	\end{example}

	\section{Bounds on (asymptotic) critical exponent\ of balanced sequences}
	\label{sec:bounds}
	
	The critical and asymptotic critical exponent of a Sturmian sequence $\uu$ can be computed from the continued fraction expansion of the slope $\theta = [0, a_1, a_2, \ldots]$ of $\uu$ by the following formulae (see~\cite{Cade} and \cite{DaLe}):
	$$
	E(\uu)
	=
	2 + \sup \left\{ a_{N+1} + \frac{Q_{N-1} - 2}{Q_N} \; : \; N \in \N \right\}
	\qquad \mbox{and}
	$$
	$$
	E^*(\uu)
	=
	2 + \limsup_{N \to \infty} \left( a_{N+1} + \frac{Q_{N-1} }{Q_N} \right)\,.
	$$
	Hence, the (asymptotic) critical exponent of a Sturmian sequence is finite if and only if the sequence of coefficients in the continued fraction expansion of $\theta$ is bounded.
	The projection $\uu$ of a balanced sequence $\vv$ satisfies
	$E(\uu)\geq E(\vv)$ \ and \ $E^*(\uu)\geq E^*(\vv)$.
	In other words, if the critical exponent of a Sturmian sequence $\uu$ is finite, then every colouring of $\uu$ has a finite critical exponent as well.
	The following lemma shows the opposite implication.
	
	\begin{lemma}
		\label{lem:bigCoeff}
		Let $\vv = \barva(\uu, \yy, \yy')$ be a balanced sequence and $\theta = [0, a_1, a_2, a_3, \ldots ]$ be the slope of $\uu$.
		If there exists $N \in \N$, $N > 1$ such that $1 + a_{N+1} \geq P$, where $P = \lcm \left\{ {\rm Per}(\yy), {\rm Per}(\yy') \right\}$, then for some bispecial factor $w$ of $\vv$ and a return word $v$ to $w$ one has
		$$
		1 + \frac{|w|}{|v|}
		\geq
		\frac{2 + a_{N+1}}{P} + \frac{Q_{N-1} - 2}{P Q_{N}}
		\geq
		\frac{2 + a_{N+1}}{P}\,.
		$$
	\end{lemma}
	\begin{proof}
		Put $m = a_{N+1} - P+1 $ and consider the bispecial factor $b$ of $\uu$ to which the pair $(N,m)$ is assigned by Proposition~\ref{ParikhRSB1}.
		By the same proposition, the slope $\theta'$ of the derived sequence $\dd_\uu(b)$ equals ${\theta}' = [0, P-1, a_{N+2}, a_{N+3}, \ldots]$.
		In particular $\frac{1}{\theta'} > P-1$.
		
		As $P$ is divisible by every $n \in {\rm{gap}}(\yy, |b|_{\tt a})$ and every $n' \in {\rm{gap}}(\yy', |b|_{\tt b})$, by Observation~\ref{obs:per}, the vector $\left( \begin{smallmatrix} \ell \\ k \end{smallmatrix} \right) = \left( \begin{smallmatrix} 0 \\ P \end{smallmatrix} \right)$ obviously belongs to $\S_2(b) \cap \S_3$.
		To show that it belongs to $\S_1(b)$ as well, we have to check, by Lemma~\ref{lem_kl}, that $(P-1) \theta' - 1 < 0 < (P+1) \theta' + 1$ or, equivalently, that $P - 1 < \tfrac{1}{\theta'}$.
		
		Let $w$ be a bispecial factor of $\vv$ such that $\pi(w) = b$ and $v$ be a shortest return word to $w$ in $\vv$.
		By Corollary~\ref{indexlong}, we have
		$$
		1 + \frac{|w|}{|v|}
		\geq
		1 + \frac{(2+a_{N+1}-P) Q_N+Q_{N-1} - 2}{P Q_{N}}\,.
		$$
	\end{proof}
	
	We can thus deduce the following result.
	
	\begin{corollary}
		\label{finiteAsympt}
		Let $\vv = \barva(\uu, \yy, \yy')$.
		Then $E^*(\vv)$ is finite if and only if $E^*(\uu)$ is finite.
	\end{corollary}
	
	The (asymptotic) critical exponent of a colouring $\vv = \barva(\uu, \yy, \yy')$ is bounded from above by the (asymptotic) critical exponent of the Sturmian sequence $\uu$.
	Here we give a lower bound on $E^*(\vv)$.

	\begin{theorem}
		\label{thm:lowerBound}
		Let $\uu$ be a Sturmian sequence, $\yy, \yy'$ two constant gap sequences and $\vv = \barva(\uu, \yy, \yy')$.
		One has
		$$
		\CR(\vv)
		\geq
		\CR^*(\vv)
		\geq
		1 + \frac{1}{{{\rm Per}(\yy)}{{\rm Per}(\yy') }}\,.
		$$
		Moreover, $\CR^*(\vv)$ depends only on ${{\rm Per}(\yy)}$ and ${{\rm Per}(\yy')}$ (not on the structure of $\yy$ and $\yy'$).
	\end{theorem}
	\begin{proof}
		To find a lower bound on the asymptotic critical exponent we use Theorem~\ref{Prop_FormulaForCR}, i.e., we have to consider the values of the form $1 + \frac{|w|}{|v|}$, where $v$ is a shortest return word to a bispecial factor $w \in \LL(\vv)$.
		Obviously, only long bispecial factors $w$ play a role for $\CR^*(\vv)$.
		By Lemma~\ref{lem:special}, every sufficiently long bispecial factor $w$ belongs to $\pi^{-1}(b)$, where $b$ is a bispecial factor in $\LL(\uu)$.
		By Theorem~\ref{prop:ShortestReturn}, a shortest return word $v$ to $w$ has length $|v| = k_0 |r| + \ell_0 |s|$,
		where $\left( \begin{smallmatrix} \ell_0 \\ k_0 \end{smallmatrix} \right) \in \S(b)$.
		In particular, $k_0 + \ell_0 \leq {{\rm Per}(\yy) {\rm Per}(\yy')}$.
		Due to the relation between the Parikh vectors of $b$, $r$ and $s$ given in Proposition~\ref{ParikhRSB1}, we know that $|b| = |r| + |s| - 2$.
		Therefore
		$$
		\frac{|w|}{|v|} =
		\frac{|b|}{k_0 |r| + \ell_0 |s|} =
		\frac{|r |+ |s| }{k_0 |r| + \ell_0 |s|} - \frac{2}{|v|} \geq
		\frac{1}{k_0 + \ell_0} - \frac{2}{|v|} \geq
		\frac{1}{{{\rm Per}(\yy) {\rm Per}(\yy')}} - \frac{2}{|v|}\,.
		$$
		By Remark~\ref{longBS}, the values $k_0$ and $\ell_0$ depend only on $\rm{Per}(\yy)$ and $\rm{Per}(\yy')$.
		In other words, for long bispecial factors $w$ the ratio $|w|/|v|$ is independent from the structure of $\yy$ and $\yy'$ themselves.
		
		For every $n \in \N$ let $v_n$ denote a shortest return word to the $n^{\text{th}}$ bispecial factor $w_n$ in $\vv$.
		The above inequality says that
		$$
		\frac{|w_n|}{|v_n|}
		\geq
		\frac{1}{{{\rm Per}(\yy) {\rm Per}(\yy')}} - \frac{2}{|v_n|}.
		$$
		Since $\vv$ is uniformly recurrent, Lemma~\ref{Lem_URbispecial_retwords} implies that $\lim\limits_{n\to \infty} |v_n| = \infty$.
		Hence
		$$
		\limsup\limits_{n \to \infty} \frac{|w_n|}{|v_n|}
		\geq
		\frac{1}{{{\rm Per}(\yy) {\rm Per}(\yy')}}
		$$
		and Theorem~\ref{Prop_FormulaForCR} gives the required inequality.
	\end{proof}
	
	Let us note that the obtained lower bound is not optimal, at least over small alphabets (see Table~\ref{NaradTable}).
	
	\begin{remark}
		Corollary~\ref{finiteAsympt} implies that the asymptotic critical exponent of a balanced sequence $\vv = \barva(\uu, \yy, \yy')$ is finite if and only if the slope $\theta = [0, a_1, a_2, \ldots]$ of the Sturmian sequence $\uu$ has coefficients $a_n$ bounded by a constant, say $K$.
		Let us outline a method for computing $ \CR^*(\vv)$ in this case.
		
		As the set $\S(b)$ described in Definition~\ref{def:setS} is a subset of the finite set $\S_3$ and the values $m$ are bounded by $m < a_{N+1} \leq K$, the set of all pairs $(N,m)$ corresponding to the projection $b=\pi(w)$ of a sufficiently long bispecial factor $w \in \LL(\vv)$ can be split into a finite number of subsets such that the parameter $m$ and $\S(b)$ are the same in the whole subset.
		For each infinite subset, say $C$, we then compute
		$$
		E^*(C) := \limsup\limits_{N \to \infty, \; (N,m) \in C} I(N,m)
		$$
		and choose the maximal value among $E^*(C)$.
		Due to Formula~\eqref{modified} and monotony of $\frac{1 + m + x}{k + \ell m + \ell x}$ (increasing if $k > \ell$, decreasing if $k < \ell$ and constant if $k = \ell$), we just need to determine $\limsup\limits_{N \to \infty}\tfrac{Q_{N-1}}{Q_N}$ or $\liminf\limits_{N \to \infty}\tfrac{Q_{N-1}}{Q_N}$ in each subset.
		In the next section we will see that if the continued fraction expansion of $\theta$ is eventually periodic, then the partition  of pairs $(N,m)$ into described subsets can be done and the relevant limits can be computed explicitly.
	\end{remark}

	\section{Computation of the asymptotic critical exponent} \label{sec:ComputationUltCriticalExponent}
	From now on we consider a standard Sturmian sequence $\uu$ with slope $\theta$ having eventually periodic continued fraction expansion.
	The goal of this section is to compute the asymptotic critical exponent of a sequence $\vv$ obtained by colouring of $\uu$.
	By Theorem~\ref{Prop_FormulaForCR}, to determine $E^*(\vv)$ we only need to consider sufficiently long bispecial factors $w$.
	
	For this purpose, we write the continued fraction expansion of $\theta$ as
	\begin{equation}
		\label{fixedTheta}
		\theta  = [0, a_1, a_2, \ldots, a_h, \overline{z_0, z_1, \ldots, z_{M-1}}]\,,
	\end{equation}
	where the preperiod $h$ is chosen so that each bispecial factor $b$ associated with $(N,m)$, $N\geq h$, satisfies $|b|_{\tt a}>\beta(\yy)$ and $|b|_{\tt b}>\beta(\yy')$. 
	Let us stress that in the sequel, $M$ will always denote a period of the continued fraction. 
	
	We then decompose the set $\W$ of all nonempty bispecial factors of $\vv=\barva(\uu, \yy, \yy')$ into two subsets:
	\medskip
	
	$\W^{\text{long}} ~:=
	\left\{ w \in \W: \pi(w) \text{ bispecial in } \uu \ \text{ assigned to } (N,m) \text{ with }\ N \geq h \right\}\,;$
	
	$\W^{\text{short}} := \W \setminus \W^{\text{long}}$.
	
	\begin{remark}
		\label{rem:preperiod}
		Let us explain how to find a minimal preperiod length $h$ of the continued fraction of $\theta$ from Equation~\eqref{fixedTheta}.
		By Proposition~\ref{ParikhRSB1}, one has to find the smallest $h$ satisfying $p_h + p_{h-1} - 1 > \beta(\yy)$ and $q_h + q_{h-1} - 1 > \beta(\yy')$ and at the same time $h$ has to be longer than or equal to the shortest possible preperiod of the continued fraction of $\theta$.
	\end{remark}
	
	\begin{example}
		\label{Ex:valueh}
		Let us consider the sequence $\vv = \barva(\uu, \yy, \yy')$ from Example~\ref{ex:Lubka_S(b^3)}, i.e., we have
		$\theta = [0, 3, 2, \overline{3, 1}]$,
		$\yy = ({\tt 01})^{\omega}$ and
		$\yy' = ({\tt 234235})^{\omega}$,
		$$
		\uu = {\tt bbb a bbb a bbbb a bbb a bbbb a bbb a bbbb a bbb a bbb a bbbb a bbb} \cdots,
		$$
		$$
		\vv = {\tt 2 \textcolor{green}{3} 4 \textcolor{red}{0} 235 1 2342  \textcolor{red}{0} \textcolor{green}{3} 52 1 3423 0 523 1 4 \textcolor{blue}{235 0 234} 1 \textcolor{blue}{235 0 234} 2 1 352} \cdots.
		$$
		Clearly, $\beta(\yy) = 0$, the only bispecial factor in $\yy$ being $\varepsilon$, and $\beta(\yy') = 2$, the longest bispecial factor in $\yy'$ being ${\tt 23}$.
		By Remark~\ref{rem:preperiod}, we have $h \geq 2$.
		In fact $h = 2$, since  Table~\ref{tab:Lubka_parameters} gives  $p_2 + p_{1} - 1 = 2 > 0 = \beta(\yy)$ and $q_2 + q_{1} - 1 = 9 > 2 = \beta(\yy')$.
		
		Consider the following bispecial factors in $\vv$:
		${\tt \textcolor{red}{0}}$,
		${\tt \textcolor{green}{3}}$,
		${\tt \textcolor{blue}{2350234}}$
		and
		${\tt 23402351234}$.
		
		\begin{itemize}
			\item ${\tt \textcolor{red}{0}} \in \W^{\text{short}}$, because $0$ is bispecial in $\vv$, but $\pi({\tt 0}) = {\tt a}$ is not bispecial in $\uu$;
			
			\item ${\textcolor{green}{\tt 3} } \in \W^{\text{short}}$, because $\pi({\tt 3}) = {\tt b}$ is the first bispecial factor in $\uu$, hence assigned to $(N,m) = (0,1)$, i.e., $N < h = 2$;
			
			\item ${\textcolor{blue}{\tt 235 0 234}} \in \W^{\text{short}}$: even though $\pi({\tt 235 0 234}) = {\tt b^3ab^3}$ satisfies $|{\tt b^3ab^3}|_{\tt a}=1>0=\beta(\yy)$ and $|{\tt b^3ab^3}|_{\tt b}=6>2=\beta(\yy')$, the factor ${\tt b^3ab^3}$ is the fourth bispecial factor in $\uu$, hence assigned to $(N,m) = (1,1)$, i.e., $N < h = 2$;
			
			\item ${\tt 234 0 235 1 234} \in \W^{\text{long}}$, because $\pi({\tt 234 0 235 1 234}) = {\tt b^3ab^3ab^3}$ is the fifth bispecial factor in $\uu$, hence assigned to $(N,m) = (2,0)$, i.e., $N \geq h = 2$.
		\end{itemize}
	\end{example}
	
	To solve the task of this section, namely to compute $E^*(\vv)$, we will apply Corollary~\ref{indexlong} to sufficiently long bispecial factors and manipulate the numbers $I(N,m)$ defined in Equation~\eqref{modified}.
	
	Our approach consists in partitioning the set $\W^{\text{long}}$ into a finite number of subsets such that the set $\S(\pi(w))$ will be the same for all factors $w$ in the same subset.
	The partition will be based on partition of the pairs $(N,m)$ assigned to their projection $\pi(w)$.
	A suitable partition of $\W^{\text{long}}$ (described later in Definition~\ref{def:classes}) uses the following equivalence relation on the first component of the pairs $(N,m)$.
	
	\begin{definition}
		\label{def:equiv}
		Let $N_1, N_2 \in \N$ and $N_1, N_2 \geq h$.
		We say that $N_1$ is equivalent to $N_2$, and write $N_1 \sim N_2$, if the following three conditions are satisfied:
		\begin{enumerate}
			\item $N_1 \equiv N_2 \pmod M$,
			\item $\left( \begin{smallmatrix}{p_{N_1 - 1}} \\ q_{N_1 - 1} \end{smallmatrix} \right) \equiv \left( \begin{smallmatrix}{p_{N_2 - 1}} \\ q_{N_2 - 1} \end{smallmatrix} \right) \pmod {\left( \begin{smallmatrix} \rm{Per}(\yy) \\ \rm{Per}(\yy') \end{smallmatrix} \right)}$,
			\item $\left( \begin{smallmatrix}{p_{N_1}} \\ q_{N_1} \end{smallmatrix} \right) \equiv \left( \begin{smallmatrix}{p_{N_2}} \\ q_{N_2} \end{smallmatrix} \right) \pmod {\left( \begin{smallmatrix} \rm{Per}(\yy) \\ \rm{Per}(\yy') \end{smallmatrix} \right)}$.
		\end{enumerate}
	\end{definition}
	
	The properties of the equivalence $\sim$ are summarised in the following lemma.
	They follow from the definition of convergents to $\theta$ and from the periodicity of the continued fraction expansion of $\theta$.
	
	\begin{lemma}
		\label{propertyEquiv}
		Let $\sim$ be the equivalence on the set $\{ N \in \N: N \geq h\}$ introduced in Definition~\ref{def:equiv} and let $H$ denote the number of equivalence classes.
		\begin{enumerate}
			\item If $N_1 \sim N_2$, then $a_{N_1 + 1} = a_{N_2 + 1}$.
			\item $N_1 \sim N_2$ if and only if $N_1 + 1 \sim N_2 + 1$.
			\item $N_1 \sim N_2$ if and only if $N_2 \equiv N_1 \pmod H$.
			\item $H = \min \left\{ i \in \N, i > 0 : h + i \sim h \right\} \leq M {\rm Per}(\yy)^2 {\rm Per}(\yy')^2$.
			\item $H$ is divisible by $M$.
		\end{enumerate}
	\end{lemma}
	
	\begin{corollary}
		\label{classes}
		Let $b^{(1)}$ and $b^{(2)}$ be bispecial factors of $\uu$ and $(N_1, m_1)$ and $(N_2, m_2)$, with $N_1 \geq h$ and $N_2 \geq h$, be the pairs assigned to $b^{(1)}$ and $b^{(2)}$ respectively.
		
		If $N_1 \sim N_2$ and $m_1 = m_2$, then $\S(b^{(1)}) = \S(b^{(2)})$.
	\end{corollary}
	\begin{proof}
		Let us recall that the set $\S(b)$ defined in   Definition~\ref{def:setS} for a bispecial factor $b$ of a Sturmian sequence equals  $\S_1(b) \cap \S_2(b) \cap \S_3$.
		
		Let $s^{(1)}$ and $r^{(1)}$ denote the return words to $b^{(1)}$ in $\uu$, and $s^{(2)}$ and $r^{(2)}$ the return words to $b^{(2)}$ in $\uu$.
		Let $\theta_1'$ and $\theta_2'$ denote the slopes of the derived sequences $\dd_\uu(b^{(1)})$ and $\dd_\uu(b^{(2)})$ respectively.
		
		By Proposition~\ref{ParikhRSB1} and the definition of $\sim$, we have
		$$
		\Parickh (s^{(1)})
		\equiv
		\Parickh(s^{(2)}) \pmod {\begin{pmatrix} \rm{Per}(\yy) \\ \rm{Per}(\yy')\end{pmatrix}}
		\ \ \text{and}\ \
		\Parickh(r^{(1)})
		\equiv
		\Parickh(r^{(2)}) \pmod {\begin{pmatrix} \rm{Per}(\yy) \\ \rm{Per}(\yy')\end{pmatrix}}.
		$$
		
		As $N_1 \geq h$, we have $|b^{(1)}|_{\tt a}>\beta(y)$ and $|b^{(1)}|_{\tt b}>\beta(\yy')$ and similarly for $b^{(2)}$.
		Thus Remark~\ref{longBS} implies  $\S_2(b^{(1)}) = \S_2(b^{(2)})$.
		
		Since $N_1 \equiv N_2 \pmod M$ and $m_1 = m_2$, Proposition~\ref{ParikhRSB1} says that for some $i \in \{ 0, 1, \ldots, M-1 \}$ we have
		$$
		\theta_1' =
		\theta_2' =
		[0, z_i - m_1, z_{i+1}, \ldots, z_{M-1}, \overline{z_0,z_1, \ldots,z_{M-1}}]\,.
		$$
		Thus $\dd_\uu(b^{(1)}) = \dd_\uu(b^{(2)})$ and $\S_1(b^{(1)}) = \S_1(b^{(2)})$ too.
	\end{proof}
	
	Now we define a partition of the set $\W^{\text{long}}$ of long bispecial factors of a balance sequence.
	
	\begin{definition}
		\label{def:classes}
		Let $\vv = \barva(\uu, \yy, \yy')$.
		Let $\sim$ be the equivalence given in Definition~\ref{def:equiv} and $H$ be the number of its equivalence classes.
		For $0 \le i < H$ and $0 \le m < z_{i \bmod M}$ we define the set $C(i,m)$ as follows:
		a bispecial factor $w \in \W^{\text{long}}$ belongs to $C(i,m)$ if the pair assigned in Proposition~\ref{ParikhRSB1} to the bispecial factor $b = \pi(w)$ of $\uu$ is $(h + i + NH, m)$ for some $N \in \N$.
	\end{definition}
	
	Clearly, the sets $C(i,m)$ form a partition of $\W^{\text{long}}$.
	
	\begin{example}
		\label{ex:Lubka_H}
		Let us consider the sequence $\vv = \barva(\uu, \yy, \yy')$ from Example~\ref{ex:Lubka_S(b^3)}, where ${\rm Per}(\yy) = 2$, ${\rm Per}(\yy') = 6$ and $\theta = [0, 3, 2, \overline{3, 1}]$.
		Therefore $M = 2$ and $h = 2$ as determined in Example~\ref{Ex:valueh}.
		Let us find the number $H$ of equivalence classes from Definition~\ref{def:equiv}.
		Observing Table~\ref{tab:Lubka_parameters}, we have
		$$
		H = \min \left\{ i \in \N, i > 0 : 2 + i \sim 2 \right\} = 6.
		$$
		Indeed, we have
		$\textcolor{green}{2} \equiv \textcolor{green}{8} \pmod 2$,
		$\left( \begin{smallmatrix}{\textcolor{red}{p_{1}}} \\ {\textcolor{red}{q_{1}}}\end{smallmatrix} \right) \equiv \left( \begin{smallmatrix}{\textcolor{red}{p_{7}}} \\ {\textcolor{red}{q_{7}}} \end{smallmatrix} \right) \pmod {\left( \begin{smallmatrix} 2 \\ 6 \end{smallmatrix} \right)}$
		and
		$\left( \begin{smallmatrix}\textcolor{blue}{{p_{2}}} \\ \textcolor{blue}{q_{2}} \end{smallmatrix} \right) \equiv \left( \begin{smallmatrix} \textcolor{blue}{p_{8}} \\ \textcolor{blue}{q_{8}} \end{smallmatrix} \right) \pmod {\left( \begin{smallmatrix} 2 \\ 6\end{smallmatrix} \right)}$.
		According to Definition~\ref{def:classes}, we have 12 subsets $C(i,m)$, where $(i,m)$ belongs to the set
		$$
		\left\{ (0,0), (0,1), (0,2), (1,0), (2,0), (2,1), (2,2), (3,0), (4,0), (4,1), (4,2), (5,0) \right\}.
		$$
		
		\begin{table}[htb!]
			\centering
			\setlength{\tabcolsep}{5pt}
			\renewcommand{\arraystretch}{1.3}
			\resizebox{\textwidth}{!} {
				\begin{tabular}{|l|l|l|l|l|l|l|l|l|l|l|}
					\hline
					$N$ & $0$ & $1$ & $\textcolor{green}{2}$ & $3$ & $4$ & $5$ & $6$ & $7$ & $\textcolor{green}{8}$ & $9$ \\
					\hline
					$a_N$ & $0$ & $3$ & $2$ & $3$ & $1$ & $3$ & $1$ & $3$ & $1$ & $3$ \\
					\hline
					$p_{N}$ & $0$ & $1$ & $2$ & $7$ & $9$ & $34$ & $43$ & $163$ & $206$ & $781$\\
					\hline
					$p_{N} \bmod {\rm Per(\yy)}$ & $0$ & $\textcolor{red}{1}$ & $\textcolor{blue}{0}$ & $1$ & $1$ & $0$ & $1$ & $\textcolor{red}{1}$ & $\textcolor{blue}{0}$ & $1$\\
					\hline
					$q_{N}$ & $1$ & $3$ & $7$ & $24$ & $31$ & $117$ & $148$ & $561$ & $709$ & $2688$ \\
					\hline
					$q_{N} \bmod {\rm Per(\yy')}$ & $1$ & $\textcolor{red}{3}$ & $\textcolor{blue}{1}$ & $0$ & $1$ & $3$ & $4$ & $\textcolor{red}{3}$ & $\textcolor{blue}{1}$ & $0$ \\
					\hline
					$Q_{N}$ & $1$ &  $4$ & $9$ & $31$ & $40$ & $151$ & $191$ & $724$ & $915$ & $3469$ \\
					\hline
				\end{tabular}
			}
			\vspace{0.3cm}
			\caption{The first values of $a_N, p_N, q_N, Q_N$ for $\uu$ with $\theta = [0, 3, 2, \overline{3, 1}]$.}
			\label{tab:Lubka_parameters}
		\end{table}
	\end{example}
	
	To show the advantages of the chosen partition into the sets $C(i,m)$ we use a property of primitive matrices.
	Let us recall that a matrix $A$ with non-negative entries is said to be primitive if there exists an exponent $k \in \N$ such that all entries of $A^k$ are positive.
	
	\begin{lemma}
		\label{lem:append1}
		Let $A \in \N^{2\times 2}$ be a primitive matrix with $\det A = \pm 1$, and $\left( S_N \right)$, $\left( T_N \right)$ be two sequences of integers given by the recurrence relation
		$(S_{N+1}, T_{N+1}) = (S_{N}, T_{N}) A$
		for each $N \in \N$, with
		$S_0, T_0 \in \N$ such that
		$S_0 + T_0 > 0$.
		Let
		$\left( \begin{smallmatrix} x \\ y \end{smallmatrix} \right)$
		be an eigenvector of $A$ associated with the non-dominant eigenvalue $\lambda$.
		Then
		\begin{enumerate}
			\item $\lim\limits_{N \to \infty} \tfrac{S_N}{T_N} = -\tfrac{y}{x}$, and
			\item $S_N + \tfrac{y}{x} T_N = \lambda^N (S_0 + \tfrac{y}{x} T_0)$ for each $N \in \N$.
		\end{enumerate}
	\end{lemma}
	\begin{proof}
		Since $A$ is a primitive matrix with non-negative entries, the components $x$ and $y$ of an eigenvector corresponding to the non-dominant eigenvalue have opposite signs.
		In particular $x, y \neq 0$.
		Obviously,
		$$
		(S_{N}, T_{N}) = (S_{0}, T_{0}) A^N \text{ for each $N \in \N$.}
		$$
		Multiplying both sides of the previous equation by the eigenvector $\left( \begin{smallmatrix} x \\ y \end{smallmatrix} \right)$, we obtain
		$$
		x S_N + y T_N = \lambda^N (xS_0 + yT_0),
		$$
		i.e., Item 2 is proven.
		
		As $|\lambda| < 1$, Item 2 implies that
		$$
		\lim\limits_{N \to \infty} x T_N \left( \tfrac{S_N}{T_N} + \tfrac{y}{x} \right) = \lim\limits_{N \to \infty} (x S_N + y T_N) = 0.
		$$
		Since $\lim\limits_{N \to \infty } T_N =  +\infty$, necessarily $\lim\limits_{N \to \infty} \left( \tfrac{S_N}{T_N} + \tfrac{y}{x} \right) = 0$.
		This proves Item 1.
	\end{proof}
	
	\begin{corollary}
		\label{cor:limits}
		Let $\frac{p_N}{q_N}$ denote the $N^{\text{th}}$ convergent to $\theta$ defined by Equation~\eqref{fixedTheta} and $Q_N = p_N + q_N$.
		Fix $i \in \{0,1, \ldots, M-1\}$.
		Then
		\begin{equation}
			\label{Limits}
			L_i:=\lim_{N \to \infty} \frac{Q_{{MN} + h + i - 1}}{Q_{MN + h + i}} =  -\frac{y_i}{x_i},
		\end{equation}
		where $\left( \begin{smallmatrix} x_i \\ y_i \end{smallmatrix} \right)$ is an eigenvector of the matrix
		\begin{equation}
			\label{A}
			A^{(i)} =
			\begin{pmatrix} 0 & 1 \\ 1 & z_i \end{pmatrix}
			\begin{pmatrix} 0 & 1 \\ 1 & z_{i+1} \end{pmatrix}
			\cdots
			\begin{pmatrix} 0 & 1 \\ 1 & z_{M-1} \end{pmatrix}
			\begin{pmatrix} 0 & 1 \\ 1 & z_0 \end{pmatrix}
			\cdots
			\begin{pmatrix} 0 & 1 \\ 1 & z_{i-1} \end{pmatrix}
		\end{equation}
		corresponding to the non-dominant eigenvalue $\lambda$. Moreover,
		\begin{equation}
			\label{discrepancy}
			{Q_{{MN} + h + i - 1}} - L_i{Q_{MN + h + i}} = \lambda^N \bigl( {Q_{ h + i - 1}} - L_i{Q_{ h + i}}\bigr)\,.
		\end{equation}
	\end{corollary}
	\begin{proof}
		We apply the previous lemma to the sequences $S_N := Q_{{MN} + h + i - 1}$ and $T_N := Q_{MN + h + i}$.
		Periodicity of the continued fraction expansion of $\theta$ ensures that the sequences $(S_N)$ and $(T_N)$ satisfy the recurrence relation $(S_{N+1}, T_{N+1}) = (S_N, T_N)A^{(i)}$.
	\end{proof}
	
	\begin{remark}
		The matrices $A^{(i)}$ defined in Equation~\eqref{A} are mutually similar.
		In particular, they share the same spectrum.
	\end{remark}
	
	Let us point out two important properties of the partition into subsets $C(i,m)$:
	
	\begin{enumerate}
		\item
		By Corollary~\ref{classes}, the sets $\S(\pi(w))$ are the same for all $w \in C(i,m)$.
		Therefore, we put for $0 \le i < H$ and $0 \le m < z_{i \bmod M}$
		\begin{equation}
			\label{SetS(im)}
			\S(i,m) := \S(\pi(w))\,,
			\text{ \ where }
			w \in C(i,m)\,.
		\end{equation}
		
		\item As $M$ divides $H$, Corollary~\ref{cor:limits} ensures existence of the limit
		\begin{equation}
			\label{valueL}
			L_i :=
			\lim_{N \to \infty} \frac{Q_{{HN} + h + i - 1}}{Q_{HN + h + i}}
		\end{equation}
		for each $i =0,1, \ldots, H-1$.
		Moreover, if $i =0, 1, \ldots, M-1$, the limit $L_i$ coincides with the one given in~\eqref{Limits}.
		For $j \geq M$ we have $L_j = L_{j \bmod M}$.
	\end{enumerate}
	
	\begin{example}
		\label{ex:Lubka_L_i}
		Consider the sequence $\vv = \barva(\uu, \yy, \yy')$ from Example~\ref{ex:Lubka_S(b^3)}.
		Since $\theta = [0, 3, 2, \overline{3, 1}]$, we have $$
		A^{(0)} =
		\begin{pmatrix} 0 & 1 \\ 1 & 3 \end{pmatrix}
		\begin{pmatrix} 0 & 1 \\ 1 & 1 \end{pmatrix} =
		\begin{pmatrix} 1& 1 \\ 3 & 4 \end{pmatrix}.
		$$
		The non-dominant eigenvalue of $A^{(0)}$ is $\lambda = \frac{5 - \sqrt{21}}{2}$.
		An eigenvector of $A^{(0)}$ corresponding to $\lambda$ is, for instance, $\left(\begin{smallmatrix} x_0 \\ y_0 \end{smallmatrix}\right) = \left(\begin{smallmatrix} 2 \\ 3-\sqrt{21} \end{smallmatrix}\right)$.
		Therefore
		$$
		L_0 =
		\lim_{N \to \infty} \frac{Q_{2N+1}}{Q_{2N+2}} =
		-\frac{y_0}{x_0} =
		\frac{\sqrt{21} - 3}{2}.
		$$
		Similarly,
		$$
		A^{(1)} =
		\begin{pmatrix} 0 & 1 \\ 1 & 1 \end{pmatrix}
		\begin{pmatrix} 0 & 1 \\ 1 & 3 \end{pmatrix}=(A^{(0)})^T.
		$$
		The matrices $A^{(1)}$ and $A^{(0)}$ are similar; thus they have the same eigenvalues.
		An eigenvector of $A^{(1)}$ corresponding to $\lambda$ is, for instance, $\left(\begin{smallmatrix} x_1 \\ y_1 \end{smallmatrix}\right) = \left(\begin{smallmatrix} 6 \\ 3-\sqrt{21} \end{smallmatrix}\right)$.
		Therefore
		$$
		L_1 =
		\lim_{N \to \infty} \frac{Q_{2N+2}}{Q_{2N+3}} =
		-\frac{y_1}{x_1} =
		\frac{\sqrt{21}-3}{6}.
		$$
		In Example~\ref{ex:Lubka_H} we have seen that $H=6$, thus we need to know $L_i$ for $i \in \{0, \ldots, 5\}$.
		Since $M=2$, we get $L_i = L_0$ for $i$ even and $L_i = L_1$ for $i$ odd.
	\end{example}
	
	Combining Theorem~\ref{Prop_FormulaForCR}, Corollary~\ref{indexlong} and notation of Formula~\eqref{modified}, we can transform our task to determine $E^*(\vv)$ into looking for
	$$
	E^*(i,m) :=
	\limsup\limits_{N \to \infty} I(h + i + NH, \ m)
	$$
	as
	\begin{equation}
		\label{asymptE}
		E^*(\vv) =
		\max \left\{ E^*(i,m) \ : \ 0\leq i < H,\ \ 0 \leq m < z_{i \bmod M} \right\}\,.
	\end{equation}
	
	Formula~\eqref{modified} immediately gives
	\begin{equation}
		\label{asymptE(i,m)}
		E^*(i,m) =
		1 + \max \left\{ \frac{1 + m + L_i}{k+ \ell m + \ell L_i }
		\ : \
		\left( \begin{smallmatrix} \ell \\ k \end{smallmatrix} \right) \in \S(i,m) \right\} \,.
	\end{equation}
	
	Equations~\eqref{asymptE} and~\eqref{asymptE(i,m)} provide an algorithm for computing $E^*(\vv)$.
	The only detail we need to recall is how to find the set $\S(i,m)$ defined in Formula~\eqref{SetS(im)}.
	We find it as $\S(\pi(w))$ for the shortest bispecial factor $w$ in the set $C(i,m)$.
	Its projection $b = \pi(w)$ corresponds to the pair $(h+i, m)$.
	
	\medskip
	
	By Proposition~\ref{ParikhRSB1} and the continued fraction expansion in Formula~\eqref{fixedTheta} of $\theta$, we have the following facts:
	
	\begin{itemize}
		\item the derived sequence $\dd_{\uu}(b)$ to $b$ in $\uu$ is a Sturmian sequence with the slope
		\begin{equation}
			\label{thetaIM}
			\theta_{i,m} : =
			[0, z_i-m, z_{i+1}, \ldots, z_{M-1}, \overline{z_0, z_1 \ldots, z_{M-1}}]\,;
		\end{equation}
		
		\item the Parikh vectors of the prefix return word $r$ and of the non-prefix return word $s$ to $b$ in $\uu$ are
		$$
		\Parickh(r) =
		\begin{pmatrix} p_{h+i} \\ q_{h+i} \end{pmatrix}
		\ \ \text{and} \ \
		\Parickh(s) =
		\begin{pmatrix} m \, p_{h+i} + p_{h+i-1} \\ m \, q_{h+i} + q_{h+i-1} \end{pmatrix}.
		$$
	\end{itemize}
	
	Now we have all ingredients needed in Definition~\ref{def:setS} for describing $\S(i,m) = \S(b)$.
	Lemma~\ref{lem_kl} helps to decide which vectors occur as the Parikh vectors of factors of the derived sequence $\dd_{\uu}(b)$.
	
	\begin{corollary}
		\label{coro:S(i,m)}
		Let $\theta_{i,m}$ be as in Formula~\eqref{thetaIM}.
		Then $\left(\begin{smallmatrix} \ell \\ k \end{smallmatrix}\right) \in \S(i,m)$ if and only if
		\begin{enumerate}
			\item
			$\Bigl( \begin{smallmatrix}p_{h+i-1} & p_{h+i} \\ q_{h+i - 1} & q_{h+i} \end{smallmatrix} \Bigr)
			\Bigl( \begin{smallmatrix} 1 & 0 \\ m & 1 \end{smallmatrix} \Bigr)
			\Bigl( \begin{smallmatrix} \ell \\ k  \end{smallmatrix} \Bigr)
			\equiv
			\Bigl( \begin{smallmatrix} 0 \\0 \end{smallmatrix} \Bigr)
			\pmod
			{\left( \begin{smallmatrix} {\rm Per}(\yy) \\ {\rm Per}(\yy')\end{smallmatrix} \right)}$;
			
			\medskip
			
			\item
			$(k-1)\theta_{i,m}-1 \leq \ell \leq (k+1)\theta_{i,m} + 1$;
			
			\medskip
			
			\item
			$1 \leq \ell + k \leq {{{\rm Per}(\yy) {\rm Per}(\yy')}}$.
		\end{enumerate}
	\end{corollary}
	
	\begin{example}
		\label{ex:Lubka_AsymptoticCR}
		Let us consider the sequence $\vv = \barva(\uu, \yy, \yy')$ from Example~\ref{ex:Lubka_S(b^3)}.
		Let us determine $E^*(\vv)$.
		See also Examples~\ref{ex:Lubka_H} and~\ref{ex:Lubka_L_i} for important ingredients.
		By Remark~\ref{rem:tildeS}, it is sufficient to describe $\hat \S(i,m)$ for $(i,m)$ in the set
		$$
		\left\{ (0,0), (0,1), (0,2), (1,0), (2,0), (2,1), (2,2), (3,0), (4,0), (4,1), (4,2), (5,0) \right\}.
		$$
		Then $E^*(i,m) =
		1 + \max \left\{ \cfrac{1 + m + L_i}{k+ \ell m + \ell L_i }
		\ : \
		\left( \begin{smallmatrix} \ell \\ k \end{smallmatrix} \right) \in \hat\S(i,m) \right\} \,.$
		
		The reader is invited to verify the following calculations. 
		
		\begin{itemize}
			\item $\hat\S(0,0)$:
			We have $\theta_{0,0} = [0, \overline{3,1}] = \frac{\sqrt{21}-3}{6}$.
			By Corollary~\ref{coro:S(i,m)} and Remark~\ref{rem:tildeS}, if $\left(\begin{smallmatrix} \ell \\ k \end{smallmatrix}\right) \in \hat\S(0,0)$, then
			$$
			k \begin{pmatrix} 0 \\ 1 \end{pmatrix} + \ell \begin{pmatrix} 1 \\ 3 \end{pmatrix}
			\equiv
			\begin{pmatrix} 0 \\ 0 \end{pmatrix}
			\pmod
			{\begin{pmatrix} 2 \\ 6 \end{pmatrix}}
			\quad \text{and} \quad
			(k-1) \theta_{0,0} - 1 \leq \ell \leq (k+1) \theta_{0,0} + 1.
			$$
			We get $\hat\S(0,0) = \left\{ \left( \begin{smallmatrix} 2 \\ 6 \end{smallmatrix} \right) \right\}$ and $E^*(0,0) = 1 + \frac{1+L_0}{6+2L_0} = 1 + \frac{\sqrt{21}-1}{6+2\sqrt{21}} \doteq 1.236$.
			
			\item $\hat\S(0,1)$:
			We have $\theta_{0,1} = [0, 2, \overline{1,3}] = \frac{\sqrt{21}-3}{9-\sqrt{21}}$.
			If $\begin{pmatrix} \ell \\ k \end{pmatrix} \in \hat\S(0,1)$, then
			$$
			k \begin{pmatrix} 0 \\ 1 \end{pmatrix} + \ell \begin{pmatrix} 1 \\ 4 \end{pmatrix}
			\equiv
			\begin{pmatrix} 0 \\ 0 \end{pmatrix}
			\pmod
			{\begin{pmatrix} 2 \\ 6 \end{pmatrix}}
			\quad \text{and} \quad
			(k-1) \theta_{0,1}-1 \leq \ell \leq (k+1) \theta_{0,1} + 1.
			$$
			We get $\hat\S(0,1) = \left\{ \left( \begin{smallmatrix} 2 \\ 4 \end{smallmatrix} \right) \right\}$ and $E^*(0,1) = 1 + \frac{2+L_0}{6+2L_0} = 1 + \frac{\sqrt{21}+1}{6+2\sqrt{21}} \doteq 1.368$.
			
			\item $\hat\S(0,2)$:
			We have $\theta_{0,2} = [0, 1, \overline{1,3}] = \frac{\sqrt{21}-3}{12-2\sqrt{21}}$.
			If $\left( \begin{smallmatrix} \ell \\ k \end{smallmatrix} \right) \in \hat\S(0,2)$, then
			$$
			k \begin{pmatrix} 0 \\ 1 \end{pmatrix} + \ell \begin{pmatrix} 1 \\ 5 \end{pmatrix}
			\equiv
			\begin{pmatrix} 0 \\0 \end{pmatrix}
			\pmod
			{\begin{pmatrix} 2 \\ 6 \end{pmatrix}}
			\quad \text{and} \quad
			(k-1) \theta_{0,2} - 1 \leq \ell \leq (k+1) \theta_{0,2} + 1.
			$$
			We get $\hat\S(0,2) = \left\{ \left( \begin{smallmatrix} 2 \\ 2 \end{smallmatrix} \right) \right\}$ and $E^*(0,2) = 1 + \frac{3+L_0}{6+2L_0} = 1.5$.
			
			\item $\hat\S(1,0)$:
			We have $\theta_{1,0} = [0, \overline{1,3}] = \frac{\sqrt{21}-3}{2}$.
			If $\left( \begin{smallmatrix} \ell \\ k \end{smallmatrix} \right) \in \hat\S(1,0)$, then
			$$
			k \begin{pmatrix} 1 \\ 0 \end{pmatrix} + \ell \begin{pmatrix} 0 \\ 1 \end{pmatrix}
			\equiv
			\begin{pmatrix} 0 \\ 0 \end{pmatrix}
			\pmod
			{\begin{pmatrix} 2 \\ 6 \end{pmatrix}}
			\quad \text{and} \quad
			(k-1) \theta_{1,0} - 1 \leq \ell \leq (k+1) \theta_{1,0} + 1.
			$$
			We get $\hat\S(1,0) = \left\{ \left( \begin{smallmatrix} 0 \\ 2 \end{smallmatrix} \right) \right\}$ and $E^*(1,0) = 1 + \frac{1+L_1}{2} = 1 + \frac{\sqrt{21}+3}{12} \doteq 1.63$.
		\end{itemize}
		
		Note that the values of $E^*(i,m)$ periodically repeat, i.e., $E^*(i,m) = E^*(0,m)$ for $i$ even and $E^*(i,0) = E^*(1,0)$ for $i$ odd.
		We conclude that $E^*(\vv) = 1 + \frac{\sqrt{21}+3}{12} \doteq 1.63$.
	\end{example}
	
	As we have already mentioned, the asymptotic critical exponent depends only on the length of the periods of $\yy$ and $\yy'$ and it does not depend on their structure.
	On the other hand, the asymptotic critical exponent depends on the matrix $\Bigl( \begin{smallmatrix} p_{h-1} & p_{h} \\ q_{h - 1} & q_{h} \end{smallmatrix} \Bigr)$, i.e., on the preperiod of the continued fraction of $\theta$, in contrast to the asymptotic critical exponent of the associated Sturmian sequence, see the beginning of Section~\ref{sec:bounds}. 
	
	\begin{example}
		Let us consider the sequence $\vv = \barva(\uu, \yy, \yy')$ associated with $\theta = [0, \overline{2}]$, and such that ${\rm Per}(\yy) = 1$ and ${\rm Per}(\yy') = 2$.
		One can check that
		$E^*(\vv) = 3 + \sqrt{2} \doteq 4.41$.
		For $\vv' = \barva(\uu', \yy, \yy')$ associated with $\theta' = [0, 1, \overline{2}]$, one has $E^*(\vv') = 2 + \frac{\sqrt{2}}{2} \doteq 2.7$.
	\end{example}

	\section{Computation of the critical exponent}
	
	We again consider a colouring of a Sturmian sequence $\uu$ having slope $\theta$ with eventually periodic continued fraction expansion fixed in Formula~\eqref{fixedTheta}.
	In order to evaluate the critical exponent of $\vv = \barva(\uu, \yy, \yy')$, we use Theorem~\ref{Prop_FormulaForCR}.
	We have to determine
	$$
	E(\vv) = 1 + \sup \left\{ \tfrac{|w|}{|v|} : w  \in \W \text{ and } v \in  \R_{\vv}(w) \right\},
	$$
	where $\W$ denotes the set of bispecial factors in $\vv$.
	This set was written in the previous section in the form $\W = \W^{\text{short}} \cup \W^{\text{long}}$.
	Moreover the set $\W^{\text{long}}$ was partitioned into the subsets $C(i,m)$.
	Thanks to Remark~\ref{rem:supset}, in the formula for $E(\vv)$ one can replace the set $\W$ by any its superset.
	As we have no tool for finding elements of $\W^{\text{short}}$ we will use Lemma~\ref{lem:special} and consider instead its superset
	$$
	\LL^{\text{short}} := \W^{\text{short}} \cup \{w \in \LL(\vv) : |\pi(w)|_{\tt a} \leq \bs(\yy) \text{ or }|\pi(w)|_{\tt b} \leq \bs(\yy')\}.
	$$
	Let us define the numbers
	\begin{itemize}
		\item
		$E^{\text{short}}(\vv) := 1 + \max \left\{ {|w|}/{|v|}: w \in \LL^{\text{short}} \text{ and } v \in \R_{\vv}(w) \right\}$;
		
		\item
		$E(i,m) := 1 + \sup \left\{ {|w|}/{|v|} : w \in C(i,m) \text{ and } v \in  \R_{\vv}(w) \right\}$, where $i, m \in \N$, $0 \leq i < H$ and $0 \leq m < z_{i \bmod M}$.
	\end{itemize}
	
	Obviously, $E(\vv)$ is the maximum value from the finite list formed by  $E^{\text{short}}(\vv)$ and   $E(i,m)$, with $0 \leq i < H$ and $0\leq m < z_{i \bmod M}$.
	Let us comment on the individual steps of the computation.
	
	\begin{description}
		\item[$E^{\text{short}}(\vv) $\ :]
		Let us point out that if two factors $w^{(1)}, w^{(2)} \in \LL^{\text{short}}$ have the same length and $w^{(1)}$ has a shorter return word than the length of each return word to $w^{(2)}$, then only the factor $w^{(1)}$ may influence $E^{\text{short}}(\vv)$.
		To compute the relevant ratio $1 + \frac{|w|}{|v|}$ we use Theorem~\ref{prop:ShortestReturn} if the projection $\pi(w)$ of $w \in \LL^{\text{short}}$ is not a bispecial factor in $\uu$.  Otherwise, Corollary~\ref{indexlong} facilitates our computation.
		
		\item[$E(i,m) $\ :]
		Unlike the previous case, now the set of bispecial factors we have to take into consideration is infinite.
		Proposition~\ref{pro:onlySmall} we will present below shows that the knowledge of the value $E^*(i,m)$ reduces our task to examination of only a finite number of bispecial factors.
		By Corollary~\ref{indexlong}, we have
		$$
		E(i,m)
		=\sup\{I(h + i + NH,m): N\in \N \} \geq \limsup\limits_{N \to \infty} I(h + i + NH, m) = E^*(i,m).
		$$
		
		In fact, $I(h + i + NH, m)$ may exceed $E^*(i,m)$ only for a finite number of indices $N \in \N$.
	\end{description}
	
	\begin{proposition}
		\label{pro:onlySmall}
		Let $i \in \{0, 1, \ldots, H-1\}$.
		Let $L_i$ be the limit given in Equation~\eqref{valueL} and  $\lambda$ the non-dominant eigenvalue of the matrix ${A^{(0)}}$ from Corollary~\ref{cor:limits}.
		Assume that $w \in C(i,m)$ and let $(h+i+NH, m)$ be the pair assigned to $\pi(w)$.
		If for $N_0 \in \N$
		$$
		{|\lambda|}^{N_0H/M} \left| Q_{h+i -1} -L_i Q_{h+i} \right| \ \leq \ 2 L_i, \ \text{then} \ I(h+i+NH, m) \ \leq \ E^*(i,m)
		$$ for all $N \geq N_0$.
	\end{proposition}
	\begin{proof}
		Having in mind that $H$ is divisible by $M$, Equation~\eqref{discrepancy} gives for each $N \in \N$
		$$
		Q_{h+i+NH -1} - L_i Q_{h+i+NH}
		=
		{\lambda}^{NH/M} \left( Q_{h+i -1} - L_i Q_{h+i} \right)\,.
		$$
		Since $|\lambda|<1$, the sequence $\left| Q_{h+i+NH -1} - L_i Q_{h+i+NH} \right|$ is decreasing in $N$.
		Hence, it is enough to show the implication:
		$$
		I(h+i+NH, m) > E^*(i,m)
		\quad \Longrightarrow \quad
		\left| Q_{h+i+NH -1} - L_i Q_{h+i+NH} \right| > 2L_i\,.
		$$
		For this sake, we abbreviate the notation by putting
		$$
		S = Q_{h+i+NH -1},
		\quad
		T = Q_{h+i+NH}
		\quad \text{and} \quad
		L = L_i.
		$$
		Recall that $0 < L <1$ and $\S(\pi(w)) = \S(i,m)$ for every $w \in C(i,m)$.
		Choose $\left( \begin{smallmatrix} \ell \\ k \end{smallmatrix} \right) \in \S(i,m)$ such that
		$$
		I(h+i+NH, m) =
		1 + \frac{(1 + m) T + S-2}{(k + \ell m) T + \ell S}
		>
		E^*(i,m)
		\geq
		1 + \frac{1 + m + L}{k + \ell m + \ell L}\,.
		$$
		Thus we have $(k-\ell)(S-LT) > 2(k + \ell m + \ell L) \geq 2 L |k-\ell|$.
		Hence $|S-LT| > 2L$.
	\end{proof}
	
	\begin{example}
		Let us consider the sequence $\vv = \barva(\uu, \yy, \yy')$ from Example~\ref{ex:Lubka_S(b^3)}, i.e., $\theta = [0, 3, 2, \overline{3, 1}]$, $\yy = ({\tt 01})^{\omega}$, $\yy' = ({\tt 234235})^{\omega}$,
		$$
		\uu = {\tt bbb a bbb a bbbb a bbb a bbbb a bbb a bbbb a bbb a bbb a bbbb a bbb  }\cdots,
		$$
		$$
		\vv = {\tt 2340 235 1 2342  0  352 1 3423 0 523 1 4235 0 234 1235 0 2342 1 352}\cdots.
		$$
		In this example, we will determine $\CR(\vv)$.
		
		\vspace{0.3cm}
		
		\noindent First, let us inspect $E^{\text{short}}(\vv)$.
		We know that $\beta(\yy) = 0$ and $\beta(\yy') = 2$.
		Thus, we have the following set of short factors in $\vv$:
		$$
		\LL^{\text{short}} = \left \{w \in \LL(\vv) : \pi(w) \in \{\tt b, b^2, b^3, b^3ab^3, b^4, a, ab, ba, ab^2, bab, b^2a \} \right\}\,.
		$$
		The first four projections ${\tt b, b^2, b^3, b^3ab^3}$ are bispecial, with $(N,m)$ satisfying $N < h = 2$.
		The other projections are not bispecial in $\uu$.
		By Theorem~\ref{prop:ShortestReturn}, we have to examine $\S(\pi(w))$ for all $w \in \LL^{\text{short}}$.
		In fact, we examine $\hat\S(\pi(w))$ instead (see Remark~\ref{rem:tildeS}).
		The sets $\S({\tt b^3})$ and $\S({\tt ab})$ were already described in Example~\ref{ex:Lubka_S(b^3)}; we have
		$
		\hat\S({\tt b^3}) =
		\hat\S({\tt ab}) =
		\left\{
		\left( \begin{smallmatrix} 0 \\ 2 \end{smallmatrix} \right) \right \}.
		$
		Moreover, the shortest return word to $w \in \LL(\vv)$ with $\pi(w) = {\tt b^3}$ is of length $8$ (see Example~\ref{ex:Lubka_shortest_retword}).
		The shortest return word to $w \in \LL(\vv)$ with $\pi(w) = {\tt ab}$ satisfies, by Theorem~\ref{prop:ShortestReturn},
		$$
		|v| =
		\min \left\{ 4k + 5\ell : \begin{pmatrix} \ell \\ k \end{pmatrix} \in \hat\S({\tt ab}) \right\}
		= 8.
		$$
		We will inspect in a similar manner the lengths of the shortest return words to all remaining factors $w$ in $\LL^{\text{short}}$.
		Let us distinguish the following cases according to the projection $\pi(w)$.
		
		\begin{itemize}
			\item ${\tt b}$:
			Checking the prefix of $\uu$, we can see that $r = {\tt b}$ and $s = {\tt ba}$.
			Thus, the corresponding derived sequence starts as follows: $\dd_\uu({\tt b}) = {\tt rrsrrsrrrs} \cdots$, where we remind that ${\tt r}$ and ${\tt s}$ are the letters corresponding to the return words $r$ and $s$ respectively.
			Moreover, $\gap{\yy}{|{\tt b}|_{\tt a}} = \gap{\yy}{0} = \{ 1 \}$ and $\gap{\yy'}{|{\tt b}|_{\tt b}} = \gap{\yy'}{1} = \{ 3,6 \}$.
			If $\left( \begin{smallmatrix} \ell \\ k \end{smallmatrix} \right) \in \hat\S(\tt b)$, then
			$$
			\begin{pmatrix} \ell \\ k \end{pmatrix} \ \text{is a Parikh vector of a factor in} \ \dd_\uu({\tt b});
			$$
			$$
			k \Parickh(r) + \ell \Parickh(s) =
			k \begin{pmatrix} 0 \\ 1 \end{pmatrix} + \ell \begin{pmatrix} 1 \\ 1 \end{pmatrix}
			\equiv
			\begin{pmatrix} 0 \\ 0 \end{pmatrix}
			\pmod
			{\begin{pmatrix} 1 \\ 3 \ \text{or} \ 6 \end{pmatrix}}\,.
			$$
			
			Examining the above conditions, we get
			$
			\hat\S({\tt b}) =
			\left\{
			\left( \begin{smallmatrix} 1 \\ 2 \end{smallmatrix} \right),
			\left( \begin{smallmatrix} 0 \\ 3 \end{smallmatrix} \right) \right\}.
			$
			Indeed, ${\tt rrs}$ and ${\tt rrr}$ are factors of $\dd_\uu({\tt b})$, thus $\left( \begin{smallmatrix} 1 \\ 2 \end{smallmatrix} \right), \left( \begin{smallmatrix} 0 \\ 3 \end{smallmatrix} \right)$ are Parikh vectors of some factors in $\dd_\uu({\tt b})$.
			The shortest return word to $w \in \LL(\vv)$ with $\pi(w) = {\tt b}$ satisfies
			$$
			|v| =
			\min \left\{  k + 2\ell : \begin{pmatrix} \ell \\ k \end{pmatrix} \in \hat\S({\tt b}) \right\}
			= 3.
			$$
			
			\item ${\tt b^2}$:
			The Parikh vectors $\Parickh({r}), \ \Parickh({s})$ and the slope $\theta'$ were determined in Example~\ref{ex:ParikhRetwords}.
			Moreover, $\gap{\yy}{|{\tt b^2}|_{\tt a}} = \gap{\yy}{0} = \{ 1 \}$ and $\gap{\yy'}{|{\tt b^2}|_{\tt b}} = \gap{\yy'}{2} = \{ 3,6 \}$.
			If $\left( \begin{smallmatrix} \ell \\ k \end{smallmatrix} \right) \in \hat\S(\tt b^2)$, then
			$$
			(k-1) \theta' - 1 < \ell < (k+1) \theta' + 1
			\ \ \text{and} \
			k, \ell \in \N;
			$$
			$$
			k \Parickh({r}) + \ell \Parickh({s}) =
			k \begin{pmatrix} 0 \\ 1 \end{pmatrix} + \ell \begin{pmatrix} 1 \\ 2 \end{pmatrix}
			\equiv
			\begin{pmatrix} 0 \\ 0 \end{pmatrix}
			\pmod
			{\begin{pmatrix} 1 \\ 3 \ \text{or} \ 6\end{pmatrix}}\,.
			$$
			
			Examining the above conditions, we get
			$
			\hat\S({\tt b^2}) =
			\left\{ \left(\begin{smallmatrix} 1 \\ 1 \end{smallmatrix}\right) \right\}.
			$
			The shortest return word to $w \in \LL(\vv)$ with $\pi(w) = {\tt b^2}$ satisfies
			$$
			|v| =
			\min \left\{ k + 3\ell : \begin{pmatrix} \ell \\ k \end{pmatrix} \in \hat\S({\tt b^2}) \right\}
			= 4.
			$$
			
			\item ${\tt b^3ab^3}$:
			The Parikh vectors $\Parickh({r}), \ \Parickh({s})$ and the slope $\theta'$ were determined in Example~\ref{ex:ParikhRetwords}.
			Moreover, $\gap{\yy}{|{\tt b^3ab^3}|_{\tt a}} = \gap{\yy}{1} = \{ 2 \}$ and $\gap{\yy'}{|{\tt b^3ab^3}|_{\tt b}} = \gap{\yy'}{6} = \{ 6 \}$.
			If $\left( \begin{smallmatrix} \ell \\ k \end{smallmatrix} \right) \in \hat\S({\tt b^3ab^3})$, then
			$$
			(k-1) \theta' - 1 < \ell < (k+1) \theta' + 1
			\ \ \text{and} \
			k, \ell \in \N;
			$$
			$$
			k \Parickh({r}) + \ell \Parickh({s}) =
			k \begin{pmatrix} 1 \\ 3 \end{pmatrix} + \ell \begin{pmatrix} 1 \\ 4 \end{pmatrix}
			\equiv
			\begin{pmatrix} 0 \\ 0 \end{pmatrix}
			\pmod
			{\begin{pmatrix} 2 \\6\end{pmatrix}}\,.
			$$
			
			Examining the above three conditions, we get
			$
			\hat\S({\tt b^3ab^3}) =
			\left\{
			\left(\begin{smallmatrix} 0 \\ 2 \end{smallmatrix}\right)
			\right\}.
			$
			The shortest return word to $w \in \LL(\vv)$ with $\pi(w) = {\tt b^3ab^3}$ satisfies
			$$
			|v| =
			\min \left\{ 4k + 5\ell : \begin{pmatrix} \ell \\ k \end{pmatrix} \in \hat\S({\tt b^3ab^3}) \right\}
			= 8.
			$$
			
			\item ${\tt b^4}$:
			The shortest bispecial factor containing ${\tt b^4}$ is ${\tt b^3 a b^3 a b^4 a b^3 a b^3}$.
			By Remark~\ref{rem:retwords_extension_to_BS}, the derived sequences satisfy $\dd_\uu({\tt b^4}) = \dd_\uu({\tt b^3 a b^3 a b^4 a b^3 a b^3})$ and the Parikh vectors of the corresponding return words coincide.
			It is the $6^{\text th}$ bispecial factor of $\uu$, hence associated with $(N,m) = (2,1)$ (see Remark~\ref{rem:N,m}).
			The slope of $\dd_\uu({\tt b^4})$ is, according to Proposition~\ref{ParikhRSB1}, equal to $\theta' = [0, 2,\overline{1, 3}]$.
			One can easily determine that $\theta' = \frac{2}{\sqrt{21}+1} \doteq 0.358$.
			By Proposition~\ref{ParikhRSB1} and using Table~\ref{tab:Lubka_parameters}, the return words $r$ and $s$ to ${\tt b^4}$ have the Parikh vectors
			$$
			\Parickh({r}) =
			\begin{pmatrix} p_2 \\ q_2 \end{pmatrix} =
			\begin{pmatrix} 2 \\ 7 \end{pmatrix}
			\quad \mbox{and} \quad
			\Parickh(s) =
			\begin{pmatrix} p_2+p_1 \\ q_2+q_1\end{pmatrix} =
			\begin{pmatrix} 3 \\ 10 \end{pmatrix}.
			$$
			Moreover, $\gap{\yy}{|{\tt b^4}|_{\tt a}} = \gap{\yy}{0} = \{ 1 \}$ and $\gap{\yy'}{|{\tt b^4}|_{\tt b}} = \gap{\yy'}{4} = \{ 6 \}$.
			If $\left( \begin{smallmatrix} \ell \\ k \end{smallmatrix} \right) \in \hat\S(\tt b^4)$, then
			$$
			(k-1) \theta' - 1 <
			\ell <
			(k+1) \theta' + 1
			\ \ \text{and} \
			k, \ell \in \N;
			$$
			$$
			k \Parickh({r}) + \ell \Parickh({s}) =
			k \begin{pmatrix} 2 \\ 7 \end{pmatrix} + \ell \begin{pmatrix} 3 \\ 10 \end{pmatrix}
			\equiv
			\begin{pmatrix} 0 \\ 0 \end{pmatrix}
			\pmod
			{\begin{pmatrix} 1 \\ 6 \end{pmatrix}}\,.
			$$
			
			Examining the above conditions, we get
			$\hat\S({\tt b^4}) =
			\left\{
			\left( \begin{smallmatrix} 1\\ 2 \end{smallmatrix} \right)
			\right\}.$
			
			The shortest return word to $w \in \LL(\vv)$ with $\pi(w) = {\tt b^4}$ satisfies
			$$
			|v| =
			\min \left\{ 9k + 13\ell : \begin{pmatrix} \ell \\ k \end{pmatrix} \in \hat\S({\tt b^4}) \right\}
			= 31.
			$$
			
			\item ${\tt a}$:
			We have $r = {\tt ab^3}$ and $s = {\tt ab^4}$.
			Let us write down a short prefix of $\dd_\uu({\tt a}) = {\tt rsrsrsrr} \cdots$.
			Moreover, $\gap{\yy}{|{\tt a}|_{\tt a}} = \gap{\yy}{1} = \{ 2\}$ and $\gap{\yy'}{|{\tt a}|_{\tt b}} = \gap{\yy'}{0} = \{1\}$.
			If $\left( \begin{smallmatrix} \ell \\ k \end{smallmatrix} \right) \in \hat\S(\tt a)$, then
			$$
			\left( \begin{smallmatrix} \ell \\ k \end{smallmatrix} \right) \ \text{is a Parikh vector of a factor in} \ \dd_\uu({\tt a});
			$$
			$$
			k \Parickh(r) + \ell \Parickh(s) =
			k \begin{pmatrix} 1 \\ 3 \end{pmatrix} + \ell \begin{pmatrix} 1 \\ 4 \end{pmatrix}
			\equiv
			\begin{pmatrix} 0 \\ 0 \end{pmatrix}
			\pmod
			{\begin{pmatrix} 2 \\ 1\end{pmatrix}}\,.
			$$
			Examining the above conditions, we get
			$
			\hat\S({\tt a}) =
			\left\{
			\left(\begin{smallmatrix} 0 \\ 2 \end{smallmatrix}\right),
			\left(\begin{smallmatrix} 1 \\ 1 \end{smallmatrix}\right)
			\right\}.
			$
			Indeed, ${\tt rr}$ and ${\tt rs}$ are factors of $\dd_\uu({\tt a})$, thus $\left( \begin{smallmatrix} 0 \\ 2 \end{smallmatrix} \right), \left( \begin{smallmatrix} 1 \\ 1 \end{smallmatrix} \right)$ are Parikh vectors of some factors in $\dd_\uu({\tt a})$.
			The shortest return word to $w \in \LL(\vv)$ with $\pi(w) = {\tt a}$ satisfies
			$$
			|v| =
			\min \left\{  4k + 5\ell  : \begin{pmatrix} \ell \\ k \end{pmatrix} \in \hat\S({\tt a}) \right\}
			= 8.
			$$
			
			\item ${\tt ab}$:
			The shortest bispecial factor containing ${\tt ab}$ is ${\tt b^3 a b^3}$.
			By Remark~\ref{rem:retwords_extension_to_BS}, the derived sequences satisfy $\dd_\uu({\tt ab}) = \dd_\uu({\tt b^3ab^3})$ and the Parikh vectors of the corresponding return words coincide.
			The only new parameters we have to determine in order to calculate $\hat\S({\tt ab})$ are the gaps: $\gap{\yy}{|{\tt ab}|_{\tt a}} = \gap{\yy}{1} = \{ 2 \}$ and $\gap{\yy'}{|{\tt ab}|_{\tt b}} = \gap{\yy'}{1} = \{ 3,6 \}$.
			If $\left( \begin{smallmatrix} \ell \\ k \end{smallmatrix} \right) \in \hat\S({\tt ab})$, then
			$$
			\begin{pmatrix} \ell \\ k \end{pmatrix} \ \text{is a Parikh vector of a factor in} \ \dd_\uu({\tt b^3ab^3});
			$$
			$$
			k \Parickh(r) + \ell \Parickh(s) =
			k \begin{pmatrix} 1 \\ 3 \end{pmatrix} + \ell \begin{pmatrix} 1 \\ 4 \end{pmatrix}
			\equiv
			\begin{pmatrix} 0 \\ 0 \end{pmatrix}
			\pmod
			{\begin{pmatrix} 2 \\ 3 \ \text{or} \ 6\end{pmatrix}}\,.
			$$
			
			Examining the above conditions, we get
			$
			\hat\S({\tt ab}) =
			\left\{
			\left( \begin{smallmatrix} 0 \\ 2 \end{smallmatrix} \right)
			\right\}.
			$
			The shortest return word to $w \in \LL(\vv)$ with $\pi(w) = {\tt ab}$ satisfies
			$$
			|v| =
			\min \left\{ 4k + 5\ell : \begin{pmatrix} \ell \\ k \end{pmatrix} \in \hat\S({\tt ab}) \right\}
			= 8.
			$$
			
			\item ${\tt ab^2}$:
			The shortest bispecial factor containing ${\tt ab^2}$ is ${\tt b^3 a b^3}$.
			Thus proceeding exactly as for ${\tt ab}$ we deduce that the shortest return word to $w \in \LL(\vv)$ with $\pi(w) = {\tt ab^2}$ satisfies
			$$
			|v| =
			\min \left\{ 4k + 5\ell : \begin{pmatrix} \ell \\ k \end{pmatrix} \in \hat\S({\tt ab^2}) \right\}
			= 8.
			$$
			
			\item ${\tt ba}, \ {\tt bab}, \ {\tt b^2a}$:
			Using similar arguments as for $\hat\S({\tt ab})$, we get $\hat\S({\tt ba}) = \hat\S({\tt ab})$ and $\hat\S({\tt bab}) = \hat\S({\tt b^2a}) = \hat\S({\tt ab^2})$.
			The lengths of the shortest return words to factors in $\vv$ with projections ${\tt ba}$ and ${\tt ab}$ (resp., ${\tt bab, b^2a}$ and ${\tt ab^2}$) are the same by Theorem~\ref{prop:ShortestReturn}.
		\end{itemize}
		
		Finally, we have
		$$
		\begin{array}{rcl}
			E^{\text{short}}(\vv) & = & 1 + \max \left\{ {|w|}/{|v|}: w \in \LL^{\text{short}} \text{ and } v \in \R_{\vv}(w) \right \} \\
			& = & 1 + \max\left\{ \cfrac{1}{8}\,, \ \cfrac{4}{31}\,, \ \cfrac{3}{8}\,, \ \cfrac{1}{4}\,, \ \cfrac{1}{3}\,, \ \cfrac{1}{2}\,, \ \cfrac{7}{8} \right\} = 1+\cfrac{7}{8}\,.
		\end{array}
		$$
		
		\vspace{0.3cm}
		\noindent Second, we will describe $E(i,m)$ for $(i,m)$ determined in Example~\ref{ex:Lubka_H}.
		Let us recall all needed ingredients:
		$L_0 = \frac{\sqrt{21}-3}{2}, \
		L_1 = \frac{\sqrt{21}-3}{6}$ and
		$\lambda = \frac{5-\sqrt{21}}{2}$
		(see Example~\ref{ex:Lubka_L_i}).
		The values of $Q_N$ are given in Table~\ref{tab:Lubka_parameters}.
		Let us apply Proposition~\ref{pro:onlySmall} in order to determine, which values $I(h+i+NM, m)$ influence $E(i,m)$ besides the value $E^*(i,m)$.
		
		\begin{enumerate}
			\item $i = 0$:
			As  
			$|\lambda|^{N_0 H/M} |Q_{h-1}-L_0 Q_h| = |\lambda|^{3N_0} |Q_{h-1}-L_0 Q_h| \leq 2 L_0 $
			for $N_0 = 1$, we have
			$E(0,m) = \max\{ E^*(0,m), I(h,m)\}$
			for $0 \leq m < z_0 = 3$.
			Since $h = 2$, we need to treat separately $I(2,0), I(2,1), I(2,2)$.
			
			\item $i = 1$:
			As
			$|\lambda|^{3N_0} |Q_{h}-L_1 Q_{h+1}| \leq 2 L_1$
			holds for $N_0 = 1$ we have
			$E(1,m) = \max\{ E^*(1,m), I(h+1,m) \}$
			for $0 \leq m < z_1 = 1$.
			Thus we have to treat separately $I(3,0)$.
			
			\item $i \in \{2, 4\}$:
			Since
			$|\lambda|^{3N_0} |Q_{h+i-1}-L_0 Q_{h+i}| \leq 2 L_i = 2L_0$
			holds for $N_0 = 0$, we have
			$E(i,m) = E^*(i,m)$
			for each admissible $m$.
			
			\item $i \in  \{3, 5\}$:
			Since
			$|\lambda|^{3N_0} |Q_{h+i-1}-L_1 Q_{h+i}| \leq 2 L_i = 2 L_1$
			holds for $N_0 = 0$, we have
			$E(i,m) = E^*(i,m)$
			for each admissible $m$.
		\end{enumerate}
		
		To find $I(h+i,m)$ for $(i,m) \in \{ (0,0), (0,1), (0,2), (1,0) \}$, we use Formula~\eqref{modified}.
		The sets $\hat\S(i,m)$ required by this formula and $E^*(i,m)$ have been determined in Example~\ref{ex:Lubka_AsymptoticCR}. 
		Thus we have all we need to compute $I(h+i, m)$ and to compare it with $E^*(i,m)$.
		
		\begin{itemize}
			\item $I(2, 0) =
			1 + \max \left\{ \cfrac{Q_{2} + Q_{1}-2}{k Q_{2} + \ell Q_{1} } : \begin{pmatrix} \ell \\ k \end{pmatrix} \in \hat\S(0,0) \right\}
			= 1 + \frac{11}{62} < E^*(0,0).$
			\item $I(2, 1) =
			1 + \max \left\{ \cfrac{2 Q_{2} + Q_{1}-2}{(k + \ell )Q_{2} + \ell Q_{1} } : \begin{pmatrix} \ell \\ k \end{pmatrix} \in \hat\S(0,1) \right\}
			= 1 + \frac{10}{31} < E^*(0,1).$
			\item $I(2,2) =
			1 + \max \left\{ \cfrac{3 Q_{2} + Q_{1}-2}{(k + 2\ell )Q_{2} + \ell Q_{1} } : \begin{pmatrix} \ell \\ k \end{pmatrix} \in \hat\S(0,2) \right\}
			= 1 + \frac{29}{62} < E^*(0,2).$
			\item $I(3,0) =
			1 + \max \left\{ \cfrac{Q_{3} + Q_{2}-2}{k Q_{3} + \ell Q_{2} } : \begin{pmatrix} \ell \\ k \end{pmatrix} \in \hat\S(1,0) \right\}
			= 1 + \frac{19}{31} < E^*(1,0).$
		\end{itemize}
		
		To summarise, we have shown that $\max E(i,m) = \max E^*(i,m) = E^*(\vv)$.
		
		\vspace{0.3cm}
		\noindent To conclude,
		$$
		E(\vv) =
		\max\{ E^{\text{short}}(\vv), E^*(\vv) \} =
		\max \left\{ 1 + \frac{7}{8}, 1 + \frac{\sqrt{21}+3}{12} \right \} =
		1 + \frac{7}{8}\,.
		$$
	\end{example}

	\section{Balanced sequences with minimal critical exponent}
	\label{sec:balanced_minimal}
	
	Rampersad, Shallit and Vandomme~\cite{RaShVa} 
	focused on  balanced sequences with the least critical exponent. 
	For every $d \in \{ 3, 4, \ldots, 10 \}$ they defined a balanced sequence $\xx_d$ and conjectured that such sequence has the least critical exponent among all balanced sequences over a $d$-letter alphabet.
	They also proved the conjecture for $d = 3$ and $d = 4$.~\footnote{More precisely, the minimality in the case $d = 4$ was proved by Peltom\"{a}ki in a private communication to Rampersad.}
	Exploiting computer assistance, they found for $5 \leq d \leq 10$ that $\frac{d-2}{d-3}$ would be the least possible critical exponent for balanced sequences over a $d$-letter alphabet. 
	Later, Baranwal and Shallit~\cite{BaShBalanced,BaranThesis}
	confirmed that $E({\bf x}_d) = \frac{d-2}{d-3}$ for alphabets of size $5$ to $8$.
	We used the algorithm described in the previous chapters and implemented by our student Daniela Opo\v censk\'a to show that $E({\bf x}_d) = \frac{d-2}{d-3} $ also for $d = 9$ and $d = 10$.
	A detailed computation for the case $d = 9$ can be found in Appendix.
	
	The above mentioned results on the least critical exponent are summarised in the table, which is taken from~\cite{RaShVa} (we use the slope $\theta$ instead of the parameter $\alpha=\frac{1}{1+\theta}$ used in the original table).
	We erased the question marks accompanying the values $E(\xx_{9})$ and $E(\xx_{10})$ in the original table and we also added to the table a column containing the asymptotic critical exponent $E^*(\xx_{d})$.
	
	\begin{table}[htb!]
		\centering
		\setlength{\tabcolsep}{5pt}
		\renewcommand{\arraystretch}{1.3}
		\resizebox{\textwidth}{!} {
			\begin{tabular}{|c|l|l|l|l|l|}
				\hline
				$d$ & $\theta$ & $\yy$ & $\yy'$ & $\CR(\vv)$ & $\CR^{*}(\vv)$ \\
				\hline
				\hline
				3 & $[0, 1, \overline{2}]$ & ${\tt 0}^{\omega}$ & $({\tt 12})^{\omega}$ & $2 + \frac{1}{\sqrt{2}}$ & $2 + \frac{1}{\sqrt{2}}$ \\
				\hline
				4 & $[0, \overline{1}]$ & $({\tt 01})^{\omega}$ & $({\tt 23})^{\omega}$ & $1 + \frac{1 + \sqrt{5}}{4}$ & $1 + \frac{1 + \sqrt{5}}{4}$ \\
				\hline
				5 & $[0, 1, \overline{2}]$ & $({\tt 01})^{\omega}$ & $({\tt 2324})^{\omega}$ & $\frac{3}{2}$ & $\frac{3}{2}$ \\
				\hline
				6 & $[0, 2, 1, 1, \overline{1, 1, 1, 2}]$& ${\tt 0}^{\omega}$ & $({\tt 123415321435})^{\omega}$ & $\frac{4}{3}$ & $\frac{4}{3}$ \\
				\hline
				7 & $[0, 1, 3,\overline{1, 2, 1}]$ & $({\tt 01})^{\omega}$ & $({\tt 234526432546})^{\omega}$ & $\frac{5}{4}$ & $\frac{5}{4}$ \\
				\hline
				8 & $[0, 3, 1, \overline{2}]$ & $({\tt 01})^{\omega} $ & $({\tt 234526732546237526432576})^{\omega}$ & $\frac{6}{5} = 1.2$ & $\frac{12 + 3\sqrt{2}}{14} \doteq 1.16$ \\
				\hline
				9 & $[0, 2, 3,\overline{2}]$ & $({\tt 01})^{\omega} $ & $({\tt 234567284365274863254768})^{\omega}$ & $\frac{7}{6} \doteq 1.167$ & $1 + \frac{2 \sqrt{2} - 1}{14} \doteq 1.13$ \\
				\hline
				10 & $[0, 4, 2,\overline{3}]$ & $({\tt 01})^{\omega} $ & $({\tt 234567284963254768294365274869})^{\omega}$ & $\frac{8}{7} \doteq 1.14$ & $1 + \frac{\sqrt{13}}{26} \doteq 1.139$ \\
				\hline
			\end{tabular}
		}
		\vspace{0.3cm}
		\caption{The balanced sequences with the least critical exponent over alphabets of size $d$.}
		\label{NaradTable}
	\end{table}
	
	We see that $\CR^*({\bf x}_d) = \CR({\bf x}_d)$ for $d = 3, 4, 5, 6, 7$.
	Observing Table~\ref{NaradTable}, we can deduce that there exists a balanced sequence $\xx$ over an $8$-letter alphabet with $\CR^*(\xx) < \CR^*(\xx_8)$.
	The sequence $\xx$ uses the same pair $\yy$ and $\yy'$ as $\xx_8 $.
	The slope of $\xx$ is $\theta = [0, 2, 3, \overline{2}]$.
	Since $\xx$	and $\xx_9$ have the same slope and the same period of constant gap sequences, we have $\CR^*({\bf x}) = \CR^*({\bf x}_9) < \CR^*({\bf x}_8)$.
	
	In the course of the referee process, the conjecture by Rampersad, Shallit and Vandomme was disproved by introducing $d$-ary balanced sequences with the critical exponent equal to $\frac{d-1}{d-2}$ for $d=11$ and also for all even $d$'s larger than 10 (see~\cite{DvOpPeSh2022}). Moreover, it was shown ibidem that $\frac{d-1}{d-2}$ is a lower bound on the critical exponent for all $d$-ary balanced sequences. It thus remains open to prove or disprove that the value $\frac{d-1}{d-2}$ is the least critical exponent of $d$-ary balanced sequences for all odd $d$'s larger than 11.  
	
	In addition, a new method for determining the least asymptotic critical exponent of balanced sequences was introduced in~\cite{DvOpPe2022} and the least asymptotic critical exponent was computed for alphabets of size 3 to 10.
	It follows that the least critical exponent and the least asymptotic critical exponent of balanced sequences are equal for $d\in \{3,4,5\}$, but the asymptotic version is smaller for larger $d$.
	It remains as an open problem to find the least asymptotic critical exponent over larger alphabets.

	\section{Appendix}
	\label{sec:appendix}
	
	In this appendix we illustrate our method for computing the critical exponent on the balanced sequence $\xx_9$.
	The sequence was introduced in~\cite{RaShVa} as a candidate for the balanced sequence having the least critical exponent over a $9$-letter alphabet.
	It was shown ibidem that $E(\vv) \geq \frac{7}{6}$ for every balanced sequence over a $9$-letter alphabet.
	In the sequel we will show that $E(\xx_9) = \frac{7}{6}$ and we will thus confirm the minimality of the critical exponent for this sequence.
	
	Let us consider the following constant gap sequences
	$$
	\yy = ({\tt 01})^\omega\,
	\qquad \text{and} \qquad
	\yy' = ({\tt 234567284365274863254768})^{\omega}\,.
	$$
	
	The sequence $\yy'$ is a constant gap sequence because ${\rm gap}_{\yy'}({i}) = 6$ for $i \in \{ {\tt 2}, {\tt 4}, {\tt 6} \}$ and ${\rm gap}_{\yy'}({j}) = 8$ for $j \in \{ {\tt 3}, {\tt 5}, {\tt 7}, {\tt 8} \}$.
	Moreover, ${\rm gap}(\yy', n) = \{ 24 \}$ for all $n \ge 2$.
	The minimal period of $\yy'$ is ${\rm Per}(\yy') = 24$.
	
	The only bispecial factor in $\yy$ is the empty word, while the only bispecial factors in $\yy'$ are the empty word $\varepsilon$ and the letters.
	Consequently, $\beta(\yy) = 0$ and $\beta(\yy') = 1$.
	
	We define the sequence $\xx_{9} = \barva(\uu_9, \yy, \yy')$, where $\uu_9$ is the standard Sturmian sequence with slope $\theta_9 =  [0, 2, 3, \overline{2}]$.
	Here are the prefixes of the studied sequences:
	$$
	\begin{array}{rcl}
		\uu_9 &=&
		{\tt bbabbabbabbbabbabbabbbabbabbabbabbbabbabbabbbabbabb}\cdots, \\
		\xx_9 &=& {\tt 230451670284136052174806312504716820341560728143065}\cdots.
	\end{array}
	$$
	
	The sequence $\xx_9$ is balanced according to Theorem~\ref{Hubert}.
	Moreover, it is easy to check that $\uu_9 \in \{ {\tt b^2 a},  {\tt b^3a} \}^{\N}$.

	\subsection*{Asymptotic critical exponent of $\xx_9$}
	
	In Table~\ref{tab:x9_parameters} we write all needed parameters for the sequence $\uu_9$.
	
	\begin{table}[htb!]
		\centering
		\setlength{\tabcolsep}{5pt}
		\renewcommand{\arraystretch}{1.3}
		\resizebox{\textwidth}{!} {
			\begin{tabular}{|l|l|l|l|l|l|l|l|l|l|l|l|}
				\hline
				$N$ & $0$ & $1$ & $\textcolor{green}{2}$ & $3$ & $4$ & $5$ & $6$ & $7$ & $8$ & $9$ & $\textcolor{green}{10}$ \\
				\hline
				$a_N$ & $0$ & $2$ & $3$ & $2$ & $2$ & $2$ & $2$ & $2$ & $2$ & $2$ & $2$ \\
				\hline
				$p_{N}$ & $0$ & $1$ & $3$ & $7$ & $17$ & $41$ & $99$ & $239$ & $577$ & $1393$ & $3363$\\
				\hline
				$p_{N} \bmod {\rm Per(\yy)}$ & $0$ & $\textcolor{red}{1}$ & $\textcolor{blue}{1}$ & $1$ & $1$ & $1$ & $1$ & $1$ & $1$ & $\textcolor{red}{1}$ & $\textcolor{blue}{1}$\\
				\hline
				$q_{N}$ & $1$ & $2$ & $7$ & $16$ & $39$ & $94$ & $227$ & $548$ & $1323$ & $3194$ & $7711$ \\
				\hline
				$q_{N} \bmod {\rm Per(\yy')}$ & $1$ & $\textcolor{red}{2}$ & $\textcolor{blue}{7}$ & $16$ & $15$ & $22$ & $11$ & $20$ & $3$ & $\textcolor{red}{2}$ & $\textcolor{blue}{7}$\\
				\hline
				$Q_{N}$ & $1$ &  $3$ & $10$ & $23$ & $56$ & $135$ & $326$ & $787$ & $1900$ & $4587$  & $11074$\\
				\hline
			\end{tabular}
		}
		\vspace{0.3cm}
		\caption{The first values of $a_N, p_N, q_N, Q_N$ for $\uu_9$ with $\theta = [0, 2,3, \overline{2}]$.}
		\label{tab:x9_parameters}
	\end{table}
	
	By Remark~\ref{rem:preperiod}, we have $p_2 + p_{1} - 1 = 3 > 0 = \beta(\yy)$ and $q_2 + q_{1} - 1 = 8 > 1 = \beta(\yy')$, hence $h = 2$ is the minimal preperiod length of the continued fraction of $\theta$ from~\eqref{fixedTheta}.
	
	Let us find the number $H$ of equivalence classes from Definition~\ref{def:equiv}.
	Observing Table~\ref{tab:x9_parameters}, we have
	$$
	H = \min \left\{ i \in \N, i > 0 : 2 + i \sim 2 \right\} = 8.
	$$
	Indeed, we have
	$\textcolor{green}{2} \equiv \textcolor{green}{10} \pmod 1$,
	$\left( \begin{smallmatrix}{\textcolor{red}{p_{1}}} \\ {\textcolor{red}{q_{1}}}\end{smallmatrix} \right) \equiv \left( \begin{smallmatrix}{\textcolor{red}{p_{9}}} \\ {\textcolor{red}{q_{9}}} \end{smallmatrix} \right) \pmod {\left( \begin{smallmatrix} 2 \\ 24 \end{smallmatrix} \right)}$
	and
	$\left( \begin{smallmatrix}\textcolor{blue}{{p_{2}}} \\ \textcolor{blue}{q_{2}} \end{smallmatrix} \right) \equiv \left( \begin{smallmatrix} \textcolor{blue}{p_{10}} \\ \textcolor{blue}{q_{10}} \end{smallmatrix} \right) \pmod {\left( \begin{smallmatrix} 2 \\ 24\end{smallmatrix} \right)}$.
	By Definition~\ref{def:classes}, we have 16 subsets $C(i,m)$, where $0 \leq i \leq 7$ and $0 \leq m \leq 1$.
	
	Since the period length is $M = 1$, we have $L = L_i = \lim_{N \to \infty} \frac{Q_{N-1}}{Q_N} = \sqrt{2} - 1$ for all $i \in \{ 0, 1, \dots, 7 \}$.
	
	We have $E^*(\xx_9) = \max \{ E^*(i,m) : 0 \leq i \leq 7,\ 0 \leq m \leq 1 \}$, where
	$$
	E^*(i,m) = 1 + \max \left\{ \cfrac{1 + m + L}{k+ \ell m + \ell L }
	\ : \
	\begin{pmatrix} \ell \\ k \end{pmatrix} \in \hat\S(i,m) \right\} \,.
	$$
	
	The reader is invited to verify the following calculations that we obtained using our computer program.
	
	\begin{itemize}
		\item $\hat\S(0,0)$:
		We have $\theta_{0,0} = [0, \overline{2}] = \sqrt{2} - 1 = L$.
		By Corollary~\ref{coro:S(i,m)} and Remark~\ref{rem:tildeS}, if $\left(\begin{smallmatrix} \ell \\ k \end{smallmatrix}\right) \in \hat\S(0,0)$, then
		$$
		k \begin{pmatrix} p_h \\ q_h \end{pmatrix} + \ell \begin{pmatrix} p_{h-1} \\ q_{h-1} \end{pmatrix} =
		k \begin{pmatrix} 1 \\ 7 \end{pmatrix} + \ell \begin{pmatrix} 1 \\ 2 \end{pmatrix}
		\equiv
		\begin{pmatrix} 0 \\ 0 \end{pmatrix}
		\pmod
		{\begin{pmatrix} 2 \\ 24 \end{pmatrix}};
		$$
		$$
		(k-1) \theta_{0,0} - 1 \leq \ell \leq (k+1) \theta_{0,0} + 1.
		$$
		We get $\hat\S(0,0) = \left\{ \left( \begin{smallmatrix} 6 \\ 12 \end{smallmatrix} \right) \right\}$ and $E^*(0,0) = 1 + \frac{1+L}{12+6L} = 1 + \frac{2-\sqrt{2}}{6} \doteq 1.0976$.
		
		\item $\hat\S(0,1)$:
		We have $\theta_{0,1} = [0, 1, \overline{2}] = \frac{1}{\sqrt{2}}$.
		If $\left( \begin{smallmatrix} \ell \\ k \end{smallmatrix} \right) \in \hat\S(0,1)$, then
		$$
		k \begin{pmatrix} p_h \\ q_h \end{pmatrix} + \ell \begin{pmatrix} p_h + p_{h-1} \\ q_h + q_{h-1} \end{pmatrix} =
		k \begin{pmatrix} 1 \\ 7 \end{pmatrix} + \ell \begin{pmatrix} 0 \\ 9 \end{pmatrix}
		\equiv
		\begin{pmatrix} 0 \\ 0 \end{pmatrix}
		\pmod
		{\begin{pmatrix} 2 \\ 24 \end{pmatrix}};
		$$
		$$
		(k-1) \theta_{0,1} - 1 \leq \ell \leq (k+1) \theta_{0,1} + 1.
		$$
		We get $\hat\S(0,1) = \left\{ \left( \begin{smallmatrix} 16 \\ 24 \end{smallmatrix} \right) \right\}$ and $E^*(0,1) = 1 + \frac{2+L}{40+16L} = 1 + \frac{1+\sqrt{2}}{24+16\sqrt{2}} \doteq 1.0518$.
		
		\item $\hat\S(1,0)$:
		We have $\theta_{1,0} = [0, \overline{2}] = \theta_{0,0}$.
		If $\left( \begin{smallmatrix} \ell \\ k \end{smallmatrix} \right) \in \hat\S(1,0)$, then
		$$
		k \begin{pmatrix} p_{h+1} \\ q_{h+1} \end{pmatrix} + \ell \begin{pmatrix} p_h \\ q_h \end{pmatrix} =
		k \begin{pmatrix} 1 \\ 16 \end{pmatrix} + \ell \begin{pmatrix} 1 \\ 7 \end{pmatrix}
		\equiv
		\begin{pmatrix} 0 \\ 0 \end{pmatrix}
		\pmod
		{\begin{pmatrix} 2 \\ 24 \end{pmatrix}};
		$$
		$$
		(k-1) \theta_{1,0} - 1 \leq \ell \leq (k+1) \theta_{1,0} + 1.
		$$
		We get $\hat\S(1,0) = \left\{ \left( \begin{smallmatrix} 8 \\ 16 \end{smallmatrix} \right) \right\}$ and $E^*(1,0) = 1 + \frac{1+L}{16+8L} = 1 + \frac{2-\sqrt{2}}{8} \doteq 1.0732$.
		
		\item $\hat\S(1,1)$:
		We have $\theta_{1,1} = [0, 1, \overline{2}] = \theta_{0,1}$.
		If $\left( \begin{smallmatrix} \ell \\ k \end{smallmatrix} \right) \in \hat\S(1,1)$, then
		$$
		k \begin{pmatrix} p_{h+1} \\ q_{h+1} \end{pmatrix} + \ell \begin{pmatrix} p_{h+1} + p_h \\ q_{h+1} + q_h \end{pmatrix} =
		k \begin{pmatrix} 1 \\ 16 \end{pmatrix} + \ell \begin{pmatrix} 0 \\ 23 \end{pmatrix}
		\equiv
		\begin{pmatrix} 0 \\ 0 \end{pmatrix}
		\pmod
		{\begin{pmatrix} 2 \\ 24 \end{pmatrix}};
		$$
		$$
		(k-1) \theta_{1,1} - 1 \leq \ell \leq (k+1) \theta_{1,1} + 1.
		$$
		We get $\hat\S(1,1) = \left\{ \left( \begin{smallmatrix} 16 \\ 22 \end{smallmatrix} \right) \right\}$ and $E^*(1,1) = 1 + \frac{2+L}{38+16L} = 1 + \frac{\sqrt{2}+1}{22+16\sqrt{2}} \doteq 1.0541$.
		
		\item $\hat\S(2,0)$:
		We have $\theta_{2,0} = [0, \overline{2}] = \theta_{0,0}$.
		If $\left( \begin{smallmatrix} \ell \\ k \end{smallmatrix} \right) \in \hat\S(2,0)$, then
		$$
		k \begin{pmatrix} p_{h+2} \\ q_{h+2} \end{pmatrix} + \ell \begin{pmatrix} p_{h+1} \\ q_{h+1} \end{pmatrix} =
		k \begin{pmatrix} 1 \\ 15 \end{pmatrix} + \ell \begin{pmatrix} 1 \\ 16 \end{pmatrix}
		\equiv
		\begin{pmatrix} 0 \\ 0 \end{pmatrix}
		\pmod
		{\begin{pmatrix} 2 \\ 24 \end{pmatrix}};
		$$
		$$
		(k-1) \theta_{2,0} - 1 \leq \ell \leq (k+1) \theta_{2,0} + 1.
		$$
		We get $\hat\S(2,0) = \left\{ \left( \begin{smallmatrix} 6 \\ 16 \end{smallmatrix} \right) \right\}$ and $E^*(2,0) = 1 + \frac{1+L}{16+6L} = 1 + \frac{\sqrt{2}}{10+6\sqrt{2}} \doteq 1.0765$.
		
		\item $\hat\S(2,1)$:
		We have $\theta_{2,1} = [0, 1, \overline{2}] = \theta_{0,1}$.
		If $\left( \begin{smallmatrix} \ell \\ k \end{smallmatrix} \right) \in \hat\S(2,1)$, then
		$$
		k \begin{pmatrix} p_{h+2} \\ q_{h+2} \end{pmatrix} + \ell \begin{pmatrix} p_{h+2} + p_{h+1} \\ q_{h+2} + q_{h+1} \end{pmatrix} =
		k \begin{pmatrix} 1 \\ 15 \end{pmatrix} + \ell \begin{pmatrix} 0 \\ 7 \end{pmatrix}
		\equiv
		\begin{pmatrix} 0 \\ 0 \end{pmatrix}
		\pmod
		{\begin{pmatrix} 2 \\ 24 \end{pmatrix}};
		$$
		$$
		(k-1) \theta_{2,1} - 1 \leq \ell \leq (k+1) \theta_{2,1} + 1.
		$$
		We get $\hat\S(2,1) = \left\{ \left( \begin{smallmatrix} 6 \\ 10 \end{smallmatrix} \right) \right\}$ and $E^*(2,1) = 1 + \frac{2+L}{16+6L} = 1 + \frac{\sqrt{2}+1}{10+6\sqrt{2}} \doteq 1.1306$.
		
		\item $\hat\S(3,0)$:
		We have $\theta_{3,0} = [0, \overline{2}] = \theta_{0,0}$.
		If $\left( \begin{smallmatrix} \ell \\ k \end{smallmatrix} \right) \in \hat\S(3,0)$, then
		$$
		k \begin{pmatrix} p_{h+3} \\ q_{h+3} \end{pmatrix} + \ell \begin{pmatrix} p_{h+2} \\ q_{h+2} \end{pmatrix} =
		k \begin{pmatrix} 1 \\ 22 \end{pmatrix} + \ell \begin{pmatrix} 1 \\ 15 \end{pmatrix}
		\equiv
		\begin{pmatrix} 0 \\ 0 \end{pmatrix}
		\pmod
		{\begin{pmatrix} 2 \\ 24 \end{pmatrix}};
		$$
		$$
		(k-1) \theta_{3,0} - 1 \leq \ell \leq (k+1) \theta_{3,0} + 1.
		$$
		We get $\hat\S(3,0) = \left\{ \left( \begin{smallmatrix} 12 \\ 30 \end{smallmatrix} \right) \right\}$
		and $E^*(3,0) = 1 + \frac{1+L}{30+12L} = 1 + \frac{\sqrt{2}}{18+12\sqrt{2}} \doteq 1.0404$.
		
		\item $\hat\S(3,1)$:
		We have $\theta_{3,1} = [0, 1, \overline{2}] = \theta_{0,1}$.
		If $\left( \begin{smallmatrix} \ell \\ k \end{smallmatrix} \right) \in \hat\S(3,1)$, then
		$$
		k \begin{pmatrix} p_{h+3} \\ q_{h+3} \end{pmatrix} + \ell \begin{pmatrix} p_{h+3} + p_{h+2} \\ q_{h+3} + q_{h+2} \end{pmatrix} =
		k \begin{pmatrix} 1 \\ 22 \end{pmatrix} + \ell \begin{pmatrix} 0\\ 13 \end{pmatrix}
		\equiv
		\begin{pmatrix} 0 \\ 0 \end{pmatrix}
		\pmod
		{\begin{pmatrix} 2 \\ 24 \end{pmatrix}};
		$$
		$$
		(k-1) \theta_{3,1} - 1 \leq \ell \leq (k+1) \theta_{3,1} + 1.
		$$
		We get $\hat\S(3,1) = \left\{ \left( \begin{smallmatrix} 12 \\ 18 \end{smallmatrix} \right) \right\}$ and $E^*(3,1) = 1 + \frac{2+L}{30+12L} = 1 + \frac{\sqrt{2}+1}{18+12\sqrt{2}} \doteq 1.069$.
	\end{itemize}
	
	Note that the values of $E^*(i,m)$ periodically repeat, i.e., $E^*(i+4,m) = E^*(i,m)$ for $0 \leq i \leq 3$ and $0 \leq m \leq 1$.
	We conclude that $E^*(\xx_9) = 1 + \frac{\sqrt{2}+1}{10+6\sqrt{2}} \doteq 1.1306$.

	\subsection*{Critical exponent of $\xx_9$}
	
	First, let us inspect $E^{\text{short}}(\xx_9)$.
	We know that $\beta(\yy) = 0$ and $\beta(\yy') = 1$.
	Thus, we have the following set of short factors in $\xx_9$:
	$$
	\LL^{\text{short}} = \left \{w \in \LL(\xx_9) : \pi(w) \in \{\tt b, b^2, b^2ab^2, b^2ab^2ab^2, b^3, a, ab, ba \} \right\}\,.
	$$
	The first four projections ${\tt b, b^2, b^2ab^2, b^2ab^2ab^2}$ listed in the specification of $\LL^{\text{short}}$ are bispecial factors in $\uu_9$ with $(N,m)$ satisfying $N < h = 2$.
	The other projections are not bispecial in $\uu_9$.
	By Theorem~\ref{prop:ShortestReturn}, we have to examine $\S(\pi(w))$ for all $w \in \LL^{\text{short}}$.
	In fact, we examine $\hat\S(\pi(w))$ instead (see Remark~\ref{rem:tildeS}).
	
	We will show that $\frac{|w|}{|v|} \le \frac{1}{6}$ for $w \in \LL^{\text{short}}$ and $v \in \R_{\xx_9}(w)$.
	Let us distinguish the following cases according to the projection $\pi(w)$.
	
	\begin{itemize}
		
		\item ${\tt b}$:
		Checking the prefix of $\uu_9$, we can see that $r = {\tt b}$ and $s = {\tt ba}$.
		Let us write down a short prefix of $\dd_\uu({\tt b}) =  {\tt rsrsrsrrsrsr} \cdots$.
		Moreover, $\gap{\yy}{|{\tt b}|_{\tt a}} = \gap{\yy}{0} = \{ 1 \}$ and $\gap{\yy'}{|{\tt b}|_{\tt b}} = \gap{\yy'}{1} = \{ 6,8 \}$.
		If $\left( \begin{smallmatrix} \ell \\ k \end{smallmatrix} \right) \in \hat\S(\tt b)$, then
		$$
		\begin{pmatrix} \ell \\ k \end{pmatrix} \ \text{is a Parikh vector of a factor in} \ \dd_\uu({\tt b});
		$$
		$$
		k \Parickh(r) + \ell \Parickh(s) =
		k \begin{pmatrix} 0 \\ 1 \end{pmatrix} + \ell \begin{pmatrix} 1 \\ 1 \end{pmatrix}
		\equiv
		\begin{pmatrix} 0 \\ 0 \end{pmatrix}
		\pmod
		{\begin{pmatrix} 1 \\ 6 \ \text{or} \ 8 \end{pmatrix}}\,.
		$$
		Examining the above conditions, we get
		$
		\hat\S({\tt b}) =
		\left\{
		\left( \begin{smallmatrix} 3 \\ 3 \end{smallmatrix} \right), \left( \begin{smallmatrix} 2 \\ 4 \end{smallmatrix} \right) \right\}.
		$
		Indeed, ${\tt srsrrs}$ and ${\tt rsrrsr}$ are factors of $\dd_\uu({\tt b})$, thus $\left( \begin{smallmatrix} 3 \\ 3 \end{smallmatrix} \right), \left( \begin{smallmatrix} 2 \\ 4 \end{smallmatrix} \right)$ are their Parikh vectors.
		The shortest return word to $w \in \LL(\xx_9)$ with $\pi(w) = {\tt b}$ satisfies
		$$
		|v| =
		\min \left\{  k + 2\ell : \begin{pmatrix} \ell \\ k \end{pmatrix} \in \hat\S({\tt b}) \right\}
		= 8.
		$$
		
		\item ${\tt b^2}$:
		The pair associated with ${\tt b^2}$ is $(N,m) = (1,0)$, thus we have $\Parickh(r) = \left( \begin{smallmatrix} p_1 \\ q_1 \end{smallmatrix} \right) = \left( \begin{smallmatrix} 1 \\ 2\end{smallmatrix} \right), \ \Parickh(s) = \left( \begin{smallmatrix} p_0 \\ q_0 \end{smallmatrix} \right) = \left( \begin{smallmatrix} 0 \\ 1 \end{smallmatrix} \right)$ and $\theta' = [0, 3, \overline{2}] = \frac{1}{2+\sqrt{2}}$. Moreover, $\gap{\yy}{|{\tt b^2}|_{\tt a}} = \gap{\yy}{0} = \{ 1 \}$ and $\gap{\yy'}{|{\tt b^2}|_{\tt b}} = \gap{\yy'}{2} = \{ 24\}$.
		If $\left( \begin{smallmatrix} \ell \\ k \end{smallmatrix} \right) \in \hat\S(\tt b^2)$, then
		$$
		(k-1) \theta' - 1 < \ell < (k+1) \theta' + 1
		\ \ \text{and} \
		k, \ell \in \N;
		$$
		$$
		k \Parickh(r) + \ell \Parickh(s) =
		k \begin{pmatrix} 1 \\ 2 \end{pmatrix} + \ell \begin{pmatrix} 0 \\ 1 \end{pmatrix}
		\equiv
		\begin{pmatrix} 0 \\ 0 \end{pmatrix}
		\pmod
		{\begin{pmatrix} 1 \\ 24 \end{pmatrix}}\,.
		$$
		Examining the above conditions, we get
		$
		\hat\S({\tt b^2}) =
		\left\{ \left(\begin{smallmatrix} 4 \\ 10 \end{smallmatrix}\right),  \left( \begin{smallmatrix} 2 \\ 11 \end{smallmatrix} \right) \right\}.
		$
		The shortest return word to $w \in \LL(\xx_9)$ with $\pi(w) = {\tt b^2}$ satisfies
		$$
		|v| =
		\min \left\{ 3k + \ell : \begin{pmatrix} \ell \\ k \end{pmatrix} \in \hat\S({\tt b^2}) \right\}
		= 34.
		$$
		
		\item ${\tt b^2ab^2}$:
		The pair associated with ${\tt b^2ab^2}$ is $(N,m) = (1,1)$, thus we have $\Parickh(r) = \left( \begin{smallmatrix} p_1 \\ q_1 \end{smallmatrix} \right) = \left( \begin{smallmatrix} 1 \\ 2\end{smallmatrix} \right), \ \Parickh(s) = \left( \begin{smallmatrix} p_1 + p_0 \\ q_1 + q_0 \end{smallmatrix} \right) = \left( \begin{smallmatrix} 1 \\ 3 \end{smallmatrix} \right)$ and $\theta' = [0, \overline{2}] = \sqrt{2}-1$.
		Moreover, $\gap{\yy}{|{\tt b^2ab^2}|_{\tt a}} = \gap{\yy}{1} = \{ 2\}$ and $\gap{\yy'}{|{\tt b^2ab^2}|_{\tt b}} = \gap{\yy'}{4} = \{ 24\}$.
		If $\left( \begin{smallmatrix} \ell \\ k \end{smallmatrix} \right) \in \hat\S({\tt b^2ab^2})$, then
		$$
		(k-1) \theta' - 1 < \ell < (k+1) \theta' + 1
		\ \ \text{and} \
		k, \ell \in \N;
		$$
		$$
		k \Parickh(r) + \ell \Parickh(s) =
		k \begin{pmatrix} 1 \\ 2 \end{pmatrix} + \ell \begin{pmatrix} 1 \\ 3 \end{pmatrix}
		\equiv
		\begin{pmatrix} 0 \\ 0 \end{pmatrix}
		\pmod
		{\begin{pmatrix} 2 \\24\end{pmatrix}}\,.
		$$
		Examining the above three conditions, we get
		$
		\hat\S({\tt b^2ab^2}) =
		\left\{
		\left( \begin{smallmatrix} 12 \\ 30 \end{smallmatrix} \right)
		\right\}.
		$
		The shortest return word to $w \in \LL(\xx_9)$ with $\pi(w) = {\tt b^2ab^2}$ satisfies
		$$
		|v| =
		\min \left\{ 3k + 4\ell : \begin{pmatrix} \ell \\ k \end{pmatrix} \in \hat\S({\tt b^2ab^2}) \right\}
		= 138.
		$$
		
		\item ${\tt b^2ab^2ab^2}$:
		The pair associated with ${\tt b^2ab^2ab^2}$ is $(N,m) = (1,2)$, thus we have $\Parickh(r) = \left( \begin{smallmatrix} p_1 \\ q_1 \end{smallmatrix} \right) = \left( \begin{smallmatrix} 1 \\ 2 \end{smallmatrix} \right), \ \Parickh(s) = \left( \begin{smallmatrix} 2p_1 + p_0 \\ 2q_1 + q_0 \end{smallmatrix} \right) = \left( \begin{smallmatrix} 2 \\ 5 \end{smallmatrix} \right)$ and $\theta' = [0, 1, \overline{2}] = \frac{1}{\sqrt{2}}$. Moreover, we have $\gap{\yy}{|{\tt b^2ab^2ab^2}|_{\tt a}} = \gap{\yy}{2} = \{ 2\}$ and $\gap{\yy'}{|{\tt b^2ab^2ab^2}|_{\tt b}} = \gap{\yy'}{6} = \{ 24\}$.
		If $\left( \begin{smallmatrix} \ell \\ k \end{smallmatrix} \right) \in \hat\S({\tt b^2ab^2ab^2})$, then
		$$
		(k-1) \theta' - 1 < \ell < (k+1) \theta' + 1
		\ \ \text{and} \
		k, \ell \in \N;
		$$
		$$
		k \Parickh(r) + \ell \Parickh(s) =
		k \begin{pmatrix} 1 \\ 2 \end{pmatrix} + \ell \begin{pmatrix} 2 \\ 5 \end{pmatrix}
		\equiv
		\begin{pmatrix} 0 \\ 0 \end{pmatrix}
		\pmod
		{\begin{pmatrix} 2 \\24\end{pmatrix}}\,.
		$$
		Examining the above three conditions, we get
		$
		\hat\S({\tt b^2ab^2ab^2}) =
		\left\{
		\left( \begin{smallmatrix} 12 \\ 18  \end{smallmatrix} \right)
		\right\}.
		$
		The shortest return word to $w \in \LL(\xx_9)$ with $\pi(w) = {\tt b^2ab^2ab^2}$ satisfies
		$$
		|v| =
		\min \left\{ 3k + 7\ell : \begin{pmatrix} \ell \\ k \end{pmatrix} \in \hat\S({\tt b^2ab^2ab^2}) \right\}
		= 138.
		$$
		
		\item ${\tt b^3}$:
		The shortest bispecial factor containing ${\tt b^3}$ is $b = {\tt b^2 a b^2 a b^2 a b^3 a b^2ab^2ab^2}$.
		By Remark~\ref{rem:retwords_extension_to_BS} the derived sequences satisfy $\dd_\uu({\tt b^3}) = \dd_\uu(b)$ and the Parikh vectors of the corresponding return words coincide.
		It is the $6^{\text th}$ bispecial, hence associated with $(N,m) = (2,1)$ (see Remark~\ref{rem:N,m}).
		The slope of $\dd_\uu({\tt b^3})$ is, according to Proposition~\ref{ParikhRSB1}, equal to $\theta' = [0, 1, \overline{2}] = \frac{1}{\sqrt{2}}$.
		The return words $r$ and $s$ to ${\tt b^3}$ have the Parikh vectors
		$$
		\Parickh(r) =
		\begin{pmatrix} p_2 \\ q_2 \end{pmatrix} =
		\begin{pmatrix} 3 \\ 7 \end{pmatrix}
		\quad \mbox{and} \quad
		\Parickh(s) =
		\begin{pmatrix} p_2+p_1 \\ q_2+q_1 \end{pmatrix} =
		\begin{pmatrix} 4 \\ 9 \end{pmatrix}.
		$$
		Moreover, $\gap{\yy}{|{\tt b^3}|_{\tt a}} = \gap{\yy}{0} = \{ 1 \}$ and $\gap{\yy'}{|{\tt b^3}|_{\tt b}} = \gap{\yy'}{3} = \{ 24 \}$.
		If $\left( \begin{smallmatrix} \ell \\ k \end{smallmatrix} \right) \in \hat\S(\tt b^3)$, then
		$$
		(k-1) \theta' - 1 <
		\ell <
		(k+1) \theta' + 1
		\ \ \text{and} \
		k, \ell \in \N;
		$$
		$$
		k \Parickh(r) + \ell \Parickh(s) =
		k \begin{pmatrix} 3 \\ 7 \end{pmatrix} + \ell \begin{pmatrix} 4 \\ 9 \end{pmatrix}
		\equiv
		\begin{pmatrix} 0 \\ 0 \end{pmatrix}
		\pmod
		{\begin{pmatrix} 1 \\ 24 \end{pmatrix}}\,.
		$$
		Examining the above conditions, we get
		$\hat\S({\tt b^3}) =
		\left\{
		\left( \begin{smallmatrix} 3 \\ 3  \end{smallmatrix} \right) \right\}.$
		
		The shortest return word to $w \in \LL(\xx_9)$ with $\pi(w) = {\tt b^3}$ satisfies
		$$
		|v| =
		\min \left\{ 10k + 13\ell : \begin{pmatrix} \ell \\ k \end{pmatrix} \in \hat\S({\tt b^3}) \right\}
		= 69.
		$$
		
		\item ${\tt a}$:
		We have $r = {\tt ab^2}$ and $s = {\tt ab^3}$.
		Let us write down a short prefix of $\dd_\uu({\tt a}) = {\tt rrsrrsrr} \cdots$.
		Moreover, $\gap{\yy}{|{\tt a}|_{\tt a}} = \gap{\yy}{1} = \{ 2\}$ and $\gap{\yy'}{|{\tt a}|_{\tt b}} = \gap{\yy'}{0} = \{1\}$.
		If $\left( \begin{smallmatrix} \ell \\ k \end{smallmatrix} \right) \in \hat\S(\tt a)$, then
		$$
		\begin{pmatrix} \ell \\ k \end{pmatrix} \ \text{is a Parikh vector of a factor in} \ \dd_\uu({\tt a});
		$$
		$$
		k \Parickh(r) + \ell \Parickh(s) =
		k \begin{pmatrix} 1 \\ 2 \end{pmatrix} + \ell \begin{pmatrix} 1 \\ 3 \end{pmatrix}
		\equiv
		\begin{pmatrix} 0 \\ 0 \end{pmatrix}
		\pmod
		{\begin{pmatrix} 2 \\ 1\end{pmatrix}}\,.
		$$
		Examining the above conditions, we get
		$
		\hat\S({\tt a}) =
		\left\{
		\left( \begin{smallmatrix} 0 \\ 2 \end{smallmatrix} \right),
		\left( \begin{smallmatrix} 1 \\ 1 \end{smallmatrix} \right)
		\right\}.
		$
		Indeed, ${\tt rr}$ and ${\tt rs}$ are factors of $\dd_\uu({\tt a})$, thus $\left( \begin{smallmatrix} 0 \\ 2 \end{smallmatrix} \right), \left( \begin{smallmatrix} 1 \\ 1 \end{smallmatrix} \right)$ are Parikh vectors of some factors in $\dd_\uu({\tt a})$.
		The shortest return word to $w \in \LL(\xx_9)$ with $\pi(w) = {\tt a}$ satisfies
		$$
		|v| =
		\min \left\{ 3k + 4\ell : \begin{pmatrix} \ell \\ k \end{pmatrix} \in \hat\S({\tt a}) \right\}
		= 6.
		$$
		
		\item ${\tt ab}$:
		The shortest bispecial factor containing ${\tt ab}$ is ${\tt b^2 a b^2}$.
		By Remark~\ref{rem:retwords_extension_to_BS} the derived sequences satisfy $\dd_\uu({\tt ab}) = \dd_\uu({\tt b^2ab^2})$ and the Parikh vectors of the corresponding return words coincide.
		The only new parameters we have to determine in order to calculate $\hat\S({\tt ab})$ are the gaps: $\gap{\yy}{|{\tt ab}|_{\tt a}} = \gap{\yy}{1} = \{ 2 \}$ and $\gap{\yy'}{|{\tt ab}|_{\tt b}} = \gap{\yy'}{1} = \{ 6, 8 \}$.
		If $\left( \begin{smallmatrix} \ell \\ k \end{smallmatrix} \right) \in \hat\S({\tt ab})$, then
		$$
		\begin{pmatrix} \ell \\ k \end{pmatrix} \ \text{is a Parikh vector of a factor in} \ \dd_\uu({\tt b^2ab^2});
		$$
		$$
		k \Parickh(r) + \ell \Parickh(s) =
		k \begin{pmatrix} 1 \\ 2 \end{pmatrix} + \ell \begin{pmatrix} 1 \\ 3 \end{pmatrix}
		\equiv
		\begin{pmatrix} 0 \\ 0 \end{pmatrix}
		\pmod
		{\begin{pmatrix} 2 \\ 6 \ \text{or} \ 8\end{pmatrix}}\,.
		$$
		Examining the above conditions, we get
		$
		\hat\S({\tt ab}) =
		\left\{
		\left( \begin{smallmatrix} 2 \\ 6 \end{smallmatrix} \right)
		\right\}.
		$
		The shortest return word to $w \in \LL(\xx_9)$ with $\pi(w) = {\tt ab}$ satisfies
		$$
		|v| =
		\min \left\{ 3k + 4\ell : \begin{pmatrix} \ell \\ k \end{pmatrix} \in \hat\S({\tt ab}) \right\}
		= 26.
		$$
		
		\item ${\tt ba}$:
		Using similar arguments as for $\hat\S({\tt ab})$, we get $\hat\S({\tt ab}) = \hat\S({\tt ba})$.
		The lengths of the shortest return words to factors in $\xx_9$ with projections ${\tt ba}$ and ${\tt ab}$ are the same by 
		Theorem~\ref{prop:ShortestReturn}.
	\end{itemize}
	
	Finally, we have
	$$
	\begin{array}{rcl}
		E^{\text{short}}(\xx_9) & = & 1 + \max \left\{ {|w|}/{|v|}: w \in \LL^{\text{short}} \text{ and } v \in \R_{\xx_9}(w) \right \} \\
		& = & 1 + \max\left\{ \cfrac{5}{138}\,,\ \cfrac{1}{23}\,,\ \cfrac{4}{69}\,,\ \cfrac{1}{17}\,,\ \cfrac{1}{13}\,,\ \cfrac{1}{8}\,,\ \cfrac{1}{6} \right\} = 1+\cfrac{1}{6}\,.
	\end{array}
	$$
	
	\vspace{0.3cm}
	
	\noindent Second, we will describe $E(i,m)$ for $0\leq i \leq 7$ and $0\leq m \leq 1$.
	Let us recall all needed ingredients:
	$L_i = L = \sqrt{2}-1$ and
	$\lambda = 1-\sqrt{2} = -L$ is the non-dominant eigenvalue of the matrix $A^{(0)} = \left( \begin{smallmatrix} 0 & 1 \\ 1 & 2 \end{smallmatrix} \right)$ from Corollary~\ref{cor:limits}.
	The values of $(Q_N)$ are given in Table~\ref{tab:x9_parameters}.
	Let us apply Proposition~\ref{pro:onlySmall} in order to determine which values $I(h+i+NM, m)$ influence $E(i,m)$ besides the value $E^*(i,m)$.
	
	\begin{enumerate}
		\item $i = 0$:
		Since $|\lambda|^{N_0 H/M} |Q_{h-1}-L Q_h| = |\lambda|^{8N_0} |Q_{h-1}-L Q_h| \leq 2 L$ holds for $N_0 = 1$, we have $E(0,m) = \max\{ E^*(0,m), I(h,m) \}$ for $0 \leq m \leq 1$.
		Thus we have to treat separately $I(2,0), I(2,1)$.
		
		\item $i \in \{1,\dots, 7\}$:
		Since $|\lambda|^{8N_0} |Q_{h+i-1}-L Q_{h+i}| \leq 2 L $ holds for $N_0 = 0$, we have $E(i,m) = E^*(i,m)$ for $0 \leq m \leq 1$.
	\end{enumerate}
	
	Combining Formul\ae~\eqref{modified} and \eqref{SetS(im)}, we get
	$$
	I(h, m)=
	1 + \max \left\{ \frac{(1 + m)Q_{h} + Q_{h-1}-2}{(k + \ell m)Q_{h} + \ell Q_{h-1}} : \begin{pmatrix} \ell \\ k \end{pmatrix} \in \hat\S(0,m) \right\}.
	$$
	When calculating $E^*(\xx_9)$ we have determined $\hat\S(i,m)$ and $E^*(i,m)$ for the considered values $(i,m) \in \{(0,0), (0,1)\}$.
	Thus we have all we need to compute $I(2, 0)$ and $I(2,1)$.
	
	\begin{itemize}
		\item $I(2, 0) =
		1 + \max \left\{ \cfrac{Q_{2} + Q_{1}-2}{k Q_{2} + \ell Q_{1}} : \begin{pmatrix} \ell \\ k \end{pmatrix} \in \hat\S(0,0) \right\}
		= 1 + \frac{11}{138} <E^*(0,0).$
		
		\item $I(2, 1) =
		1 + \max \left\{ \cfrac{2 Q_{2} + Q_{1}-2}{(k + \ell )Q_{2} + \ell Q_{1}} : \begin{pmatrix} \ell \\ k \end{pmatrix} \in \hat\S(0,1) \right\}
		= 1 + \frac{21}{448} < E^*(0,1).$
	\end{itemize}
	
	\vspace{0.3cm}
	
	\noindent To conclude,
	$$
	E(\xx_9) =
	\max\{ E^{\text{short}}(\xx_9), E^*(\xx_9) \} =
	\max \left\{ 1 + \frac{1}{6} \,, \ 1 + \frac{\sqrt{2}+1}{10+6\sqrt{2}} \right \} =
	\frac{7}{6} \,.
	$$

\end{document}